\documentclass[11pt]{article}
\usepackage[utf8]{inputenc}
\usepackage{amsmath,amsfonts,amssymb}
\usepackage{amsthm}
\usepackage{latexsym}
\usepackage[noadjust]{cite}
\usepackage{graphicx}
\usepackage{xcolor}
\usepackage{dirtytalk}
\usepackage{mathtools}
\usepackage[margin=1in]{geometry}
\usepackage{thm-restate}
\usepackage{enumitem}
\usepackage[linesnumbered,ruled]{algorithm2e}
\usepackage[colorlinks=true, allcolors=blue]{hyperref}
\usepackage[nameinlink,capitalise]{cleveref}
\hypersetup{
    citecolor={violet}
}

\newtheorem{theorem}{Theorem}[section]
\newtheorem{corollary}[theorem]{Corollary}
\newtheorem{lemma}[theorem]{Lemma}

\newtheorem{claim}[theorem]{Claim}
\theoremstyle{definition}
\newtheorem{definition}{Definition}[section]
\newtheorem{remark}[theorem]{Remark}
\newtheorem{fact}[theorem]{Fact}

\newcommand{\F}{\mathbb{F}}

\newcommand{\E}{\mathbb{E}}
\newcommand{\R}{\mathbb{R}}
\newcommand{\Z}{\mathbb{Z}}

\newcommand{\calS}{\mathcal{S}}

\newcommand{\zo}{\{0, 1\}}

\newcommand{\wt}{\mathrm{wt}}

\title{Fault-Tolerant Distance Oracles Below the $n \cdot f$ Barrier}
\date{\today}
\author{Sanjeev Khanna \thanks{Courant Institute, New York University. Supported in part by NSF award CCF-2402284 and AFOSR award FA9550-25-1-0107. Contact: \href{mailto:sanjeev.khanna@nyu.edu}{sanjeev.khanna@nyu.edu}.} \and Christian Konrad\thanks{School of Computer Science, University of Bristol, Bristol, UK. Contact: \href{mailto:christian.konrad@bristol.ac.uk}{christian.konrad@bristol.ac.uk}.} \and Aaron Putterman\thanks{Harvard University. Supported in part by the Simons Investigator awards of Madhu Sudan and Salil Vadhan, and AFOSR award FA9550-25-1-0112. Contact: \href{mailto:aputterman@g.harvard.edu}{aputterman@g.harvard.edu}.}}

\begin{document}
\pagenumbering{gobble}

\maketitle

\begin{abstract}
Fault-tolerant spanners are fundamental objects that preserve distances in graphs even under edge failures. A long line of work culminating in Bodwin–Dinitz–Robelle (SODA~2022) gives $(2k-1)$-stretch, $f$-fault-tolerant spanners with $
O\big(k^2 f^{\,\frac12-\frac{1}{2k}} n^{1+\frac1k} + k f n\big)
$
edges for any odd $k$. For any $k = \tilde{O}(1)$, this bound is essentially optimal for deterministic spanners in part due to a known folklore lower bound that \emph{any} $f$-fault-tolerant spanner requires $\Omega(nf)$ edges in the worst case.
For $f \ge n$, this $\Omega(nf)$ barrier means that any $f$-fault tolerant spanners are trivial in size. Crucially however, this folklore lower bound exploits that the spanner \emph{is itself a subgraph}. It does \emph{not} rule out distance-reporting data structures that may not be subgraphs. This leads to our central question: can one beat the $n \cdot f$ barrier with fault-tolerant distance oracles?

We give a strong affirmative answer to this question. As our first contribution, we construct $f$-fault-tolerant distance oracles
with stretch $O(\log(n)\log\log(n))$ that require only $\widetilde{O}(n\sqrt{f})$ bits of space; substantially below the spanner barrier of $n \cdot f$. Beyond this, in the regime $n \leq f \leq n^{3/2}$ we show that by using our new \emph{high-degree, low-diameter} decomposition, we can even achieve constant stretch. Specifically, we obtain stretch $7$ distance oracles in space $\widetilde{O}(n^{3/2}f^{1/3})$ bits.

Building on the principles used here, we obtain  two further consequences:
(a) randomized \emph{oblivious} sketches of size $\widetilde{O}(n\sqrt{f})$ that, after any $|F| \le f$ edge deletions by an \emph{oblivious} adversary, recover an $O(\log(n)\log\log(n))$-spanner of $G \setminus F$ with high probability, and
(b) insertion–deletion streaming algorithms that, under at most $f$ deletions overall, yield \emph{deterministic} $O(\log(n)\log\log(n))$-stretch distance oracles as well as \emph{randomized} $O(\log(n)\log\log(n))$-stretch spanners against an \emph{oblivious adversary} in $\widetilde{O}(n^{4/3} f^{1/3})$ space.

Conceptually, we exploit the robustness of high-degree graph structures with small diameter under bounded deletions to \emph{compress away} the parts of a graph that are information-theoretically expensive to store for deterministic spanners but not for deterministic oracles or sketches against oblivious adversaries. Our framework combines this structural robustness with ideas from deterministic and oblivious sparse recovery: dense regions act as inherently stable components requiring little explicit storage, while the small remaining boundary is summarized through compact deterministic sketches. When randomness is allowed, the same primitives yield oblivious sketches and streaming algorithms that materialize these certificates as edges. This unifying viewpoint drives all our results, achieving strong compression below the $n\cdot f$ barrier across static, oblivious, and streaming settings.

\end{abstract}

\pagebreak

\tableofcontents

\pagebreak  

\pagenumbering{arabic}

\section{Introduction}

\subsection{Background}\label{sec:background}

Spanners are a fundamental data structure in computer science, first introduced by Peleg and Ullman \cite{peleg1987optimal} and Peleg and Sch{\"a}ffer \cite{peleg1989graph} as a compact way to store approximate shortest path information in a graph. Given a graph $G = (V, E)$, a subgraph $H = (V, E')$ is said to be a $t$-spanner if, for every pair of vertices $(u,v) \in \binom{V}{2}$ 
\[
\mathrm{dist}_{H}(u,v) \leq t \cdot \mathrm{dist}_{G}(u,v),
\]
where we use $\mathrm{dist}_{H}(u,v)$ to mean the length of the shortest path between $u$ and $v$ in the graph $H$.
While constructing a $t$-spanner is straightforward (indeed, $G$ itself is a trivial $t$-spanner), the main algorithmic challenge lies in achieving \emph{sparsity}, that is, reducing the number of edges while approximately preserving distances. A sparse spanner naturally lends itself to faster computation of shortest paths on the spanner and also leads to more efficient storage. For these reasons, spanners have found a vast array of applications, for instance in routing \cite{peleg1989trade}, synchronizers \cite{awerbuch1990network}, broadcasting \cite{awerbuch1992efficient, peleg2000distributed}, and efficient sparsification algorithms \cite{ADKKP16}.
Motivated by these diverse applications, a long line of research has focused on understanding the tradeoff between stretch and sparsity, culminating in the seminal work of Alth{\"o}fer et al.~\cite{althofer1993sparse}, which established an essentially optimal relationship between the two.

However, a critical dimension missing in the aforementioned spanner constructions is that they are not \emph{fault-tolerant}. Indeed, in real-world network dynamics, one often must worry about failures in the underlying graph, motivating the question of whether a subgraph remains a spanner even after some edge failures in the original graph $G$. Formally, we say that a subgraph $H \subseteq G$ is an \textbf{$f$-fault tolerant $t$-spanner of $G$} if, for any possible subset of $f$ edge deletions $F$, it is the case that $H - F$ is still a $t$-spanner of $G - F$. A similar notion can be defined for vertex deletions which has inspired a comprehensive and parallel line of work (see, e.g., \cite{bodwin2018optimal, bodwin2019trivial, bodwin2021optimal, dinitz2020efficient, parter2022nearly}).

Such fault-tolerant spanners were first introduced by Levcopoulos, Narasimhan, and Smid \cite{levcopoulos1998efficient}, though their work addressed only the geometric setting. For general graphs, the study was initiated later by Chechik, Langberg, Peleg, and Roditty~\cite{chechik2009fault}, launching a sustained effort to understand the three-way tradeoff between \emph{stretch}, \emph{sparsity}, and \emph{fault-tolerance}.

In their original work, \cite{chechik2009fault} showed that for an $f$-fault tolerant spanner with stretch $2k-1$, one can achieve a sparsity bound of $O(f \cdot n^{1 + 1 / k})$ edges. This was subsequently improved in a sequence of works of \cite{bodwin2018optimal,  bodwin2019trivial, bodwin2021optimal, dinitz2020efficient, popova2026new} and finally in the work of Bodwin, Dinitz, and Robelle \cite{bodwin2022partially}, culminated in $f$-fault tolerant spanners with stretch $2k-1$ and sparsity

\begin{align}\label{eq:bodwinSparsity}
O(k^2 f^{1/2 - (1/2k)}n^{1 + 1/k} + kfn) 
\end{align}
many edges when $k$ is odd, and 
\[
O(k^2 f^{1/2}n^{1 + 1/k} + kfn)
\] 
when $k$ is even.

Prior to this, the work of Bodwin, Dinitz, Parter, and Williams \cite{bodwin2018optimal} had already established \emph{lower bounds} on the sparsity of any $f$-fault tolerant $2k-1$-spanner, showing that any such graph must retain $\Omega(f^{1/2 - 1/2k} \cdot n^{1 + 1/k})$ many edges. This lower bound essentially establishes \cite{bodwin2022partially} as an optimal fault-tolerant spanner result (at least in the case of $k$ odd).
At first glance, the extra additive term $O(nf)$ in~\cref{eq:bodwinSparsity} might appear to leave a gap, but it is in fact unavoidable. Any $f$-fault-tolerant spanner must, in the worst case, contain $\Omega(nf)$ edges. To see this, consider any $f$-regular graph $G$. We claim that any $f$-fault tolerant spanner $H$ of $G$ must in fact be \emph{equal} to $G$. Indeed, if, for some vertex $v$, $H$ stores $\leq f-1$ of the edges incident to $v$, then by setting the deletion set $F = \Gamma_{H}(v)$ (i.e., all adjacent edges to $v$ in the spanner), it is then the case that $v$ is a disconnected vertex in $H - F$, while $v$ still has at least one neighbor in $G - F$, thereby violating the spanner condition. Hence, any $f$-fault tolerant spanner of $G$ (for any finite stretch) must in fact be $G$, and we obtain our $\Omega(nf)$ lower bound.

Interestingly, in the same work, Bodwin, Dinitz, Parter, and Williams~\cite{bodwin2018optimal} also examined an analogous question for more general \emph{distance-reporting data structures}. Formally:
\begin{definition}\label{def:distanceOracle}
	We say that a data structure $\mathcal{S}$ is an \textbf{$f$-fault-tolerant $t$-distance oracle} if, when instantiated on any graph $G = (V, E)$, the data structure can be queried with a failure set $F \subseteq E$ satisfying $|F| \leq f$ and a vertex pair $a,b \in V$, returning a value $\widehat{\mathrm{dist}}_{G-F}(a,b)$ such that
	\[
	\mathrm{dist}_{G-F}(a,b) \leq \widehat{\mathrm{dist}}_{G-F}(a,b) \leq t \cdot \mathrm{dist}_{G-F}(a,b).
	\]
\end{definition}

In essence, an $f$-fault-tolerant $t$-distance oracle answers the same queries as an $f$-fault-tolerant $t$-spanner, reporting approximate distances after up to $f$ failures, but \emph{without} requiring the stored object to be a subgraph of $G$.
Beyond establishing the $f$-fault-tolerant $(2k-1)$-spanner lower bound of $\Omega(f^{1/2 - 1/2k} n^{1 + 1/k})$ edges, Bodwin et al.~\cite{bodwin2018optimal} also asked whether a comparable lower bound could extend to general $f$-fault-tolerant $(2k-1)$-\emph{distance oracles}. They indeed proved an $\Omega(f^{1/2 - 1/2k} n^{1 + 1/k})$ \emph{bit-complexity} lower bound, but only for a much stronger model of deletions, one that permits deletions $F$ to be completely arbitrary, not necessarily restricted to being a subset of $E$ \cite{bodwin2018optimalCorrection}.\footnote{Note that in this stronger deletion model, one can in fact easily show a lower bound of $\Omega(nf)$ bits, and so we do not use this model in our work. See~\cref{sec:edgeSubsetCondition} for details.} 

Nevertheless, this discussion of fault-tolerant distance oracles raises a tantalizing possibility. The folklore $\Omega(nf)$ lower bound for deterministic fault-tolerant spanners relies \emph{crucially} on the fact that the representation must itself be a graph. There is, at present, no corresponding barrier known for more general distance-reporting structures. Moreover, for nearly all parameter regimes in which $f$ is polynomially large, the proposed bound of $\Omega(f^{1/2 - 1/2k} n^{1 + 1/k})$ bits is \emph{asymptotically smaller} than the $nf$ folklore limit. This discrepancy highlights the possibility that classical spanners might be storing far more information than is necessary to answer distance queries reliably. These observations lead to the central question motivating this work:

\begin{center}
\emph{Can $f$-fault tolerant $t$-distance oracles be implemented in $o(nf)$ bits of space?}
\end{center}

As our main contribution, we answer this question in the affirmative, showing that one can decisively surpass the $nf$ barrier.

\paragraph{Our conceptual approach.}
We exploit the robustness of high-degree graph structures under bounded deletions to \emph{compress away} the parts of a graph that are information-theoretically expensive to store for deterministic spanners but not for deterministic oracles or sketches against oblivious adversaries.
Our framework combines this structural robustness with ideas from deterministic and oblivious sparse recovery: the structured dense regions act as intrinsically stable “cores’’ whose connectivity can be reasoned about without storing all edges explicitly, while the small residual boundary of the graph is captured exactly using compact sketches.
This hybrid viewpoint, namely, structural robustness for the dense part, and algebraic recovery for the sparse part, forms the unifying principle behind all our results. It allows us to bypass the $n\!\cdot\!f$ barrier entirely, yielding subquadratic deterministic oracles, oblivious sketches, and streaming algorithms. 

The remainder of this section presents these results in turn, before the Technical Overview (\cref{sec:technical_overview}) describes the mechanisms in detail.

\subsection{Our Contributions}

\paragraph{Fault-Tolerant Distance Oracles}

Our first contribution is the design of an $f$-fault-tolerant $O(\log(n)\log\log(n))$-stretch distance oracle that decisively bypasses the folklore $\Omega(nf)$ lower bound for spanners:

\begin{theorem}[Deterministic oracles in $\widetilde{O}(n\sqrt{f})$ space]\label{thm:introDO} 
There is a \emph{deterministic} data structure, which, when instantiated on any graph $G=(V,E)$ and $f \in [0,\binom{n}{2}]$:
\begin{enumerate}
\item uses $\widetilde{O}(n\sqrt{f})$ bits of space, and
        \item supports a query operation which, for any $F \subseteq E$ satisfying $|F| \leq f$, and any $u, v \in V$, returns a value $\widehat{\mathrm{dist}}_{G - F}(u,v)$ such that 
        \[
        \mathrm{dist}_{G - F}(u,v)\leq \widehat{\mathrm{dist}}_{G - F}(u,v) \leq O(\log(n)\log\log(n)) \cdot \mathrm{dist}_{G - F}(u,v).
        \]
    \end{enumerate}
\end{theorem}

Since the oracle above is deterministic, its guarantees hold simultaneously for all $F \subseteq E$ with $|F|\le f$ and all vertex pairs $u,v \in V$. This result establishes a sharp separation between fault-tolerant distance oracles and fault-tolerant spanners: the former can be implemented in polynomially smaller space. For example, when $f=n$, no nontrivial deterministic fault-tolerant spanner exists for any finite stretch, whereas \cref{thm:introDO} demonstrates that significant compression, namely, by a factor of $\sqrt{f}=\sqrt{n}$, is achievable once we shift to the goal of distance reporting.

Our data structure is also robust enough to work with weighted graphs: for graphs with polynomially bounded edge weights, we show that $f$-fault-tolerant $O(\log(n)\log\log(n))$-stretch distance oracles can still be implemented in $\widetilde{O}(n\sqrt{f})$ bits of space (see \cref{sec:weighted} for details).

While \cref{thm:introDO} is stated existentially, we further show that, with only an additional $\log(n)$ factor loss in the stretch guarantee, the oracle can be constructed and queried efficiently in polynomial time:

\begin{theorem}[Poly-time deterministic oracles in $\widetilde{O}(n\sqrt{f})$ space]\label{thm:IntrofaultTolerantDOEfficient}
There is a \emph{deterministic} data structure, constructible in \emph{polynomial-time}, such that, when instantiated on any graph $G=(V,E)$ and $f \in [0,\binom{n}{2}]$:
\begin{enumerate}
\item uses $\widetilde{O}(n\sqrt{f})$ bits of space, and
\item supports a polynomial time query operation which, for any $F \subseteq E$ satisfying $|F| \leq f$, and any $u, v \in V$, returns a value $\widehat{\mathrm{dist}}_{G - F}(u,v)$ such that 
        \[
        \mathrm{dist}_{G - F}(u,v)\leq \widehat{\mathrm{dist}}_{G - F}(u,v) \leq O(\log^3(n)) \cdot \mathrm{dist}_{G - F}(u,v).
        \]
    \end{enumerate}
\end{theorem}

Both of these aforementioned results rely on \emph{expander decomposition}, and the $\mathrm{polylog}(n)$ stretch exactly reflects the naive bound on the diameter of an expander. Thus, a natural question to ask is whether these distance oracles can still achieve $o(nf)$ bits of space \emph{while also achieving constant stretch}. In this direction, we provide an affirmative answer:

\begin{theorem}[Poly-time deterministic oracles with constant stretch]\label{thm:constantStretchIntro}
There is a \emph{deterministic} data structure, constructible in \emph{polynomial-time}, such that, when instantiated on any graph $G=(V,E)$ and $n \leq f \leq n^{3/2}$:
\begin{enumerate}
\item uses $\widetilde{O}(n^{3/2}f^{1/3})$ bits of space, and
\item supports a polynomial time query operation which, for any $F \subseteq E$ satisfying $|F| \leq f$, and any $u, v \in V$, returns a value $\widehat{\mathrm{dist}}_{G - F}(u,v)$ such that 
        \[
        \mathrm{dist}_{G - F}(u,v)\leq \widehat{\mathrm{dist}}_{G - F}(u,v) \leq 7 \cdot \mathrm{dist}_{G - F}(u,v).
        \]
    \end{enumerate}
\end{theorem}

For instance, using $f = n$, the above implies that stretch $7$ distance oracles can be implemented in $\widetilde{O}(n^{11/6})$ bits of space, while any $n$-fault-tolerant spanner with stretch $7$ would necessarily require $\Omega(n^2)$ many edges.

Proving this theorem relies on a new type of graph decomposition into sub-structures that we call \emph{$2$-hop stars}, which simultaneously provide high minimum degree and constant diameter, even under potentially huge numbers of edge deletions. We view this new graph decomposition as one of our core contributions, and elaborate further upon this technique in \cref{sec:technical_overview}.

Even when unbounded computational time is available, we conjecture that $\Omega(n\sqrt{f})$ bits are necessary for any $f$-fault-tolerant $\mathrm{polylog}(n)$-stretch distance oracle, suggesting that our bounds are essentially tight. Prior work on \emph{vertex-incidence linear sketches}, a central technique (see \cite{AGM12, AGM12b, GMT15} for examples of its use) for compressing graphs in settings with an arbitrary number of edge deletions, provides compelling evidence that any linear sketch capable of recovering a $\mathrm{polylog}(n)$-stretch distance oracle requires $n^{2-o(1)}$ bits~\cite{filtser2021graph, CKL22}. In our setting when $f = \Omega(n^2)$ (morally, the regime of unbounded deletions), these results imply an $\Omega(n f^{1/2 - o(1)})$-bit barrier for general fault-tolerant distance oracles, unless there is a non-trivial separation between the sizes of vertex-incidence linear sketches and arbitrary distance oracles. This matches our upper bound up to lower-order terms. We discuss this conjecture in greater detail, and show how it would even imply an $\Omega(n f^{1/2 - o(1)})$-bit lower bound for \emph{arbitrary} $f$ in \cref{sec:oneWayCommunication}.

\paragraph{Fault-Tolerant Oblivious Spanners}

In many applications, however, it is desirable to explicitly recover a spanner of the underlying graph that reports actual paths and not merely distances. The folklore $\Omega(nf)$ lower bound rules this out for deterministic constructions or against adaptive failures. This naturally raises the question: can we achieve similar bounds when the failures are \emph{oblivious}, i.e., chosen independently of the algorithm’s randomness? Our second main result answers this affirmatively, showing that one can recover spanners after oblivious deletions while maintaining the same $\widetilde{O}(n\sqrt{f})$ space as in \cref{thm:introDO}:

\begin{theorem}[Oblivious sketches in $\widetilde{O}(n\sqrt{f})$ space]\label{thm:IntroObliviousSpanner}
There is a \emph{randomized} sketch, which, when applied to a graph $G=(V,E)$ and $f \in [0,\binom{n}{2}]$, uses $\widetilde{O}(n\sqrt{f})$ bits of space and for any set of edge deletions $F \subseteq E$ with $|F| \le f$, with probability at least $1 - 1/{\rm poly}(n)$ over the randomness of sketch,
can be used to recover an $O(\log(n)\log\log(n))$-stretch spanner of $G - F$.
\end{theorem}

Both \cref{thm:introDO} and \cref{thm:IntroObliviousSpanner} also have immediate implications for designing efficient protocols for distance oracles and oblivious spanners in the one-way two-party communication setting. In this setting, Alice is given a graph $G = (V, E)$ and Bob is given a set of edges $F \subseteq E, |F| \leq f$. The goal is to understand the smallest possible size of a message that Alice can transmit to Bob that enables him to reconstruct a distance oracle (or spanner) for the post-deletion graph $G-F$. By directly using \cref{thm:introDO} and \cref{thm:IntroObliviousSpanner} and sending the corresponding sketches, these yield protocols communicating $\widetilde{O}(n\sqrt{f})$ bits of space. We elaborate more on this setting in \cref{sec:oneWayCommunication}.

\paragraph{Preserving Distances in Streams of Edges}
We next consider the \emph{streaming setting}, where edges of the graph are revealed as an arbitrary sequence of edge insertions and deletions. Prior work has produced both non-trivial linear sketches for this model and limitations of linear sketching based approaches. In particular, Filtser, Kapralov, and Nouri~\cite{filtser2021graph} gave a linear sketching algorithm that constructs spanners of stretch $\widetilde{O}(n^{\frac{2}{3}(1-\alpha)})$ using $\widetilde{O}(n^{1+\alpha})$ space, and Chen, Khanna, and Li~\cite{CKL22} later showed that, for a broad class of linear sketching algorithms and $\alpha < 1/10$, this tradeoff is essentially optimal. In particular, when the stream admits an unbounded number of deletions and the goal is to maintain a $\mathrm{polylog}(n)$-stretch spanner, all currently known algorithms require $\Omega(n^{2-o(1)})$ bits of space.

Motivated by \cref{thm:introDO} and \cref{thm:IntroObliviousSpanner}, we ask whether similar guarantees can be achieved in the \emph{bounded-deletion} streaming model, where the total number of edge deletions in the stream is limited to $f$ arbitrary deletions, for some given parameter $f$. This model, studied in works such as~\cite{jayaram2018data, khanna2025near, khanna2025streaming}, captures a natural middle ground between \emph{fully dynamic} and \emph{insertion-only streaming}.

In this setting, we design algorithms that construct both \emph{deterministic distance oracles} and \emph{randomized spanners} against oblivious adversaries with subquadratic space usage for polynomially large values of $f$.

\begin{theorem}[Streaming deterministic oracles in $\widetilde{O}(n^{4/3} f^{1/3})$ space]\label{thm:IntroFaultTolerantDOStream}
    For any $f \in [0,\binom{n}{2}]$, there is a \emph{deterministic} streaming algorithm which when given an arbitrarily ordered stream consisting of edge insertions (defining a graph $G$) and at most $f$ edge deletions (denoted by $F \subseteq G$):
    \begin{enumerate}
        \item uses space at most $\widetilde{O}(n^{4/3} f^{1/3})$ bits, and
        \item supports a query operation which, for any vertices $u, v \in V$, returns a value $\widehat{\mathrm{dist}}_{G - F}(u,v)$ such that 
        \[
        \mathrm{dist}_{G - F}(u,v)\leq \widehat{\mathrm{dist}}_{G - F}(u,v) \leq O(\log(n)\log\log(n)) \cdot \mathrm{dist}_{G - F}(u,v).
        \]
    \end{enumerate}
\end{theorem}

\begin{theorem}[Streaming oblivious spanners in $\widetilde{O}(n^{4/3} f^{1/3})$ space]\label{thm:IntroObliviousSpannerStream}
    For any $f \in [0,\binom{n}{2}]$, there is a \emph{randomized} streaming algorithm which, when given an arbitrarily ordered stream consisting of edge insertions (defining a graph $G$) and at most $f$ \emph{oblivious} edge deletions (denoted by $F \subseteq G$):
    \begin{enumerate}
        \item uses space at most $\widetilde{O}(n^{4/3} f^{1/3})$ bits, and
        \item with probability at least $1 - 1/{\rm poly}(n)$ (over the internal randomness of the algorithm), outputs an $O(\log(n)\log\log(n))$-stretch spanner of $G - F$.
    \end{enumerate}
\end{theorem}

As an illustration of the above theorem, consider when $f = n$. Then the above algorithm requires only $\widetilde{O}(n^{5/3})$ bits of space while maintaining $O(\log(n)\log\log(n))$ stretch, significantly improving upon previously known fault-tolerant spanner constructions, even in the static case.

Taken together, our results highlight a perhaps somewhat surprising conclusion: once the subgraph constraint, inherent in the definition of spanners, is dropped, fault-tolerant distance oracles achieve polynomially smaller space, even in regimes where \emph{no non-trivial spanner exists}.

\paragraph{Comparison with Prior Works}

To our knowledge, only one prior work even discusses the possibility of distance oracles handling more than $n$ edge faults. Indeed, the work of Bodwin, Haeupler, and Parter \cite{BHP24} shows that for a graph $G = (V, E)$, there are spanners which store only $\widetilde{O}(nf)$ many edges, but can tolerate any set of deletions $F \subseteq E$ such that $F$'s \emph{maximum degree} is bounded by $f$ (and thus includes some sets of up to $nf$ many edges, provided they are of the prescribed form). Thus, \cite{BHP24} codifies that the primary difficulty in designing distance reporting data structures is in allowing some vertices to receive a potentially \emph{unbounded} number of deletions. 

Part of our work, like many others \cite{SW18, GKKS20,GGKKS22, BBGNSS22, BHP24}, uses an expander decomposition to compress graphs into smaller space while preserving distances. The key insight we introduce here is that good enough expanders can even handle \emph{unbounded degree edge deletions} while still ensuring that the remaining, high degree vertices are still well-connected (as opposed to prior works like \cite{BHP24} which only consider bounded-degree deletions). This is in turn coupled with the first (to our knowledge) use of \emph{deterministic sparse recovery} in the design of distance oracles to recover edges from any vertices of too low degree. Only when using both of these insights in combination do we get the first fault-tolerant distance oracles which tolerate \emph{unbounded} degree deletions. 

Even so, the use of expanders naturally leads to a barrier of $\mathrm{polylog}(n)$ stretch. To achieve the constant stretch of \cref{thm:constantStretchIntro}, we need a graph decomposition which simultaneously (a) is resilient to \emph{unbounded degree} deletions at each vertex, and (b) \emph{witnesses short paths} between the remaining (post-deletion) high-degree vertices. To our knowledge, no such graph decomposition was known before our work; rather, to overcome this obstacle we carefully introduce a new (efficient) graph decomposition into so-called \emph{$2$-hop stars} which simultaneously achieves both of the desired conditions. We believe that this decomposition is interesting in its own right, and may find other applications in distance-preserving data structures.

\subsection{Technical Overview}\label{sec:technical_overview}

At a high level, our results rely on two complementary insights. First, by carefully decomposing a graph into \emph{structured}, dense regions, we can ensure that these remain (mostly) well-connected
even after many edge deletions, effectively allowing us to treat such regions as \emph{self-certifying cores} whose connectivity can be inferred without explicitly
storing all edges. Second, the small residual boundary left after removing these cores can be compactly represented and exactly recovered using ideas from
\emph{deterministic/oblivious sparse recovery}. 
Together, these principles yield a unified compression framework for preserving distances:
robust structural cores handle large-scale connectivity,
while algebraic sketches encode the fine-grained deletions.
The remainder of this section develops this framework in increasing generality; building intuition first with $\mathrm{polylog}(n)$ stretch for high-degree expanders, generalizing this to arbitrary graphs via expander decomposition, then using our \emph{$2$-hop star} decomposition to achieve \emph{constant} stretch,
and finally applying these to oblivious and streaming settings.

\subsubsection{High-Degree Expanders Under Edge Faults}

We begin with the first component of our framework, namely, the structural robustness of high-degree expanders.
Intuitively, if a graph has sufficiently large minimum degree and expansion, then after deleting a bounded number of edges, the remaining high-degree vertices still induce an expander of nearly the same quality. This phenomenon underlies our ability to “compress away’’ dense regions without losing the ability to certify connectivity.

\paragraph{The Robustness Lemma}
The starting point for all of our results is a novel, strong ``robustness'' phenomenon for high-degree expanders under edge deletions. 
Formally, consider a graph $G=(V,E)$ on $n$ vertices with minimum degree $D$, and suppose $G$ is an $\Omega(1/\log n)$-expander under the usual conductance definition.
This means that for every set $S \subseteq V$,
\[
\frac{\delta_G(S)}{\min(\mathrm{Vol}_G(S), \mathrm{Vol}_G(V\setminus S))} = \Omega(1 / \log(n)),
\]
where $\delta_G(S)$ refers to the number of edges leaving the set $S$ in $G$, and $\mathrm{Vol}_G(S)$ refers to the sum of vertex degrees in the set $S$. 

Now, we wish to understand how $G$'s structure evolves after performing $f$ edge deletions. Towards this end, we let $F \subseteq E$ denote this set of edges to be removed, and crucially, we define $V_{\mathrm{good}} = \{ v \in V: \mathrm{deg}_{G - F}(v) \geq \frac{D}{2}\}$. Intuitively, this set of ``good'' vertices is exactly those vertices whose degree remains largely intact; at least half of the original degree lower bound $D$. The most basic form of our robustness property can then be stated as follows:

\begin{lemma}[Robustness of High-Degree Expanders, informal]
    Let $G, F, f, V_{\mathrm{good}}$ be given as above, and suppose that $D \geq 8 \sqrt{f}\log(n)$. Then, $(G-F)[V_{\mathrm{good}}]$ is an $\Omega(1 / \log(n))$-expander.
\end{lemma}

This lemma implies that, provided the minimum degree of $G$ is sufficiently large, after performing our edge deletions from $F$, the remaining high degree vertices \emph{still constitute a good expander} regardless of how the edges were deleted! Further, by leveraging the connection between edge expansion and diameter, one can show that this means that $(G-F)[V_{\mathrm{good}}]$ has diameter bounded by $O(\log^2(n))$ (i.e., that all pairs of vertices in $V_{\mathrm{good}}$ are connected by paths of length $O(\log^2(n))$). 

It turns out the minimum degree condition above is essential for establishing the above lemma. Indeed, it ensures that \emph{not too many} vertices lose most of their incident edges, and by a careful argument, allows us to argue that the ``core'' of good vertices \emph{roughly inherits} the expansion properties of the original graph. We provide a formal proof of this robustness property in \cref{sec:HDX}.

\paragraph{Extending the Robustness to Lower Degree Vertices}

In fact, while the above lemma implies that the vertices in $V_{\mathrm{good}}$ continue to be connected by short paths, a key property we exploit is that this continues to hold \emph{even for some vertices outside} $V_{\mathrm{good}}$!

\begin{lemma}[Strong Robustness of High-Degree Expanders, informal]\label{lem:strongRobustness}
    Let $G, F, f, V_{\mathrm{good}}$ be given as above, and suppose that $D \geq 8 \sqrt{f}\log(n)$. Then, any vertices $u, v$ of degree $\geq \frac{4f}{D} + 1$ in $G-F$ are connected by a path of length $O(\log^2(n))$.
\end{lemma}

This lemma relies on bounding how many vertices are \emph{not in} $V_{\mathrm{good}}$. Indeed, by showing that the number of non-good vertices is at most $\frac{4f}{D}$, this means that if two vertices $u, v$ both have degree $\geq \frac{4f}{D} + 1$, then $u$ and $v$ must also both have \emph{good} neighbors $u'$ and $v'$. But, from the robustness lemma above we already know that $u'$ and $v'$ are connected by a path of length $O(\log^2(n))$, and so then $u$ and $v$ must be as well! Thus, the high-degree expander on its own (even after deletions!) provides a proof that high-degree pairs of vertices are connected by short paths.

\subsubsection{Designing a Distance Oracle For High-Degree Expanders}

With these robustness guarantees, we now discuss how to design an entire fault-tolerant distance oracle for high-degree expanders. We continue using the same parameters as in the previous discussion: i.e., we are given an $\Omega( 1 / \log(n))$ expander graph $G = (V, E)$, with $n$ vertices, a bound $f$ on the number of edge faults, and a lower bound of $D \geq 8 \sqrt{f} \log(n)$ on the degree of $G$.

To build intuition for the distance oracle, we consider the result of applying the deletions $F$ to the graph $G$. We observe that given only the \emph{vertex degrees} in the graph $G - F$, one can already report distances for the \emph{vast majority} of vertex pairs. Indeed, by \cref{lem:strongRobustness}, we know that any vertices of degree $\geq \frac{4f}{D} +1$ must be connected by paths of length $O(\log^2(n))$, even without \emph{any other knowledge} of what the graph looks like. 

Unfortunately, this same guarantee does not hold for the low-degree vertices, and indeed their structure can be arbitrarily bad (even disconnected from the rest of the graph). But the fact that these vertices are low degree raises another intriguing possibility: namely that, because their degrees are small, perhaps we can recover \emph{all} of their incident edges. It turns out that this is indeed possible, using a type of ``deterministic sparse recovery'' sketch, which we prove in \cref{sec:efficientUniqueSum}.

\begin{claim}[Informal Deterministic Sparse Recovery]
    For a universe $[u]$ and a constant $k$ there is a linear function $g: \zo^{[u]} \rightarrow \F_u^{2k+1}$ such that for any $S \subseteq [u]$, $|S| \leq k$, $g(\zo^S)$ is unique. 
\end{claim}

This primitive allows us to store a short linear sketch of each vertex’s neighborhood that can be uniquely decoded once its degree drops below $k$.
For our use case, we initialize $k = \frac{4f}{D}$, and initialize $u = \binom{|V|}{2}$ (i.e., to be the set of all edges). Now for each vertex $v$, we start by storing the sketching function $g$ applied to the \emph{original neighborhood} of $v$ in $G$, which we denote by $\mathcal{S}_v = g(\Gamma_G(v))$. Later, when edges are deleted in the set $F$, we carefully update the sketch, importantly exploiting its \emph{linearity}. In fact, via these careful updates, the sketch guarantees that if the degree of any vertex becomes sufficiently small, then we can \emph{exactly recover} the neighborhood of $v$ in $G - F$!

Combining these two modes, implicit connectivity via short paths for vertices in dense cores, and exact recovery for vertices with sparse boundaries, yields a deterministic oracle with desired space complexity of $\widetilde{O}(n\sqrt{f})$ bits for high-degree expanders.

\subsubsection{Fault-Tolerant Distance Oracles in General Graphs}

Unfortunately, for arbitrary graphs $G$, there is no guarantee that their structure is as robust under arbitrary sets of $f$ deletions. Thus, if we are to leverage the robustness of high-degree expanders, we must show that given a graph $G$, we can \emph{find} high-degree expander subgraphs. Fortunately, this turns out to be doable (and even efficiently, with mild ramifications to the expansion).

Formally, we take advantage of a two-step decomposition beginning with a graph $G = (V, E)$, and a target minimum degree $D$: we start by iteratively removing all vertices of degree $\leq D \cdot \log^2(n)$, and then, on the remaining ``core'' of high-degree vertices, we invoke an expander decomposition. In fact, by leveraging a careful self-loop argument originating in \cite{SW18}, we can even show that the resulting expanders from this decomposition \emph{approximately inherit} the minimum degree!

After just one level of expander decomposition however, the only guarantee is that a constant fraction of edges have been covered by the expanders. To achieve the sparsity necessary for our distance oracle, we invoke $O(\log(n))$ recursive levels of this so-called \emph{minimum-degree preserving expander decomposition}. We note that similar expander decompositions have appeared in several prior works, see, for instance, \cite{SW18, GKKS20,GGKKS22, BBGNSS22}.

With the robustness of high-degree expanders and the high-degree expander decomposition in hand, the path forward is now a natural one, but one that still requires care. Indeed, we compute this recursive high-degree expander decomposition of the starting graph $G$, and for each expander, we store the list of its vertices, its degrees, and our sparse recovery sketches. However, recall that our goal is to create a \emph{fault-tolerant} distance oracle; thus, when edge deletions arrive, we still must determine where to perform these edge deletions. Indeed, incorrectly applying even a single deletion can invalidate our sparse recovery sketches. 

Nevertheless, we show that a careful ordering of our expanders in conjunction with them being induced subgraphs leads to a \emph{deterministic} procedure to determine which expander contains a given deleted edge. Then, by taking the \emph{union} of all of our expander sketches, we show that the total space required, and the stretch achieved, satisfy the conditions of \cref{thm:faultTolerantDO}.

\subsubsection{Achieving Constant Stretch via Covering by $2$-Hop Stars}

\paragraph{The Challenge}

The preceding sections outline an argument via expander decomposition for building $f$-fault tolerant distance oracles with stretch $O(\log^2(n))$ in space $\widetilde{O}(n \sqrt{f})$. While we later present a more refined version which achieves stretch $O(\log(n)\log\log(n))$, it is worth observing that \emph{any} expander-based construction will \emph{inherently} require stretch $\Omega(\log(n))$: indeed, even an $\Omega(1)$ expander as per our definition may still have a diameter of $\Omega(\log(n))$ which leads to a stretch of the same value. 

To achieve $O(1)$ stretch, we must then design a different graph structure which, for any set of deletions $F$, \emph{still witnesses} $O(1)$ length paths between any vertices which did not undergo too many deletions, \emph{even if} some vertices may have received unbounded degree deletions. Towards designing such a structure, we can observe that this condition effectively imposes two \emph{necessary} (although not necessarily sufficient) constraints on our graph structure, even before any deletions are imposed: (a) the \emph{minimum degree} should be (polynomially) large, and (b) the diameter should be $O(1)$. Thus, to build intuition for our final distance oracle data structure, we first study whether even these simpler graph structures are guaranteed to exist.

\paragraph{Finding High Min-Degree and Constant Diameter Subgraphs}

Our first step towards achieving $O(1)$ stretch is showing that the above conditions \emph{are} simultaneously achievable in dense graphs. More formally, we show the following:

\begin{lemma}
    Let $G = (V, E)$ be a graph on $n$ vertices with $n \cdot d$ edges. Then, there is an induced subgraph $G[V']$ such that:
    \begin{enumerate}
        \item $G[V']$ has minimum degree $\Omega \left (\frac{d^2}{n} \right )$.
        \item $G[V']$ has diameter $4$.
    \end{enumerate}
\end{lemma}

The proof of this lemma relies on a structure that we call the \emph{$2$-hop star}. More formally, let us suppose that our starting graph $G$ actually has \emph{minimum} degree $d$ (as opposed to just average degree). Because our goal is to find a graph with constant diameter, we choose a vertex $v$, and will grow a ball out from this vertex. Indeed, to start, we consider the graph $G[v \cup \Gamma(v)]$ which is the induced subgraph on $v$ and all of $v$'s neighbors. 

Naturally, $G[v \cup \Gamma(v)]$ satisfies the diameter constraint, as every vertex is connected to $v$, and thus all pairwise distances are bounded by $2$. However, it is possible that the \emph{minimum degree} constraint \emph{is not} satisfied; we can consider, for instance, that $G[v \cup \Gamma(v)]$ is simply a star. The key observation however is the following: if the average degree of $G[v \cup \Gamma(v)] < d / 10$ (as it would be in a star), then we will simply \emph{grow out} another layer of the star from every vertex in $\Gamma(v)$. That is to say, we now consider the graph $G[v \cup \Gamma(v) \cup \Gamma^2(v)]$ (where we now consider only the edges from $v$ to $\Gamma(v)$, and from $\Gamma(v)$ to $\Gamma^2(v)$, in particular, we do not consider the edges entirely contained \emph{within} one of $\Gamma(v)$ or $\Gamma^2(v)$).

A priori, it may seem that this runs into the same issues as before, namely that some vertices in $G[v \cup \Gamma(v) \cup \Gamma^2(v)]$ may have small degree, and this is indeed true. However, the key property we exploit is that $G[v \cup \Gamma(v) \cup \Gamma^2(v)]$ must have \emph{high average degree}. This is because $|\Gamma(v)| \geq d$, and each vertex $u \in \Gamma(v)$ has $\geq 9d/10$ edges \emph{leaving} $\Gamma(v)$. Written another way, this implies that $|E[\Gamma(v), \Gamma^2(v)]| = \Omega(d^2)$. If we trivially bound the number of vertices in $v \cup \Gamma(v) \cup \Gamma^2(v)$ by $n$, then this means that the average degree of $G[v \cup \Gamma(v) \cup \Gamma^2(v)]$ is $\Omega \left ( \frac{d^2}{n}\right )$. By simply peeling off vertices of degree $o(d^2 / n)$, we can even turn this into an induced subgraph with $\Omega \left ( \frac{d^2}{n}\right )$ \emph{minimum degree}.

All that remains to be seen is that this procedure does not decrease the diameter. Because every remaining vertex in $\Gamma^2(v)$ is connected to vertices only in $\Gamma(v)$, to bound the diameter by $4$, we effectively must only show that the \emph{root vertex} $v$ is \emph{not peeled} during this procedure. Indeed, if the vertex $v$ remains intact, then every other vertex trivially finds a path to $v$ of length $2$.
Now, to see this, we observe that each remaining \emph{unpeeled} vertex in $\Gamma^2(v)$ must have degree $\Omega \left ( \frac{d^2}{n}\right )$. Since $\Gamma(v)$ and $\Gamma^2(v)$ form a bipartite graph, this is only possible if $\Gamma(v)$ itself retains $\Omega \left ( \frac{d^2}{n}\right )$ unpeeled vertices. But, because $v$ has an edge to every unpeeled vertex in $\Gamma(v)$, this means that $v$'s degree is also sufficiently high to have not been peeled!

\paragraph{Showing Robustness Under Deletions for $2$-hop Stars}

Formally, a $2$-hop star is specified by a \emph{root vertex} $v$, the first layer of neighbors $L_1 \subseteq \Gamma(v)$, and the second layer of neighbors $L_2 \subseteq \Gamma^2(v)$, with the convention that the only edges participating in the star are those edges either (a) between $v$ and $L_1$ or (b) between $L_1$ and $L_2$. In a graph $G$ with minimum degree $d$, we showed above that \emph{for every vertex $v$}, one can find a $2$-hop star (or $1$-hop)\footnote{Note that this case is even simpler than the $2$-hop case, so we omit it from the discussion here.} with minimum degree $d^* = \Omega(d^2/n)$ and diameter $4$.

However, there is already a glaring shortcoming in $2$-hop stars when it comes to guaranteeing small diameter under deletions: indeed, if in a $2$-hop star, \emph{all of the root's edges} are deleted, then there is immediately no guarantee that the rest of graph has small diameter or even that it is connected! The key fact that we exploit though, is that the root is the \emph{only} ``sensitive'' vertex. Indeed, provided the root vertex does not receive too many deletions, then \emph{all the other} high-degree vertices (after performing deletions) are still connected with constant diameter!

We formalize this with the following key lemma:

\begin{lemma}\label{lem:highDegreeConnectRootDepth2Intro}
    Let $H$ be a $2$-hop star with root $v$, and vertex set $\{v, L_1, L_2\}$ that has minimum degree $d^*$. Now, consider a set $F$ of $\leq f$ edge deletions such that at most $d^*/2$ of $v$'s edges are deleted, and assume $\frac{d^*}{2} \geq \frac{5f}{d^*}$. Then, for any vertex $u$ with $\mathrm{deg}_{H - F}(u) \geq \frac{5f}{d^*}$, if $u$ is in $L_1$, $u$ has a path of length $\leq 3$ to $v$ in $H - F$, and if $u$ is in $L_2$, $u$ has a path of length $\leq 4$ to $v$ in $H - F$.
\end{lemma}

For intuition behind this proof, consider a vertex $u$ in $L_1$ with $\mathrm{deg}_{H - F}(u) \geq \frac{5f}{d^*}$. If $u$ happens to still have an edge to $v$ in the post-deletion graph, then the lemma follows trivially. Otherwise, suppose that all $\geq \frac{5f}{d^*}$ of $u$'s edges only go to $L_2$. The key point is that, because the total deletion budget was only $f$ edges, some of $u$'s neighbors must still have high degree. Indeed, for the vertices in $u$'s neighborhood in $H - F$, their total degree \emph{before deletions} is
\[
\sum_{w \in \Gamma_{H -F}(u)} \mathrm{deg}_H(w) \geq \frac{5f}{d^*} \cdot d^* = 5f.
\]
Thus, after $f$ deletions,
\[
\sum_{w \in \Gamma_{H -F}(u)} \mathrm{deg}_{H- F}(w) \geq \sum_{w \in \Gamma_{H -F}(u)} \mathrm{deg}_H(w) - f > \frac{1}{2} \cdot \sum_{w \in \Gamma_{H -F}(u)} \mathrm{deg}_{H}(w),
\]
and so there must be some vertex $w \in \Gamma_{H -F}(u)$ whose degree post-deletions is $\mathrm{deg}_{H- F}(w)  > \mathrm{deg}_{H}(w)/2 \geq d^* / 2$. Thus, $w$ has more than $d^* / 2$ neighbors in $L_1$. But, by assumption, \emph{at most} $d^* / 2$ vertices in $L_1$ had their edges to the root vertex $v$ deleted. Thus, $w$ must still have \emph{at least} one neighbor in $L_1$ whose edge to $v$ remains intact post deletions! Denoting this vertex in $L_1$ by $y$, we trace out a path from $u$ to $w$ to $y$ to $v$, meaning $y$'s distance from $v$ is at most $3$. A similar argument will show that for vertices in $L_2$ of sufficiently high degree, they retain a path to $v$ of length $4$.

\paragraph{Building the Fault Tolerant Distance Oracle}

Now, starting with \cref{lem:highDegreeConnectRootDepth2Intro}, we can already make some key insights. Indeed, imagine that an edge $e$ participates in a $2$-hop star originating from root vertex $v$, and that $v$ does not receive too many deletions in the deletion set $F$. Then, \cref{lem:highDegreeConnectRootDepth2Intro} implies that any high degree vertices (here meaning $\approx f / d^*$) are connected by short paths. But what about those vertices with degree smaller than $f / d^*$? Here, just as in the expander-based distance oracle we take advantage of our \emph{sparse recovery sketches}. By instantiating our sparse recovery sketches with threshold $\approx f / d^*$ on each (non-root) vertex in a $2$-hop star, we can ensure that provided that the root $v$ doesn't have too many edges deleted:

\begin{enumerate}
    \item If a non-root vertex $u$ has degree $\geq f / d^*$, then $u$ has a short path to $v$.
    \item If a non-root vertex $u$ has degree $< f / d^*$, then we can recover all of its neighboring edges!
\end{enumerate}

Denoting our $2$-hop star by $H$, this means that for any edge $e = (x,y) \in H$, as long as the root $v$ doesn't have too many deletions, we either \emph{explicitly} recover the edge $e$ (via sparse recovery), or we are guaranteed a short path (of length $\leq 7$)\footnote{The $7$ comes from the fact that at most one of $x,y$ can be in $L_2$, as there are no edges between vertices in $L_2$. So, the stretch is bounded by $4 + 3$.} exists between $x$ and $y$ via $v$!

This now motivates the final piece of our data structure. Imagine that starting with our graph $G$, we decomposed $G$ into a sequence of $2$-hop stars $H_1, H_2, \dots H_k$. Now, after performing deletions $F$, imagine that it was the case that, \emph{for every edge} $e \in G - F$, it was the case that there was \emph{some} $2$-hop star $H_i$ such that $e \in H_i - F$, and the root vertex in $H_i - F$ didn't receive more than $d^* / 2$ deletions. The above would then imply that the stretch reported between $e$'s endpoints is not too large.

Naturally however, the problem is that a naive decomposition into $2$-hop stars will not provide the above guarantee, as it is possible that for many edges, their corresponding $2$-hop stars had too many edge deletions incident upon their root. To address this, we introduce redundancy: for every edge $e \in E$, rather than enforcing that $e$ participates in only one $2$-hop star, we enforce that $e$ participates in $\approx \frac{4f}{d^*}$ many $2$-hop stars, each from a different root. In this way, no matter which edges are deleted in the future, we are still guaranteed that there is \emph{some} $2$-hop star $H_i$ which contains $e$ and preserves high root degree (and thus witnesses a short path between $e$'s endpoints).

In \cref{sec:constantStretch}, we show that such a decomposition into different $2$-hop stars which ``covers'' each edge sufficiently many times is possible, while still using space below the $nf$ bit barrier. This construction, along with a more careful version of the above argument which bounds the stretch, then leads to the proof of \cref{thm:constantStretchIntro}.

\subsubsection{Beyond the Static Setting: Oblivious Adversary and Streaming Spanners}

\paragraph{Oblivious Spanners}
Upon revisiting the aforementioned expander-based construction, one can see that the only reason the distance oracle is not a spanner is that \emph{within} the high degree vertices of each expander component (post-deletions), we do not recover an actual subgraph, but rather only have the guarantee that all vertices in a high-degree expander core are connected by short paths. As mentioned earlier, this shortcoming is inherent: there is no way to \emph{deterministically} recover a subgraph which witnesses this tight connectivity without then creating a spanner, which in turn carries with it an $\Omega(nf)$ lower bound. However, if we allow ourselves some randomness, and to only report distances correctly for \emph{most} deletion sets $F$, then we can in fact design sketches to recover subgraphs which witness the small diameter property within these expanders, thereby recovering spanners of the post-deletion graph (which we call \emph{oblivious spanners}).

Doing this relies on the notion of a \emph{randomized} $\ell_0$-sampler linear sketch. Such linear sketches behave similarly to that of the deterministic sparse recovery sketch, with the modification that when the degree of a vertex is \emph{larger} than $\frac{4f}{D}$, the sampler now returns a \emph{random sample} of $\frac{4f}{D}$ many edges (with high probability).

Thus, to recover these subgraph witnesses to the low-diameter-ness of the high degree vertices we store (a) $\ell_0$-samplers to sample $\frac{4f}{D}$ many edges from each vertex and (b) a \emph{global random sample} of roughly $n \cdot D$ of the edges in the graph. The sketch in (a) ensures that high-degree vertices recover at least one neighbor in $V_{\mathrm{good}}$, and a careful analysis of (b) ensures that the resulting sampled edges from the induced subgraph on $V_{\mathrm{good}}$ in fact constitute a good expander. Together with the sketches used in the distance oracle (recovering all edges incident on low degree vertices) this constitutes an $O(\log^2(n))$-spanner in the post-deletion graph.

\paragraph{Distance Oracles and Spanners in Streams}

All of the preceding discussion has relied on the ability to compute the expander decomposition and sketches of the \emph{entire graph}, presented all at the same time. To generalize these results to the streaming setting, where edges are only presented one at a time, requires more care. To simplify the parameters in the discussion below, we focus on the setting when $f = n$.

Because our streaming algorithm cannot afford to store \emph{all} of the edges in the graph, it instead stores a ``buffer'' of edges, in which it places all of the non-expander edges so far. Whenever the size of the buffer reaches a set threshold (which we set to be roughly $n^{5/3}$ many edges), this then triggers a round of high-degree ($D = n^{2/3}$) expander decomposition on the buffer edges. Any edges landing in an expander are removed from the buffer, while those not in expanders remain in the buffer. By our expander decomposition properties, we know that an $\Omega(1)$ fraction of edges are then placed in \emph{some} expander, and so this opens up more capacity (an $\Omega(1)$ fraction) in the buffer.

Importantly, in this setting we can no longer afford to use a degree lower bound of $D \approx \sqrt{f}$, and instead we require a larger choice of the degree. Indeed, in the regime above where $f = n$, we require $D \approx n^{2/3}$. 

There is one final detail which appears when dealing with edges in streams: typically, the expander decomposition has the nice property that expanders are all induced subgraphs. In the static setting, this gave us an important structural property to exploit when finding the expander component that contains a deleted edge $e$. In the streaming setting, this is no longer the case, as when we first compute our expander decomposition, it is only an induced subgraph with respect to the \emph{current state of the buffer}. Naively, this then leads to ambiguity when it comes time to delete an edge $e$, as it could foreseeably be in any one out of many different expander components. 

To overcome this, we introduce a minor caveat in the above buffer algorithm: whenever an edge $e$ is inserted in the stream, we first check if \emph{any of the already existing} expander components can receive the edge $e$ (i.e., if both endpoints of $e$ are in the same component), and if so, we insert this edge into the sketches of that component. By careful book-keeping, this leads to a deterministic way to identify \emph{which} expander component an edge deletion belongs to, but also requires a stronger version of the robustness lemma: namely that, for any high-degree expander component, \emph{and any set of edge insertions}, after a subsequent round of $f$ edge deletions, it is still the case that high-degree vertices have low diameter. We prove this lemma in \cref{sec:HDX} thereby yielding the correctness of the streaming algorithms in \cref{sec:streaming}.

\subsection{Open Questions}
Our results highlight several interesting directions for future work:

\begin{enumerate}
\item \textbf{Optimal Space for Fault-Tolerant Connectivity.} Consider the problem of designing fault-tolerant \emph{connectivity} oracles. In this setting, one is given a graph $G = (V, E)$, along with a parameter $f$, and wishes to report, for any $F \subseteq E, |F| \leq f$, the connected components of $G - F$. Our work shows that this can be done in $\widetilde{O}(n \sqrt{f})$ many bits of space. Is there a matching lower bound? It suffices to show that for $f = n^2$ any such connectivity data structure must store $\Omega(n^2)$ bits of space.
    \item \textbf{Optimal Space / Fault Tolerance / Stretch Tradeoff.} Our work shows that stretch $7$ in the presence of $f$ faults can be achieved in $\widetilde{O}(n^{3/2} f^{1/3})$ bits of space, and that stretch $O(\log(n) \log\log(n))$ can be achieved in $\widetilde{O}(n \sqrt{f})$ bits of space. But what is the optimal amount of space for a data structure achieving stretch $2k-1$ and $f$-fault tolerance? 
\end{enumerate}

\section{Preliminaries}

\subsection{Notation}

Throughout the paper, we use $G = (V, E)$ to denote a graph, with $V$ being the vertex set and $E$ being the edge set. We use $n$ to denote the number of vertices, and use $m$ to denote the number of edges. For a set $V' \subseteq  V$, we let $G[V']$ denote the induced subgraph on $V'$, i.e., all edges $e$ such that $e \subseteq V'$. For sets $S, T \subseteq V$, we use $E_G[S, T]$ to denote all edges that have one endpoint in $S$ and one endpoint in $T$. When the graph $G$ is clear from context, we drop $G$ from the subscript.

Often, we use $|G[V']|$ to denote the number of edges in the induced subgraph on $V'$.

\subsection{Linear Sketching Basics}

To start, we introduce several notions that we will require. The first are those of linear sketching and $\ell_0$-samplers:

\begin{definition}[Linear Sketching]
    For graphs over a vertex set $V$, we say that a sketch $\mathcal{S}: \zo^{\binom{V}{2}} \rightarrow \Z^s$ is a \emph{linear sketch}, if, for a graph $G_A = (V, E_A)$ and a graph $G_B = (V, E_B)$ such that $E_B \subseteq E_A$, it is the case that $\mathcal{S}(G_A - G_B) = \mathcal{S}(G_A) - \mathcal{S}(G_B)$.
\end{definition}

In this setting, we will represent a graph via its incidence matrix in $\{-1,0,1\}^{V \times E}$. Formally:

\begin{definition}\cite{AGM12}\label{def:graphEncoding}
    Given an unweighted graph $G = (V, E)$ with $|V| = n$, define the $n \times \binom{n}{2}$ matrix $A$ with entries $(i, e)$, where $i \in [n]$ and $e \in [\binom{n}{2}]$. We say that 
    \[
    A_{i,e} = \begin{cases}
        1 & \text{ if } i \in e, i \neq \max_{j \in e} j, \\
        -1 & \text{ if } i \in e, i = \max_{j \in e} j, \\
        0 & \text{ else.}
    \end{cases}
    \]
    Let $a_1, \dots a_n$ be the rows of the matrix $A$. The support of $a_i$ corresponds with the neighborhood of the $i$th vertex.
\end{definition}

\begin{definition}\label{def:ell0samp}
    Let $S$ be a turnstile stream, $S = s_1, \dots s_t$, where each $s_i = (u_i, \Delta_i)$, and the aggregate vector $x \in \R^u$ where $x_i = \sum_{j:u_j = i} \Delta_i$. Given a target failure probability $\delta$, a $\delta$ $\ell_0$-sampler for a non-zero vector $x$ returns $\perp$ with probability $\leq \delta$, and otherwise returns an element $i \in [u]$ with probability $\frac{|x_i|_0}{|x|_0}$.
\end{definition}

Throughout this paper, we will often set $u = \binom{n}{2}$ to be the number of possible edges (sometimes referred to as the number of \emph{edge slots}) in the graph. We will use the following construction of $\ell_0$-samplers:

\begin{theorem}\cite{jowhari2011tight, CF14}\label{thm:ell0samp}
    For any universe of size $u$, there exists a \emph{linear} sketch-based $\delta$-$\ell_0$-sampler using space $O(\log^2(u) \log (1 / \delta))$.\footnote{Assuming that the number of deletions is bounded by $\mathrm{poly}(u)$.} 
\end{theorem}

We also formally state the following deletion property of $\ell_0$-samplers, as we will make use of this throughout our paper.

\begin{remark}\label{rmk:removeEdge}
    Suppose we have a linear sketch for the $\ell_0$-sampler of the edges leaving some vertex $v$, denoted by $\calS(v, R)$ (where $R$ denotes the randomness used to build the sketch). Suppose further that we know there is some edge $e$ leaving $v$ that we wish to remove from the support of $\calS(v, R)$, which was recovered independently of the randomness $R$. Then, we can simply add a linear vector update to $\calS(v, R)$ that cancels out the coordinate corresponding to this edge $e$. 
\end{remark}

Note that the sketches described so far are \emph{randomized}, and thus are only valid against oblivious adversaries. Ultimately, one of our contributions is to design sketches that are valid against adaptive adversaries. For this, we use the following facts regarding deterministic sparse recovery.

\begin{fact}\label{fact:uniqueSum}[See, for instance, \cite{KLS93}, Lemma 2.1]
    For a universe $[u]$ and a constant $k$, there is a function $g: [u] \rightarrow \R^{\geq 0}$ such that:
    \begin{enumerate}
        \item For any $S \subseteq [u]$, $|S| \leq k$, $g(S) = \sum_{i \in S}g(i)$ is unique. 
        \item For every $i \in [u]$, $g(i) \leq \mathrm{poly}(u^k)$.
    \end{enumerate}
\end{fact}

For efficient implementations, we also use the following generalization of the sketch mentioned above. This sketch can also be used instead of the one above for the inefficient sketches, but we include the one above as it is simpler to state.

\begin{claim}\label{clm:efficientUniqueSum}
    For a universe $[u]$ and a constant $k$, let $q$ be a prime of size $[u, 2u]$. There is a polynomial time computable linear function $g: \zo^{[u]} \rightarrow \F_q^{2k+1}$ such that:
    \begin{enumerate}
        \item For any $S \subseteq [u]$, $|S| \leq k$, $g(\zo^S)$ is unique. 
        \item For any vector $z \in \F_q^{2k+1}$ such that there exists $S \subseteq [u]$, $|S| \leq k$ and $g(\zo^S) = z$, there is a polynomial time algorithm for recovering $S$ given $z$.
    \end{enumerate}
\end{claim}

We include a proof of the above claim in the appendix (see \cref{sec:efficientUniqueSum}).

\subsection{Expanders and Expander Decomposition}

Next, we require some definitions and notions of expansion.

Notationally, we will use $\mathrm{deg}(v)$ to refer to the degree of vertex $v$. In general, this degree may be taken with respect to arbitrary subgraphs $H$ of a graph $G$. When this degree is taken with respect to such a subgraph, it will be denoted by $\mathrm{deg}_H(v)$. Going forward, we use $G[V']$ to denote the \emph{induced subgraph} of $G$ on vertex set $V'$.

\begin{definition}
	For a graph $G = (V, E)$, and for a set $S \subseteq V$, we let $\mathrm{Vol}(S) = \sum_{v \in S} \mathrm{deg}(v)$, and we let $\delta(S) = \left | \{ e = (u,v) \in E: u \in S, v \in \bar{S} \vee u \in \bar{S}, v \in S \} \right |$.
\end{definition}

With this, we can now define the expansion of each set $S \subseteq V$:

\begin{definition}
	For a graph $G = (V, E)$, and for a set $S \subseteq V$, we say the conductance of $S$ is
	\[
	\phi(S) = \frac{\delta(S)}{\min(\mathrm{Vol}(S), \mathrm{Vol}(\bar{S}))}.
		\]
\end{definition}

The expansion of the graph $G$ is now the minimum conductance of any set $S \subseteq V$:

\begin{definition}
	For a graph $G = (V, E)$, we define its expansion as $\phi(G) = \min_{S: \emptyset \subsetneq S \subsetneq V} \phi(S)$.
\end{definition}

We also occasionally use the notion of a \emph{lopsided-expander} (see \cite{bringmann2025near}):

\begin{definition}
    For a graph $G = (V, E)$, and a cut $S \subseteq V$, we define the lopsided-conductance of $S$ as
    \[
    \psi(S) = \frac{\delta(S)}{\min(\mathrm{Vol}(S), \mathrm{Vol}(\bar{S})) \cdot \log \frac{\mathrm{Vol}(V)}{\min(\mathrm{Vol}(S), \mathrm{Vol}(\bar{S}))}}.
    \]
    A graph $G = (V, E)$ is a $\psi_0$-lopsided expander if for every $S \subseteq V$, $\psi(S) \geq \psi_0$.
\end{definition}

We will frequently use the following expander-decomposition statement, which is essentially a restatement of \cite{SW18, BBGNSS22, GGKKS22, bringmann2025near}:

\begin{theorem}\label{thm:expanderDecomp}
    Let $G = (V, E)$ be a graph, with $|V| = n$ and $|E| = m$, and let $1 \leq D \leq n$ be an arbitrary parameter. Then, there exists a set of disjoint parts of the vertex set, denoted $V_1, \dots V_k$ such that:
    \begin{enumerate}
        \item Each induced subgraph $G[V_i]$ is an $\Omega(1 / \log(n))$ lopsided-expander. 
        \item Each induced subgraph $G[V_i]$ has minimum degree $D$. 
        \item The total number of edges \emph{not in} one of $G[V_i], i \in [k]$ is bounded by
        \[
        \left |E - \bigcup_{i \in [k]} G[V_i] \right | \leq O(n D\log^2(n)) + \frac{m}{2}.
        \]
    \end{enumerate}
\end{theorem}

We prove this statement below:

\begin{proof}
    To start, we iteratively remove vertices in $G$ of degree $\leq D \log^2(n)$. Note that each time we remove such a vertex, the number of edges that is removed is $\leq D \log^2(n)$. By default, these removed vertices and their edges will not participate in any expanders, and yield the $O(n D \log^2(n))$ in the above bound. Now, let $V'$ be the resulting vertex set for which the minimum degree of $G[V']$ is now $\geq D \log^2(n)$.

    On $G[V']$ we now apply the expander decomposition of \cite{SW18} adapted to lopsided expanders, (see \cite{bringmann2025near} for this argument).\footnote{See \cite{BBGNSS22, GGKKS22} for similar invocations of this expander decomposition.} This guarantees a decomposition of $V'$ into $V_1, \dots V_k$ such that each $G[V_i]$ is an $\Omega(1 /  \log(n))$ lopsided-expander, and that the number of edges not contained in any expander is at most $m / 2$. Further, this guarantees that, for a vertex $v \in V_i$, 
    \[
    \frac{\delta_{G[V_i]}(v)}{\mathrm{Vol}_G(v)} = \Omega(1 /\log(n)),
    \]
    i.e., that the expansion guarantee holds even when the volume is measured with respect to the original vertex degrees (this argument relies on adding self-loops to vertices to keep their degrees the same throughout the expander decomposition process - in this way, the volume of each set of vertices is unchanged after removing the crossing edges, see \cite{SW18} for this argument). Importantly, this directly implies that 
    \[
    \mathrm{deg}_{G[V_i]}(v) = \Omega \left ( \frac{\mathrm{deg}_{G}(v)}{\log(n)}  \right )\geq D.
    \]
\end{proof}

A property of expanders that we will frequently use is that expansion provides a bound on the diameter of the graph:

\begin{definition}
    The diameter of a graph $G = (V, E)$ is the maximum distance between any two vertices $u, v \in V$.
\end{definition}

We have the following elementary connection between expansion and diameter:

\begin{claim}\label{clm:diameterBound}
    Let $G = (V, E)$ be a $\phi$ lopsided-expander for $\phi = O(1 / \log(n))$. Then, the diameter of $G$ is bounded by $O(\log\log(n) / \phi)$.
\end{claim}

Likewise, we will make use of the fact that randomly sampling edges in a high-degree expander preserves expansion up to a constant factor (with high probability). Formally:

\begin{claim}\label{clm:subsampleExpander}
    Let $G = (V, E)$ be a graph with minimum degree $D$ and (lopsided-)expansion $\Omega(1 / \log(n))$. Then, randomly sampling $\frac{|E|\log^3(n)}{D}$ many edges from $G$ yields an $\Omega(1 / \log(n))$ (lopsided-)expander with probability $1 - \frac{1}{2^{\Omega(\log^2(n))}}$.
\end{claim}

\begin{proof}
Consider any set $S \subseteq V$ of size $|S| = k$. By definition $\mathrm{Vol}(S) \geq k \cdot D$, and because $G$ is an $\Omega(1 / \log(n))$ expander, this means that $\delta(S) = \Omega(k D / \log(n))$ as well.

Now, when we sample $\frac{|E|\log^3(n)}{D}$ many edges (yielding a graph $\widetilde{G}$), this means that in expectation, $\mathrm{Vol}_{\widetilde{G}}(S)$ has $\frac{\mathrm{Vol_G}(S)\log^3(n)}{D} \geq k \log^2(n)$ surviving edges, and in expectation, $\delta_{\widetilde{G}}(S) =\Omega(k \cdot \log^{2}(n))$ as well. By a Chernoff bound, this means that with probability $1 - 2^{-\Omega(k \log^2(n))}$, 
\[
\mathrm{Vol}_{\widetilde{G}}(S) \leq 2 \cdot \E[\mathrm{Vol}_{\widetilde{G}}(S)] \leq \frac{\log^3(n)}{D} \cdot \mathrm{Vol}_{G}(S),
\]
and 
\[
\delta_{\widetilde{G}}(S) \geq \frac{1}{2} \cdot \E[\delta_{\widetilde{G}}(S)] \geq \frac{\log^3(n)}{D} \delta_G(S). 
\]
In particular, this means that with probability $1 - 2^{-\Omega(k \log^2(n))}$,
\[
\phi_{\widetilde{G}}(S) \geq \frac{\delta_{\widetilde{G}}(S)}{\mathrm{Vol}_{\widetilde{G}}(S)} = \Omega \left (\frac{\delta_G(S)}{\mathrm{Vol}_{G}(S)} \right )  = \Omega(1 / \log(n)). 
\]
Note that in the case of lopsided-expanders, the above equation implies that random sampling even preserves the stronger lopsided-expansion requirement.

Now, because there are $\leq 2^{k \log(n)}$ sets $S$ of size $k$, this means that by a union bound, \emph{every set} $S$ of size $k$ has expansion $\Omega(1 / \log(n))$ with probability $1 - 2^{k \log(n) -\Omega(k \log^2(n))} = 1 - 2^{-\Omega(k \log^2(n))}$. Finally, we can take a union bound over $k$ ranging from $1$ to $n$ conclude that every set $S$ has expansion $\Omega(1 / \log(n))$ with probability $1 - 2^{- \Omega(\log^2(n))}$.
\end{proof}

For the efficient implementation of our distance oracle, we use the following efficient statement of expander decomposition.

\begin{theorem}\label{thm:efficientExpanderDecomp}
    Let $G = (V, E)$ be a graph, with $|V| = n$ and $|E| = m$, and let $1 \leq D \leq n$ be an arbitrary parameter. Then, there is a polynomial time algorithm finding a set of disjoint parts of the vertex set, denoted $V_1, \dots V_k$ such that:
    \begin{enumerate}
        \item Each induced subgraph $G[V_i]$ is an $\Omega(1 / \log^2(n))$-expander. 
        \item Each induced subgraph $G[V_i]$ has minimum degree $D$. 
        \item The total number of edges \emph{not in} one of $G[V_i], i \in [k]$ is bounded by
        \[
        \left |E - \bigcup_{i \in [k]} G[V_i] \right | \leq O(n D\log^4(n)) + \frac{m}{2}.
        \]
    \end{enumerate}
\end{theorem}

\begin{proof}
    The algorithm is simply to use a $C \cdot \log(n)$-approximation algorithm for calculating the conductance of a graph $G$ (see \cite{leighton1999multicommodity}) and then to iteratively remove the sparsest cut found by the approximation algorithm so long as the conductance of the returned cut is $\leq \frac{1}{100 C \log(n)}$ (for a large constant $C$). 

    As in \cite{SW18}, every time a sparse cut $S$ is found with conductance $\leq \frac{1}{100 C \log(n)}$ we charge the crossing edges to the remaining edges on the smaller side of the cut. Note that each edge can only be charged $O(\log(m))$ times (after being on the smaller side of a cut $\log(m)$ times, the graph is empty). Likewise, because each sparse cut has conductance $\leq \frac{1}{100 C \log(n)}$, if $S$ denotes the smaller side of the cut, then at least half of the edges in $\mathrm{Vol}(S)$ \emph{stay within} $S$. Thus, when we charge the remaining edges in $S$, each receives charge $\leq 2 \cdot \Phi(S) \leq \frac{2}{100 C \log(n)}$. At the termination of the procedure, the total charge on each edge is at most $\log(m) \cdot \frac{2}{100 C \log(n)} \leq \frac{4}{100}$, and so the number of cut edges is at most $4m / 100$.

    Note that the stopping condition of the above algorithm is that the returned sparsest cut has conductance $> \frac{1}{100 C \log(n)}$, which, by our approximation guarantee, implies that all cuts have conductance $> \frac{1}{100 C \log(n)} \cdot \frac{1}{C \log(n)} = \Omega(1 / \log^2(n))$.

To get the degree condition in the above claim, we iteratively remove all vertices with degree $\leq 100 D \log^4(n)$ before starting the expander decomposition algorithm. This removes at most $O(n D \log^4(n))$ edges from the graph (note that these removed vertices are then reported as their own components), but ensures that all degrees are at least $100 D \log^4(n)$ when we start the expander decomposition. Next, whenever we cut the graph along a sparse cut, for each cut edge $(u,v)$, we add self-loops on $u, v$ to keep their degrees the same. The remaining charging argument stays the same, but in the final expander components, we know that each $G[V_i]$ is an $\Omega(1 / \log^2(n))$-expander \emph{even with} the original vertex volumes (see \cite{SW18} for the original proof using this statement). In particular each vertex cut in $G[V_i]$ has expansion $\Omega(1 / \log^2(n))$, which means for $v \in V_i$, $E_{G[V_i]}[v, V_i - v] \geq \Omega(1 / \log^2(n)) \cdot \mathrm{deg}_G(v) = \Omega(100 D \log^4(n) / \log^2(n)) \geq D$. 

The bound on the number of crossing edges follows from the number of edges removed during the degree decomposition ($O(n D \log^4(n))$) along with the number of edges removed during expander decomposition ($4m/100$).
\end{proof}

\section{Robustness of High-Degree Expanders Under Bounded Edge Deletions}\label{sec:HDX}

In this section, we motivate the high-degree expander decomposition discussed in \cref{thm:expanderDecomp} by showing the \emph{robustness} of such expanders under edge insertions and deletions. 

We denote the number of deletions by $f$. Formally, we prove the following lemma:

\begin{lemma}\label{lem:robustnessfDeletion}
    Let $f$ be a parameter, and let $H = (V_H, E_H)$ be an $\Omega(\frac{1}{\log(n)})$-lopsided-expander with minimum degree $D \geq 8\sqrt{f}\log(n)$. Now, suppose that there are at most $f$ edge deletions to $H$, resulting in a graph $H' = (V_H, E'_H)$. Then, for every pair of vertices $u_1, u_2 \in V_H$ such that $\mathrm{deg}_{H'}(u_b) \geq \frac{4f}{D} + 1$ for $b \in \{1, 2\}$, it is the case that $\mathrm{dist}_{H'}(u_1, u_2) = O(\log(n)\log\log(n))$.
\end{lemma}

To prove this lemma, we require the following notion of ``good vertices'':

\begin{definition}\label{def:HdelfDeletion}
 For a vertex $v \in V_H$, we say that $v$ is \emph{good} if $\mathrm{deg}_{H'}(v) \geq \frac{D}{2}$, and otherwise say that $v$ is bad. We let $V_{\mathrm{good}}$ denote the set of all good vertices, and likewise let $V_{\mathrm{bad}}$ denote the set of all bad (not good) vertices. 
\end{definition}

We start by showing that the graph $H'[V_{\mathrm{good}}]$ is a good expander (and therefore has small diameter):

\begin{claim}\label{clm:goodGoodExpanderfDeletion}
    Let $H' = (V_H, E'_H)$ be constructed as in \cref{def:HdelfDeletion} with the graph $H$ as in \cref{lem:robustnessfDeletion}. Then, the induced subgraph $H'[V_{\mathrm{good}}]$ has diameter $O(\log(n)\log\log(n))$.
\end{claim}

\begin{proof}
    To prove this, we in fact show that the subgraph $H'[V_{\mathrm{good}}]$ is an $\Omega(1 / \log(n))$-lopsided-expander. From here, \cref{clm:diameterBound} directly implies that diameter is bounded by $O(\log(n)\log\log(n))$. 

    Now, to show that $H'[V_{\mathrm{good}}]$ is an $\Omega(1 / \log(n))$-lopsided-expander, we start by bounding the number of bad vertices. In particular, because the minimum degree of $H$ was originally $D$ in $E_H$, and any bad vertex has degree $\leq \frac{D}{2}$ in $E'_H$, it must be the case that each vertex $v \in V_{\mathrm{bad}}$ had at least half of its edges deleted, and that therefore, there can be at most $\frac{4f}{D}$ many bad vertices.

Now, let us consider the expansion of any set $S \subseteq V_{\mathrm{good}}$. By definition, 
\[
\phi_{H'[V_{\mathrm{good}}]}(S) = \frac{\delta_{H'[V_{\mathrm{good}}]}(S)}{\min(\mathrm{Vol}_{H'[V_{\mathrm{good}}]}(S), \mathrm{Vol}_{H'[V_{\mathrm{good}}]}(V_{\mathrm{good}} - S))}.
\]
Now, immediately, we can observe that 
\[
\min(\mathrm{Vol}_{H'[V_{\mathrm{good}}]}(S), \mathrm{Vol}_{H'[V_{\mathrm{good}}]}(V_{\mathrm{good}} - S)) \leq \min(\mathrm{Vol}_{H}(S), \mathrm{Vol}_{H}(V_{\mathrm{good}} - S))
\]
\[
\leq \min(\mathrm{Vol}_{H}(S), \mathrm{Vol}_{H}(V_H - S))
\]
as undoing the deletions and restriction to the induced subgraph can only increase the individual vertex degrees. Now, it remains only to upper bound the numerator of the above expression. 

Here, we claim that $\delta_{H'[V_{\mathrm{good}}]}(S) \geq \delta_H(S) - 5f$. To see why, we can first observe $\mathrm{Vol}_H(V - V_{\mathrm{good}}) \leq 4f$, as the $f$ deletions remove \emph{at least half} of the total volume of $V - V_{\mathrm{good}}$. But, $f$ deletions can only remove a total volume of $2f$, and as such, the volume of $V - V_{\mathrm{good}}$ must be bounded by $4f$. Thus, we see that $\delta_{H'[V_{\mathrm{good}}]}(S) \geq \delta_{H'}(S) - 4f$, as restricting to the induced subgraph on $V_{\mathrm{good}}$ removes at most $4f$ volume, and thus at most $4f$ crossing edges. Now, it is straightforward to see that $\delta_{H'}(S) \geq \delta_H(S) - f$, as $H'$ is achieved by deleting at most $f$ edges from $H$. Thus,  we see that 
\[
\delta_{H'[V_{\mathrm{good}}]}(S) \geq \delta_{H'}(S) - 4f \geq \delta_H(S) - 5f,
\]
 as we desire. All together, this implies that 
 \[
 \phi_{H'[V_{\mathrm{good}}]}(S) = \frac{\delta_{H'[V_{\mathrm{good}}]}(S)}{\min(\mathrm{Vol}_{H'[V_{\mathrm{good}}]}(S), \mathrm{Vol}_{H'[V_{\mathrm{good}}]}(V_{\mathrm{good}} - S))} \geq
 \]
 \[
 \frac{\delta_H(S) - 5f}{\min(\mathrm{Vol}_{H}(S), \mathrm{Vol}_{H}(V_H - S))}.
 \]

 Now, if we consider a set $S \subseteq V_{\mathrm{good}}$ such that $\min(\mathrm{Vol}_{H'[V_{\mathrm{good}}]}(S), \mathrm{Vol}_{H'[V_{\mathrm{good}}]}(V_{\mathrm{good}} - S))\geq f\log^2(n)$, then we are immediately done, as 
 \[
 \frac{\delta_H(S) - 5f}{\min(\mathrm{Vol}_{H}(S), \mathrm{Vol}_{H}(\mathrm{GOOD} - S))} \geq \frac{\delta_H(S)}{\min(\mathrm{Vol}_{H}(S), \mathrm{Vol}_{H}(\mathrm{GOOD} - S))} - \frac{5f}{f\log^2(n)} 
 \]
 \[ = \phi_H(S) - O(1 / \log^2(n)) = \Omega(\phi_H(S)),
 \]
 as $\phi_H(S) = \Omega(1 / \log(n))$. This immediately yields that for sets $S \subseteq V_{\mathrm{good}}$ such that 
 \[
 \min(\mathrm{Vol}_{H'[V_{\mathrm{good}}]}(S), \mathrm{Vol}_{H'[V_{\mathrm{good}}]}(V_{\mathrm{good}} - S)) \geq f \log^2(n),
 \]
 $\phi_{H'[V_{\mathrm{good}}]}(S) = \Omega(\phi_H(S))$ (and thus satisfy the lopsided-expansion requirement). So it remains only to lower bound the expansion for sets $S \subseteq V_{\mathrm{good}}$ for which $\min(\mathrm{Vol}_{H'[V_{\mathrm{good}}]}(S), \mathrm{Vol}_{H'[V_{\mathrm{good}}]}(V_{\mathrm{good}} - S)) \leq f \log^2(n)$. Indeed, let us consider such a set $S$, and let us assume WLOG that the minimum is achieved for $S$ (as opposed to its complement). 
 
 To lower bound the expansion of $S$, we first lower bound the degree of every vertex in $H'[V_{\mathrm{good}}]$. Indeed, as mentioned above, we know that $\mathrm{Vol}_H(V - V_{\mathrm{good}}) \leq 4f$, so we know that 
 \begin{align}\label{eq:boundBadSizefdel}
 |V_{\mathrm{bad}}| = |V - V_{\mathrm{good}}| \leq \frac{4f}{D} \leq \frac{4f}{8\sqrt{f}\log(n)} \leq \frac{ \sqrt{f}}{2\log(n)},
 \end{align}
 where we have used our assumption that $D \geq 8\sqrt{f}\log(n)$.
 Likewise, we know that for every $v \in V_{\mathrm{good}}$, $\mathrm{deg}_{H'}(v) \geq \frac{D}{2}$. Thus, when we restrict to the induced subgraph on $V_{\mathrm{good}}$, the degree of $v$ can decrease by at most $\frac{ \sqrt{f}}{2\log(n)}$, as at most this many vertices are deleted. I.e., 
 \begin{align}\label{eq:boundDegreeGOODfdel}
 \mathrm{deg}_{H'[V_{\mathrm{good}}]}(v) \geq \mathrm{deg}_{H'}(v) - \frac{\sqrt{f}}{2\log(n)} \geq \frac{D}{2} - \frac{D}{16\log^2(n)} \geq \frac{D}{4}.
 \end{align}

 So, when working with a set $S$ such that $\mathrm{Vol}_{H'[V_{\mathrm{good}}]}(S) \leq f \log^2(n)$, it must be the case that 
 \[
 |S| \leq \frac{f\log^2(n)}{D/4} \leq \frac{f \log^2(n)}{2\sqrt{f}\log(n)} \leq \sqrt{f}\log(n)/2.
 \]
 Now, for a vertex $v \in S$, because (by \cref{eq:boundDegreeGOODfdel}) $\mathrm{deg}_{H'[V_{\mathrm{good}}]}(v) \geq \frac{D}{4} \geq 2 \sqrt{f}\log(n)$, this immediately implies that at least half of $v$'s edges leave the set $S$, as at most $|S|-1 \leq \sqrt{f}\log(n)/2$ of $v$'s neighbors can be \emph{within} the set $S$. Thus, $\phi_{H'[V_{\mathrm{good}}]}(S) \geq \frac{1}{2}$. This concludes the proof, as we have shown that for every $S \subseteq V_{\mathrm{good}}$, $\phi_{H'[V_{\mathrm{good}}]}(S) = \Omega( \phi_H(S))$, and thus $H'[V_{\mathrm{good}}]$ has the same lopsided-expansion properties as $H$, and therefore has $O(\log(n) \log\log(n))$ diameter (by \cref{clm:diameterBound}). 
\end{proof}

Now, we proceed to a proof of our more general lemma:

\begin{proof}[Proof of \cref{lem:robustnessfDeletion}.]
    Consider any two vertices $u_1, u_2 \in V_H$ such that $\mathrm{deg}_{H'}(u_b) \geq \frac{4f}{D} + 1$ for $b \in \{1, 2\}$. We have several cases:

\begin{enumerate}
    \item If $u_1, u_2 \in V_{\mathrm{good}}$, then by \cref{clm:diameterBound}, $\mathrm{dist}_{H'}(u_1, u_2) = O(\log(n) \log\log(n))$. 
    \item If $u_1 \in V_{\mathrm{good}}$ and $u_2 \notin V_{\mathrm{good}}$ (or vice versa), then importantly, we rely on the fact that $\deg_{H'}(u_2) \geq \frac{4f}{D} + 1$. In particular, by \cref{eq:boundBadSizefdel}, we know that $|V_{\mathrm{bad}}| \leq \frac{4f}{D}$. Because $\deg_{H'}(u_2) > |V_{\mathrm{bad}}|$, this means that $u_2$ must have some neighbor $v' \in V_{\mathrm{good}}$. In particular, we then obtain that 
    \[
    \mathrm{dist}_{H'}(u_1, u_2) \leq 1 + \mathrm{dist}_{H'}(u_1, v') \leq 1 + O(\log(n)\log\log(n)),
    \]
    where we use the bound on $\mathrm{dist}_{H'}(u_1, v')$ from the preceding point. 
    \item If $u_1 \notin V_{\mathrm{good}}$ and $u_2 \notin V_{\mathrm{good}}$, but $\mathrm{deg}_{H'}(u_b) \geq \frac{4f}{D} + 1$ for $b \in \{1, 2\}$, then we again use the reasoning from the preceding point. In particular, we know there is some neighbor $v'_1$ of $u_1$ in $H'$ and some neighbor $v'_2$ of $u_2$ in $H'$ such that $v'_1, v'_2 \in V_{\mathrm{good}}$. Thus, \[
    \mathrm{dist}_{H'}(u_1, u_2) \leq 2 + \mathrm{dist}_{H'}(v'_1, v'_2) \leq 2 + O(\log(n)\log\log(n)),
    \]
    yielding our bound.
\end{enumerate}
\end{proof}

Note that, in the streaming setting, we will also deal with expanders where an arbitrary number of \emph{insertions} may occur alongside the $f$ deletions. In this case, we use the following corollary of \cref{lem:robustnessfDeletion}:

\begin{corollary}\label{cor:robustnessfDeletionInsertion}
        Let $f$ be a parameter, and let $H = (V_H, E_H)$ be an $\Omega(\frac{1}{\log(n)})$-expander with minimum degree $D \geq 8\sqrt{f}\log(n)$. Now, suppose that there is an arbitrary number of edge insertions in $H$, and that there are at most $f$ edge deletions to $H$, resulting in a graph $H' = (V_H, E'_H)$. Then, for every pair of vertices $u_1, u_2 \in V_H$ such that $\mathrm{deg}_{H'}(u_b) \geq \frac{4f}{D} + 1$ for $b \in \{1, 2\}$, it is the case that $\mathrm{dist}_{H'}(u_1, u_2) = O(\log(n)\log\log(n))$.
\end{corollary}

\begin{proof}
    Consider applying \emph{only the deletions} to the graph $H$. Then, by \cref{lem:robustnessfDeletion}, we know that the vertices $V_{\mathrm{good}}$ have diameter $O(\log(n)\log\log(n))$, and that there are at most $\frac{4f}{D}$ vertices that are \emph{not} in $V_{\mathrm{good}}$.

    Now, let us consider the graph $H'$ that also includes the edge insertions, and consider any vertices $u_1, u_2 \in V_H$ such that $\mathrm{deg}_{H'}(u_b) \geq \frac{4f}{D} + 1$ for $b \in \{1, 2\}$. Because $|V_{\mathrm{bad}}| \leq \frac{4f}{D}$, it must be the case that there exists $v_1$ neighboring $u_1$ and $v_2$ neighboring $u_2$ such that $v_1, v_2 \in V_{\mathrm{good}}$, and so by \cref{lem:robustnessfDeletion}, we know there is a path of length $O(\log(n)\log\log(n))$ between $v_1, v_2$. In total then, we see that
    \[
    \mathrm{dist}_{H'}(u_1, u_2) \leq 2 + \mathrm{dist}_{H'_{\mathrm{del}}}(v_1, v_2) \leq O(\log(n)\log\log(n)),
    \]
    where we use $H'_{\mathrm{del}}$ to refer to the graph $H'$ with only deletions included (and no insertions). This concludes the corollary. 
\end{proof}

Likewise, we are not restricted to choosing our expansion parameter to be $\Omega(1 / \log(n))$. In general, we can use graphs with worse expansion provided we increase the minimum degree requirement. We include this version of the statement for the later sections when we work with \emph{efficient} expander decompositions (which lead to worse expansion).

\begin{lemma}\label{lem:robustnessfDeletionWorseExpansion}
    Let $f$ be a parameter, and let $H = (V_H, E_H)$ be an $\Omega(\frac{1}{\log^2(n)})$-expander with minimum degree $D \geq 8\sqrt{f}\log^2(n)$. Now, suppose that there are at most $f$ edge deletions to $H$, resulting in a graph $H' = (V_H, E'_H)$. Then, for every pair of vertices $u_1, u_2 \in V_H$ such that $\mathrm{deg}_{H'}(u_b) \geq \frac{4f}{D} + 1$ for $b \in \{1, 2\}$, it is the case that $\mathrm{dist}_{H'}(u_1, u_2) = O(\log^3(n))$.
\end{lemma}

We omit the proof of this lemma, as it exactly mirrors the proof of \cref{lem:robustnessfDeletion}.

\section{Fault-Tolerant Distance Oracles in $\widetilde{O}(n \sqrt{f})$ Space}

With our high-degree expander robustness lemma, we are now ready to prove our main result regarding the existence of small space, fault-tolerant distance oracles. Formally, we prove the following:

\begin{theorem}\label{thm:faultTolerantDO}
    There is a data structure which, when instantiated on a graph $G = (V, E)$ on $n$ vertices, and any parameter $f \in \left [\binom{n}{2} \right ]$:
    \begin{enumerate}
        \item Uses space at most $\widetilde{O}(n \sqrt{f})$ bits. 
        \item Supports a query operation which, for any $F \subseteq E$ satisfying $|F| \leq f$, and any vertices $u, v \in V$, returns a value $\widehat{\mathrm{dist}}_{G - F}(u,v)$ such that 
        \[
        \mathrm{dist}_{G - F}(u,v)\leq \widehat{\mathrm{dist}}_{G - F}(u,v) \leq O(\log(n)\log\log(n)) \cdot \mathrm{dist}_{G - F}(u,v).
        \]
    \end{enumerate}
\end{theorem}

The data structure $D$ that we create is \emph{deterministic}, and hence the correctness guarantee holds \emph{for every} set $F \subseteq E, |F| \leq f$, even when such queries are made by an adaptive adversary.

\subsection{Constructing the Data Structure}\label{sec:buildDO}

As alluded to above, the data structure relies on a recursive, high-degree expander decomposition of the input graph $G$. We first describe this process informally, and then present the formal pseudocode below (see \cref{alg:buildDataStructure}).

Given a starting graph $G = (V, E)$, and a parameter $f$ designating the number of edge-faults that the data structure should tolerate, we start by decomposing the graph $G$ using \cref{thm:expanderDecomp} with minimum degree $D = 8 \sqrt{f}\log(n)$. By \cref{thm:expanderDecomp}, this guarantees that the graph is decomposed into induced subgraphs $G[V_1], \dots G[V_k]$ such that each is an $\Omega(1 / \log(n))$ lopsided-expander, and each has minimum degree $D$. Furthermore, \cref{thm:expanderDecomp} guarantees that the number of edges in $E - G[V_1] - \dots - G[V_k] = X$ is bounded by $|E| / 2 + O(nD\log^2(n)) = |E|/2 + O(n \sqrt{f}\log^2(n))$. 

Now, for each expander $G[V_i]$, our data structure stores three pieces of information:
\begin{enumerate}
    \item The names of all the vertices in this component (i.e., it stores $V_i$).
    \item The degrees of all the vertices in $G[V_i]$.
    \item For every vertex $v \in V_i$, it stores a deterministic sparse recovery sketch on the edges incident to $v$ in $G[V_i]$ (using the function $g$ from \cref{fact:uniqueSum}). 
\end{enumerate}

Given an edge deletion, the first piece of information allows the data structure to decide where the edge should be deleted from (i.e., which expander). The second piece of information allows the data structure to track the degrees of vertices under edge deletions, and finally, the third piece of the data structure allows for recovery of the \emph{the entire neighborhood} of a vertex, whenever the degree becomes sufficiently small. 

However, the above procedure only reduces the number of edges in the graph to $|E| / 2 + O(nD\log^2(n)) = |E|/2 + O(n \sqrt{f}\log^2(n))$. Thus, the data structure then recurses on these remaining edges (which we denote by $X$), storing the same sketch for the result of the expander decomposition on $X$. This process repeats until the number of remaining edges becomes $O(nD\log^2(n))$, at which point the algorithm stores all the remaining edges.

We present the formal algorithm constructing the data structure in \cref{alg:buildDataStructure}.

\begin{algorithm}[h]
    \caption{BuildDistanceOracle$(G = (V, E), f)$}\label{alg:buildDataStructure}
    $\ell = 1$. \\
    $G_{\ell}= G$. \\
    $D = 8 \sqrt{f}\log(n)$. \label{line:setDegree}\\
    Let $g$ be a function as guaranteed by \cref{fact:uniqueSum}, initialized with $u = \binom{n}{2}$ and $k = \frac{4f}{D}$.\label{line:uniqueSumFunction}\\
    Let $C$ be a sufficiently large constant. \\
    \While{$G_{\ell}$ has $\geq C n \sqrt{f}\log^3(n)$ edges \label{line:whileLoop}}{
    Let $V^{(\ell)}_1, \dots V^{(\ell)}_{k_{\ell}}$ denote the high-degree expander decomposition of $G_{\ell}$ as guaranteed by \cref{thm:expanderDecomp} with $D$, and let $X_{\ell+1}$ denote the crossing edges. \label{line:expanderDecomp}\\
    \For{$j \in [k_{\ell}]$, and each vertex $v \in V_j^{(\ell)}$}{ 
    Let $\Gamma_{G_{\ell}[V_j^{(\ell)}]}(v)$ denote the edges leaving $v$ in the expander $G_{\ell}[V_j^{(\ell)}]$. \\
    Let $\mathcal{S}_v^{(\ell)} = g(\Gamma_{G_{\ell}[V_j^{(\ell)}]}(v))$. \\
    }
    Let $G_{\ell +1} = (V, X_{\ell+1})$. \\
    $\ell \leftarrow \ell + 1$.
    }
    $\mathrm{DataStructure} = \bigg ( \bigg \{(V^{(j)}_1, \dots V^{(j)}_{k_{j}}): j \in [\ell] \bigg \}, \{\mathrm{deg}_{G_{j}[V_i^{(j)}]}(v): j \in [\ell], i \in k_{j}, v \in V^{(j)}_i \}, \{ \mathcal{S}_v^{(j)}: j \in [\ell], v \in \bigcup_{i \in k_{j}} V^{(j)}_i \}, X_{\ell} \bigg )$.\\
\Return{$\mathrm{DataStructure}$}.
\end{algorithm}

\subsection{Bounding the Space}

Now, we bound the space used by \cref{alg:buildDataStructure}. We prove the following lemma:

\begin{lemma}\label{lem:spaceBound}
    For a graph $G = (V, E)$, and a parameter $f$, \cref{alg:buildDataStructure} returns a data structure that requires $\widetilde{O}(n \sqrt{f})$ bits of space.
\end{lemma}

\begin{proof}
    First, we bound the number iterations in \cref{alg:buildDataStructure} (equivalently, we bound the final value of $\ell$). For this, recall that the starting graph $G$ has $\leq n^2$ edges and by \cref{thm:expanderDecomp}, there exists a constant $C'$ such that after each invocation of the expander decomposition, the number of edges in $X_i$ is at most $|X_{i-1}|/2 + C' \cdot n \cdot D \cdot\log^2(n)$. So, if we let our large constant in \cref{alg:buildDataStructure} be $C = 4 \cdot C'$, then we know that in every iteration where $|X_i| \geq C n \sqrt{f}\log^3(n)$, it is the case that $|X_{i+1}| \leq \frac{|X_i|}{2} + \frac{|X_i|}{4} \leq \frac{3|X_i|}{4}$. Thus after only $\ell = O(\log(n))$ iterations, it must be the case that $|X_{\ell}| \leq C \cdot n \sqrt{f}\log^3(n)$. 

    The first piece of information that is stored in our data structure is the identity of all the components in each level of our expander decomposition. These components always form a partition of our vertex set, and thus storing $(V^{(\ell)}_1, \dots V^{(j)}_{k_{j}})$ for a single $ j \in [\ell]$ requires $\widetilde{O}(n)$ bits of space. Across all $\ell = O(\log(n))$ levels of the decomposition, storing this information still constitutes only $\widetilde{O}(n)$ bits of space.

    Next, to store all the vertex degrees, for each vertex we store a number in $[n]$, which requires $O(\log(n))$ bits. Across all $n$ vertices and $O(\log(n))$ levels of the decomposition, this too constitutes only $\widetilde{O}(n)$ bits of space.

    Next, we consider the sparse recovery sketches that we store for each vertex. Recall that each sketch is invoked using \cref{fact:uniqueSum} with parameter $u = \binom{n}{2}$, and $k = \frac{4f}{D}$. Thus, each sketch is the sum of at most $n$ numbers, each of which is bounded by $\mathrm{poly} \left ( \binom{n}{2} \right )^k$. Thus the bit complexity is 
    \[
    O \left (\log(\mathrm{poly} \left ( \binom{n}{2} \right )^k) \right ) = O(k \cdot \log(n)) = O \left ( \frac{4f \log(n)}{D}\right ) =  O(\sqrt{f})
    \]
    many bits. Now, across all $\ell = O(\log(n))$ levels, and all $n$ vertices, the total space required for storing these sketches is $\widetilde{O}(n \sqrt{f})$ bits of space. 

    Finally, we bound the space of the final crossing edges $X_{\ell}$. By the termination condition in \cref{alg:buildDataStructure}, it must be the case that $|X_{\ell}| = O(n \sqrt{f}\log^3(n))$, and therefore the total space required for storing these edges is $\widetilde{O}(n \sqrt{f})$ bits of space. This concludes the lemma. 
\end{proof}

\subsection{Reporting Distances}

In this section, we show how to answer distance queries using our above data structure. We first describe the procedure informally, and later present pseudocode the actual implementation. 

Indeed, let $G = (V, E)$ denote the starting graph for which our data structure is instantiated, and let $F \subseteq E, |F| \leq f$ denote the set of deleted edges. To start, we claim that, given this set of deletions $F$, there is a simple way to determine which expander each edge $e \in F$ should be deleted from. This follows from the fact in the $j$th level of our decomposition, the expanders are \emph{induced subgraphs} of $G_j$. So, for any edge $e = (u,v)$, the expander which contains $(u,v)$ is the \emph{first} expander we ever saw in our decomposition where the vertices $u$ and $v$ were together. Note that until this happens, $(u,v)$ is always a crossing edge, and once $(u,v)$ are together in an expander, this means $(u,v) \in G_j$ and there exists $V_i^{(j)}$ such that $u, v \in V_i^{(j)}$, and therefore $(u,v) \in G[V_i^{(j)}]$. So, indeed there is a simple deterministic way to determine which expander an edge $(u,v)$ is contained in. 

We include this claim formally here:

\begin{claim}\label{clm:findingEdge}
    Let $G = (V, E)$ be a graph, let $f$ be a parameter, and let $\mathrm{DataStructure}$ be constructed as in \cref{alg:buildDataStructure}. Then, an edge $e \in G_{j}[V_i^{(j)}]$ if and only if $V_i^{(j)}$ is the component with the \emph{smallest} value of $j$ such that $e \subseteq V_i^{(j)}$.
\end{claim}

Now, for every edge $e \in F$, we do the following procedure:

\begin{enumerate}
    \item We let $V_i^{(j)}$ be the component with the smallest value of $j$ such that $u, v \in V_i^{(j)}$. 
    \item We decrement the degrees of $u, v$ in $V_i^{(j)}$ by $1$. 
    \item We set $\mathcal{S}_{u}^{(j)} \leftarrow \mathcal{S}_{u}^{(j)} - g((u,v))$, and similarly  $\mathcal{S}_{v}^{(j)} \leftarrow \mathcal{S}_{v}^{(j)} - g((u,v))$. 
\end{enumerate}

Now, the property we use is \cref{lem:robustnessfDeletion}: if, in an expander $V_i^{(j)}$, the degrees of $u, v$ remain sufficiently large, then we know that the distance between $u, v$ in $V_i^{(j)}$ is quite small (bounded by $O(\log(n)\log\log(n))$), and so for these high degree vertices, we can approximate their structure by a clique. However, if one of $u,v$ (or both) has degree in $V_i^{(j)}$ which is too small after deletions, then we instead invoke \cref{fact:uniqueSum}, and observe that for this low degree vertex, we can then recover \emph{all} of the remaining incident edges. 

Thus, after the deletions, our algorithm creates this auxiliary graph, where, within each expander, the remaining high degree vertices are connected by a clique, while for the low degree vertices, we instead recover their exact remaining neighborhoods. We compute the shortest path from $u$ to $v$ in this auxiliary graph, multiply it by $O(\log(n)\log\log(n))$, and claim that this satisfies the conditions of \cref{thm:faultTolerantDO}. We present the formal distance reporting algorithm in \cref{alg:reportDistance}.

\begin{algorithm}[h]
\caption{ReportDistance$(\mathrm{DataStructure}, a, b, F)$}\label{alg:reportDistance}
\For{$e = (u,v) \in F$}{
\If{exists component $V^{(j)}_i$ such that $e \subseteq V^{(j)}_i$}{
Let $V^{(j)}_i$ denote the first such component (i.e., for the smallest value of $j$). \\
$\mathrm{deg}_{G_{j}[V_i^{(j)}]}(v) \leftarrow \mathrm{deg}_{G_{j}[V_i^{(j)}]}(v) -1 $. \\
$\mathrm{deg}_{G_{j}[V_i^{(j)}]}(u) \leftarrow \mathrm{deg}_{G_{j}[V_i^{(j)}]}(u) -1 $. \\
$\mathcal{S}_v^{(j)} \leftarrow \mathcal{S}_v^{(j)} - g(e)$. \\
$\mathcal{S}_u^{(j)} \leftarrow \mathcal{S}_u^{(j)} - g(e)$. \\
}
\Else{
$X_{\ell} \leftarrow X_{\ell} - e$. \\
}
}
Let $H = (V, \emptyset)$. \\
\For{$V^{(j)}_i: j \in [\ell], i \in [k_j]$}{
Let $V'^{(j)}_i \subseteq V^{(j)}_i$ denote all vertices $u \in V^{(j)}_i$ such that $\mathrm{deg}_{G_{j}[V_i^{(j)}]}(u) \geq \frac{4f}{D} + 1$. \\
$H \leftarrow H + K_{V'^{(j)}_i}$. \label{line:addClique} \tcp{Add a clique on the high degree vertices.} 
\For{$u \in V^{(j)}_i - V'^{(j)}_i$}{
Let $\Gamma_{V^{(j)}_i}(u)$ denote the unique set of $\leq \frac{4f}{D}$ edges in $\binom{n}{2}$ such that $g(\Gamma_{V^{(j)}_i}(u)) = \mathcal{S}_u^{(j)}$. \label{line:efficientlyRecoverEdges}\\
$H \leftarrow H + \Gamma_{V^{(j)}_i}(u)$. \label{line:addFullNeighborhood}
}
}
$H \leftarrow H + X_{\ell}$. \label{line:finishH}\\
Let $\delta$ denote the distance between $a, b$ in $H$, and let $C''$ denote a large constant.  \\
\Return{$\delta \cdot C'' \cdot \log(n)\log\log(n)$.\label{line:returnDistance}}
\end{algorithm}

Now we prove the correctness of the above algorithm. To start, we have the following claim:

\begin{claim}\label{clm:Hcontains}
    Let $G = (V, E)$ be a graph, let $f$ be a parameter, and let $\mathrm{DataStructure}$ be the result of invoking \cref{alg:buildDataStructure} on $G, f$. Now, suppose that we call \cref{alg:reportDistance} with $\mathrm{DataStructure}$, a set of deletions $F$, and vertices $a, b$. Then, it must be the case that $G - F \subseteq H$, where $H$ is the resulting auxiliary graph whose construction terminates in \cref{line:finishH}.
\end{claim}

\begin{proof}
    Let us consider any edge $e \in G - F$. By definition, any such edge must have also been in the initial graph $G$ on which we constructed our data structure. Now, there are several cases for us to consider:
    \begin{enumerate}
        \item Suppose $e \notin G_{j}[V_i^{(j)}]$ for any $i, j$. Then, it must be the case that $e \in X_{\ell}$ (i.e., the final set of crossing edges). But in \cref{line:finishH}, we add all edges in $X_{\ell} - F$ to $H$, and so if $e \notin F$ and $e \in X_{\ell}$, then $e \in H$. 
        \item Otherwise, this means that $e \in G_{j}[V_i^{(j)}]$ for some $i, j$. Again, we have two cases:
        \begin{enumerate}
            \item Suppose that after performing the deletions from $F$ it is the case that $\mathrm{deg}_{G_{j}[V_i^{(j)}] - F}(u) \geq \frac{4f}{D} + 1$ and $\mathrm{deg}_{G_{j}[V_i^{(j)}] - F}(v) \geq \frac{4f}{D} + 1$. Then, this means that $u, v \in V'^{(j)}_i $, and so $e \in K_{V'^{(j)}_i}$. Because we add $K_{V'^{(j)}_i}$ to $H$, then $e \in H$. 
            \item Otherwise, at least one of $u,v $ (WLOG $u$) has $\mathrm{deg}_{G_{j}[V_i^{(j)}] - F}(u) \leq \frac{4f}{D}$. This means that $u$ has at most $\frac{4f}{D}$ neighbors in $G_{j}[V_i^{(j)}]$ after performing the deletions in $F$, and that in particular, $\mathcal{S}_u^{(j)}$ is the sum of $\leq \frac{4f}{D}$ many terms. By \cref{fact:uniqueSum}, there must then be a unique set of $\leq \frac{4f}{D}$ edges in $\binom{V}{2}$, denoted $\Gamma_{V^{(j)}_i}(u)$, such that $g(\Gamma_{V^{(j)}_i}(u)) = \mathcal{S}_u^{(j)}$. Because at this point \[
            \mathcal{S}_u^{(j)} = \sum_{e \in \Gamma_{G_{j}[V_i^{(j)}]}(u)} g(e) - \sum_{e \in \Gamma_{G_{j}[V_i^{(j)}]}(u) \cap F} g(e) = \sum_{e \in \Gamma_{G_{j}[V_i^{(j)}] - F}(u)} g(e),
            \]
            this means that this set of edges $\Gamma_{V^{(j)}_i}(u)$ that we recover are indeed all remaining incident edges on vertex $u$ in $G_{j}[V_i^{(j)}] - F$. Because $e \notin F$, it must also be the case then that $e \in \Gamma_{V^{(j)}_i}(u)$, and so is therefore also in $H$. 
        \end{enumerate}
    \end{enumerate}
    
    Thus, for any edge $e \in G - F$, we do indeed have that $e \in H$.
\end{proof}

As an immediate corollary to the above claim, we can also lower bound the true distance between any vertices in $G - F$:

\begin{corollary}\label{cor:distanceLowerBound}
    Let $G, H, F, f$ be defined as in \cref{clm:Hcontains}. Then, for any vertices $a, b \in V$,  $\mathrm{dist}_{H}(a,b) \leq \mathrm{dist}_{G - F}(a,b)$.
\end{corollary}

\begin{proof}
    This trivially follows because any path in $G - F$ between $a, b$ will also be a path between $a, b$ in $H$. 
\end{proof}

More subtly, we also have the following upper bound on the distance between any vertices in $G - F$:

\begin{claim}\label{clm:distanceUpperBound}
    Let $G, H, F, f$ be defined as in \cref{clm:Hcontains}. Then, for any vertices $a, b \in V$,  $\mathrm{dist}_{G - F}(a,b) \leq O(\log(n)\log\log(n)) \cdot \mathrm{dist}_{H}(a,b)$.
\end{claim}

\begin{proof}
    Let $a = v_0 \rightarrow v_1 \rightarrow v_2 \dots \rightarrow v_w \rightarrow b = v_{w+1}$ denote the shortest path between $a, b$ in $H$. Now consider any edge $v_{p}, v_{p+1}$ which is taken in this path. We claim that either: 
    \begin{enumerate}
        \item $(v_{p}, v_{p+1}) \in G - F$. 
        \item $v_{p}, v_{p+1}$ were vertices whose degree remained $\geq \frac{4f}{D} + 1$ in some expander component $G_{j}[V_i^{(j)}]$.
    \end{enumerate}

    This follows essentially by construction. There are three ways edges are added to $H$:
    \begin{enumerate}
        \item In \cref{line:finishH}, we add the final crossing edges $X_{\ell}$ which were not in $F$. Note that $X_{\ell} \subseteq G$ by construction, and so $X_{\ell} - F \subseteq G - F$.
        \item In \cref{line:addFullNeighborhood}, whenever a vertex has degree $\leq \frac{4f}{D}$ in an expander component $G_{j}[V_i^{(j)}] - F$, we recover its entire neighborhood and add those edges to $H$. Importantly, these are still edges that are a subset of $G - F$.
        \item Otherwise, in \cref{line:addClique}, for vertices whose degree remains $\geq \frac{4f}{D} + 1$ in some expander component $G_{j}[V_i^{(j)}]$ after performing deletions $F$, we add a clique on these vertices.
    \end{enumerate}

    So, either an edge in $H$ is in $G - F$, or the edge was added in \cref{line:addClique}, and corresponds to the remaining high degree vertices in some expander component $G_{j}[V_i^{(j)}]$, as we claimed.

    Finally then, to conclude, we claim that for any edge $(v_{p}, v_{p+1})$ which are vertices whose degree remained $\geq \frac{4f}{D} + 1$ in some expander component $G_{j}[V_i^{(j)}]$, it must be the case that 
    \[
    \mathrm{dist}_{G-F}(v_p, v_{p+1}) \leq O(\log(n)\log\log(n)).
    \]
    This follows exactly via \cref{lem:robustnessfDeletion}: indeed, starting with the graph $G_{j}[V_i^{(j)}]$ which has minimum degree $D$ (with $D \geq 8 \sqrt{f}\log(n)$), we performed $\leq f$ deletions. \cref{lem:robustnessfDeletion} shows exactly that, under these conditions, any vertices whose degree remains $\geq \frac{4f}{D} + 1$ must be connected by a path of length $O(\log(n)\log\log(n))$.

    Thus, for the path $a = v_0 \rightarrow v_1 \rightarrow v_2 \dots \rightarrow v_w \rightarrow b = v_{w+1}$ in $H$, we know that every edge used in the path is either present in $G - F$, or can be replaced by a path of length $O(\log(n)\log\log(n))$ in $G - F$. Thus, $\mathrm{dist}_{G - F}(a,b) \leq O(\log(n)\log\log(n)) \cdot \mathrm{dist}_{H}(a,b)$.
\end{proof}

Now, we prove \cref{thm:faultTolerantDO}.

\begin{proof}[Proof of \cref{thm:faultTolerantDO}.]
        Given a graph $G$ and parameter $f$, we instantiate the data structure using \cref{alg:buildDataStructure}. By \cref{lem:spaceBound}, we know that this data structure requires $\widetilde{O}(n \sqrt{f})$ bits of space to store. 
        
        When queried with vertices $u, v, F$ using \cref{alg:reportDistance}, we know that the data structure constructs a graph $H$ (depending on $F$) such that (by \cref{cor:distanceLowerBound} and \cref{clm:distanceUpperBound})
        \[
        \mathrm{dist}_{H}(u,v) \leq \mathrm{dist}_{G - F}(u,v) \leq O(\log(n)\log\log(n)) \cdot \mathrm{dist}_{H}(u,v).
        \]
        In particular, if we let $C''$ denote a sufficiently large constant such that 
        \[
        \mathrm{dist}_{H}(u,v) \leq \mathrm{dist}_{G - F}(u,v) \leq C'' \log(n)\log\log(n) \cdot \mathrm{dist}_{H}(u,v),
        \]
        then outputting $\widehat{\mathrm{dist}}_{G - F}(u,v) = C'' \log(n)\log\log(n) \cdot \mathrm{dist}_{H}(u,v)$ satisfies 
        \[
        \mathrm{dist}_{G - F}(u,v) \leq \widehat{\mathrm{dist}}_{G - F}(u,v) \leq C'' \log(n)\log\log(n) \cdot \mathrm{dist}_{G - F}(u,v),
        \]
        thereby satisfying the conditions of our theorem. 
\end{proof}

\subsection{Efficient Construction}

Note that \cref{thm:faultTolerantDO} can also be made efficient, paying only a logarithmic-factor blow-up in the stretch of the oracle:

\begin{theorem}\label{thm:faultTolerantDOEfficient}
    There is a data structure which, when instantiated on a graph $G = (V, E)$ on $n$ vertices, and any parameter $f \in \left [\binom{n}{2} \right ]$:
    \begin{enumerate}
    \item Can be constructed in polynomial time.
        \item Uses space at most $\widetilde{O}(n \sqrt{f})$ bits. 
        \item Supports a polynomial time query operation which, for any $F \subseteq E$ satisfying $|F| \leq f$, and any vertices $u, v \in V$, returns a value $\widehat{\mathrm{dist}}_{G - F}(u,v)$ such that 
        \[
        \mathrm{dist}_{G - F}(u,v)\leq \widehat{\mathrm{dist}}_{G - F}(u,v) \leq O(\log^3(n)) \cdot \mathrm{dist}_{G - F}(u,v).
        \]
    \end{enumerate}
\end{theorem}

\begin{proof}
    We use the same algorithm \cref{alg:buildDataStructure}, with a few minor modifications:
    \begin{enumerate}
        \item Instead of using $D = 8 \sqrt{f}\log(n)$, we use $D = 8 \sqrt{f}\log^2(n)$ in \cref{line:setDegree}.
        \item We use \cref{thm:efficientExpanderDecomp} instead of \cref{thm:expanderDecomp} in \cref{line:expanderDecomp}.
        \item In \cref{line:whileLoop}, we use $\log^4(n)$ instead of $\log^3(n)$.
        \item Instead of using the function $g$ from \cref{fact:uniqueSum} in \cref{line:uniqueSumFunction}, we instead use the function $g$ from \cref{clm:efficientUniqueSum}.
    \end{enumerate}

    Note that the space bound on the data structure now follows from exactly the same bound as \cref{lem:spaceBound} with the introduction of an extra logarithmic factor, and the fact that storing the information from \cref{clm:efficientUniqueSum} also requires only $O(\frac{f}{D}\log(n))$ bits of space per vertex.

    The polynomial time construction of the data structure follows from the fact that the expander decomposition now runs in polynomial time as per \cref{thm:efficientExpanderDecomp}, and that the encoding function $g$ also runs in polynomial time as per \cref{clm:efficientUniqueSum}.

    For reporting distances, we again use \cref{alg:reportDistance}, with a few minor modifications:
    \begin{enumerate}
        \item In \cref{line:efficientlyRecoverEdges}, we now use the polynomial time recovery algorithm of \cref{clm:efficientUniqueSum} (and likewise, for updating the sketch when we delete edges).
        \item In \cref{line:returnDistance}, we instead multiply by $\log^3(n)$.
    \end{enumerate}

    The correctness of the algorithm then follows from exactly the same argument as in \cref{clm:distanceUpperBound}, with the only modification being that we now use \cref{lem:robustnessfDeletionWorseExpansion} instead of \cref{lem:robustnessfDeletion}, which leads to a diameter of $O(\log^3(n))$ instead of $O(\log(n)\log\log(n))$. This yields the theorem. 
\end{proof}

\subsection{Generalizing to Weighted Graphs}\label{sec:weighted}

As another generalization of \cref{thm:faultTolerantDO}, we show that our approach can be extended to create fault-tolerant distance oracles on weighted graphs (provided that the deletions \emph{include the weight} of the deleted edge):

\begin{theorem}\label{thm:faultTolerantDOWeighted}
    There is a data structure which, when instantiated on a graph $G = (V, E, w)$ on $n$ vertices, and any parameter $f \in \left [\binom{n}{2} \right ]$:
    \begin{enumerate}
        \item Uses space at most $\widetilde{O}(n \sqrt{f} \cdot \log(w_{\max}/ w_{\min}))$ bits (where $w_{\max}$ is the largest weight in $w$, and $w_{\min}$ is the smallest non-zero weight).
        \item Supports a query operation which, for any $F \subseteq \{(e, w_e): e \in E\}$ satisfying $|F| \leq f$, and any vertices $u, v \in V$, returns a value $\widehat{\mathrm{dist}}_{G - F}(u,v)$ such that 
        \[
        \mathrm{dist}_{G - F}(u,v)\leq \widehat{\mathrm{dist}}_{G - F}(u,v) \leq O(\log(n)\log\log(n)) \cdot \mathrm{dist}_{G - F}(u,v).
        \]
    \end{enumerate}
\end{theorem}

\begin{proof}
    The algorithm is a simple modification to \cref{alg:buildDataStructure}. WLOG, we assume the minimum edge weight is $1$, and we create graphs $G_1, G_2, G_4, \dots G_{2^{\lfloor \log(w_{\max})\rfloor}}$, where $G_i$ contains only those edges of weight $[i, 2 \cdot i)$. Importantly, in $G_i$ however, we treat the edges as being \emph{unweighted}.

    Now, for each $G_i$, we construct the distance oracle as in \cref{alg:buildDataStructure}. By \cref{thm:faultTolerantDO}, this yields the space bound of $\widetilde{O}(n \sqrt{f} \cdot \log(w_{\max}/ w_{\min}))$ bits, as we have $\log(w_{\max}/ w_{\min}))$ many copies of the sketch from \cref{thm:faultTolerantDO}.

    Next, when performing deletions, for each edge $e$ and weight $w_e$ in $F$, we perform the deletion from the corresponding graph $G_i$ such that $w_e \in [i, 2i )$. Then, using \cref{thm:faultTolerantDO}, we recover the graph $H_i$ for $G_i$ such that, for any vertices $u, v \in V$, it is the case that 
    \begin{align}\label{eq:compareDistances}
    \Omega(\mathrm{dist}_{G_i - F \cap G_i}(u,v) / \log^2(n)) \leq \mathrm{dist}_{H_i}(u,v)) \leq \mathrm{dist}_{G_i - F \cap G_i}(u,v),
    \end{align}
    as guaranteed by \cref{clm:distanceUpperBound} and \cref{cor:distanceLowerBound}.

    Finally, we now create the graph $H = H_1 \cup 2 \cdot H_2 \cup 4 \cdot H_4 \cup \dots \cup 2^{\lfloor \log(w_{\max})\rfloor} \cdot H_{2^{\lfloor \log(w_{\max})\rfloor}}$, i.e., where we weight the edges in $H_i$ by a factor of $i$. Now, we first claim that for any vertices $u, v$,
    \[
    \mathrm{dist}_H(u,v) \leq \mathrm{dist}_{G - F}(u,v).
    \]
    Indeed, consider any edge used in the path between $u,v$ in $G$, and let us denote this edge by $(v_1, v_2)$. Suppose the weight of this edge is $w_{v_1, v_2}$, and therefore is in graph $G_i - F \cap G_i$. Then, we know that $\mathrm{dist}_{G_i - F \cap G_i}(v_1,v_2) = 1$, and by the above, then $\mathrm{dist}_{H_i}(v_1,v_2) = 1$ as well. Now, when we increase the weight of edges by $i$, we see that $\mathrm{dist}_H(v_1, v_2) \leq i$, as the edges in graph $H_i$ are given weight $i$. At the same time, because $(v_1, v_2)$ is in $G_i$, $w_{v_1, v_2} \geq i$, it is the case that using the edge $(v_1, v_2)$ in $H$ contributes less weight to the path than in $G - F$.

    Next, we show that 
    \[
    \mathrm{dist}_{G - F}(u,v) \leq O(\log(n)\log\log(n)) \cdot \mathrm{dist}_H(u,v).
    \]

    Again, we consider the shortest path between $u, v$, and suppose that it uses an edge $(v_1, v_2) \in H$. Suppose further that $(v_1, v_2) \in H_i$. Then, it must be the case that $v_1, v_2$ are connected by a path of length $\leq O(\log(n)\log\log(n))$ in $G_i - F \cap G_i$ by \cref{eq:compareDistances}. Because every edge in $G_i - F \cap G_i$ has weight $\leq 2i$, the length of this path in $G$ is bounded by $O(i \log(n)\log\log(n))$. Likewise, the edge $(v_1, v_2)$ has length $\geq i$ in $H$, and so we indeed see that the contribution to the distance from the segment $v_1, v_2$ in $G-F$ is at most $O(\log(n)\log\log(n))$ times that of the contribution in $H$. 

    Together, the above imply that
    \[
    \Omega(\mathrm{dist}_{G-F}(u,v) / \log(n)\log\log(n)) \leq \mathrm{dist}_{H}(u,v)) \leq \mathrm{dist}_{G-F}(u,v),
    \]
    and so by computing the shortest paths in $H$, and then scaling the reported distances by $O(\log(n)\log\log(n))$, we obtain the stated claim. 
\end{proof}

\section{Achieving Constant Stretch via Covering by $2$-Hop Stars}\label{sec:constantStretch}

In this section, we show how to achieve \emph{constant} stretch (in fact, the stretch will be $7$) while storing $o(nf)$ bits of space. In particular, we prove the following (which is a restatement of \cref{thm:constantStretchIntro}):

\begin{theorem}\label{thm:constantStretch}
    There is an $f$-fault tolerant deterministic distance oracle with stretch $7$, which, for $n \leq f \leq n^{3/2}$, uses $\widetilde{O}(n^{3/2}f^{1/3})$ bits of space.
\end{theorem}

Note that once $f$ is larger than $n^{3/2}$ the space complexity is $O(n^2)$ which is trivial. Otherwise, for instance, when $f = n$, we achieve non-trivial space complexity:

\begin{corollary}
    There is an $n$-fault tolerant distance oracle with stretch $7$, which uses $\widetilde{O}(n^{11/6})$ bits of space.
\end{corollary}

\begin{remark}
    For comparison, any $f$-fault tolerant spanner for $f \in [n, n^{3/2}]$ requires $\Omega(n^2)$ edges (and thus bits) to store. \cref{thm:constantStretch} shows that in the setting of deterministic fault-tolerant distance oracles \emph{even with  constant-stretch}, an analogous lower bound does not hold, and we can in fact report distances in $o(n^2)$ bits of space.
\end{remark}

\subsection{Constructing High Degree $2$-Hop Stars}

To start, we present a simple graph structure with high minimum degree and small diameter, which we refer to as a \emph{$2$-hop star}:

\begin{definition}
    Given a graph $G = (V, E)$, a $2$-hop star for a vertex $v \in V$, is a set of vertices $\{v, L_1, L_2\}$, where $L_1, L_2 \subseteq V$, $L_1 \cap L_2 = \emptyset$, $L_1 \subseteq \Gamma(v)$, and $L_2 \subseteq \Gamma(L_1)$. That is to say, $L_1$ is a subset of the neighbors of $v$, and $L_2$ is a subset of the neighbors of $L_1$. A $2$-hop star consists only of the edges $E[\{v\}, L_1]$ and $E[L_1, L_2]$. Importantly, it does not contain the edges that \emph{are entirely contained inside} $L_1$ or $L_2$.
\end{definition}

We will also use $1$-hop stars:

\begin{definition}
    Given a graph $G = (V, E)$, a $1$-hop star for a vertex $v \in V$ is a set of vertices $v, L_1$ such that $L_1 \subseteq \Gamma(v)$. 
\end{definition}

With this, we can state our first lemma:

\begin{lemma}\label{lem:highDegreeLowDiam}
    Let $G$ be a graph on $n$ vertices with minimum degree $d$. Then, for any vertex $v \in V$, one can find either:
    \begin{enumerate}
        \item A $1$-hop star $S = \{v, L_1\}$ such that $G[S]$ has minimum degree $\geq d / 10$.
        \item A $2$-hop star $S = \{v, L_1, L_2\}$ such that $|L_1| \geq \frac{d^2}{2n}$ and the induced bipartite graph between $L_1$ and $L_2$ has minimum degree $\geq \frac{d^2}{2 n}$.
    \end{enumerate}
\end{lemma}

Before proving this, we first recall a some basic fact about graphs:

\begin{claim}\label{clm:edgesToMinDegree}
    Let $G$ be a graph on $n$ vertices with $n \cdot d'$ edges. Then, there exists an induced subgraph of $G$ with minimum degree $\geq d'$.
\end{claim}

\begin{proof}
    Indeed, starting with the original, entire graph $G$, while the minimum degree is $< d'$, we remove a vertex of degree $< d'$. We continue repeating this procedure until all remaining vertices have degree $\geq d'$. Observe that each time a vertex is removed, we lose $< d'$ many edges from the graph. Thus, across $\leq n$ peels, it is only possible to peel $< n d'$ many edges from the graph, meaning that the resulting graph must be non-empty (and thus, it does indeed have minimum degree $\geq d'$).
\end{proof}

Now, we proceed to a proof of \cref{lem:highDegreeLowDiam}: 

\begin{proof}[Proof of \cref{lem:highDegreeLowDiam}.]
    Indeed, let $G$ be a graph with minimum degree $d$, and let $v \in V$ be arbitrary. Now, to start, let us consider the set of all neighbors of $v$, i.e., $\Gamma(v)$. We proceed by cases:
    \begin{enumerate}
        \item Suppose that $G[\Gamma(v) ]$ contains $\geq (|\Gamma(v)|) \cdot d / 10$ many edges. Then, by \cref{clm:edgesToMinDegree}, we know that there is some subgraph $L_1 \subseteq \Gamma(v)$ such that $G[L_1]$ has minimum degree $\geq d/10$. Because $G[L_1]$ has minimum degree $\geq d/10$, we also know that $|L_1| \geq d/10$, and thus the graph $G[\{v \} \cup L_1]$ also has minimum degree $\geq d/10$. This is because $v$ has an edge to every vertex in $L_1$ (as $L_1 \subseteq \Gamma(v)$), while the other vertices satisfy the minimum degree requirement by virtue of $G[L_1]$'s minimum degree.
        \item Otherwise, suppose that $G[\Gamma(v) ]$ contains $< (|\Gamma(v)|) \cdot d / 10$ many edges. This instead means that among the edges leaving vertices in $\Gamma(v)$, at least $(d - d/10 -1) \cdot |\Gamma(v)|$ (the $1$ is for the edge to $v$) many edges go to $V \setminus (\{v \} \cup \Gamma(v))$. Let us use $V'$ to denote these vertices in $\Gamma(\Gamma(v))$ which are not $v$ or $\Gamma(v)$. In particular, we can build a bipartite graph between $\Gamma(v)$ and $V'$ with the promise that this bipartite graph has at least $|\Gamma(v)| \cdot (d - d/10 -1) \geq d^2 / 2$ many edges, and $\leq n$ many vertices. By \cref{clm:edgesToMinDegree}, there exists an induced subgraph of this bipartite graph with minimum degree $\geq \frac{d^2}{2n}$. Let us denote the left and right sides of this induced bipartite subgraph by $L_1$ and $L_2$.

        Importantly, because the minimum degree of $L_2$ is $\geq \frac{d^2}{2n}$, this implies that $|L_1| \geq \frac{d^2}{2n}$, which in turn implies that $v$ has $\geq \frac{d^2}{2n}$ many edges to $L_1$. Thus, the induced subgraph on $S = \{v, L_1, L_2\}$ has minimum degree $\geq \frac{d^2}{2n}$, and likewise, the induced bipartite graph between $L_1$ and $L_2$ has minimum degree $\geq \frac{d^2}{2 n}$.
    \end{enumerate}

    Because the above cases are collectively exhaustive, this yields our desired lemma.
\end{proof}

In fact, the above claim is even implementable in polynomial time, as it relies only on iteratively peeling vertices and computing degrees in induced subgraphs:

\begin{claim}
     Let $G$ be a graph on $n$ vertices with minimum degree $d$. For any vertex $v \in V$, one can find the promised induced subgraphs of \cref{lem:highDegreeLowDiam} in polynomial time.
\end{claim}

\subsection{Robustness of High Degree Stars Under Edge Deletions} 

As we shall see, the utility of these high degree stars is that they ensure that even under (potentially a significant number) of edge deletions, vertices with mildly high degree still remain connected. We start by proving this in the setting of $1$-hop stars, and then generalize this to $2$-hop stars. Before we begin, we also remark that we will use the term \textbf{root} to refer to the vertex $v$ from which the resulting star is grown. 

\subsubsection{Robustness of High Degree $1$-Hop Stars Under Edge Deletions}

Our starting claim is the following:

\begin{claim}\label{clm:highDegreeConnectRoot}
    Let $H$ be a $1$-hop star with root $v$, and vertex set $\{v, L_1\}$ such that $H$ has minimum degree $d$. Now, consider a set $F$ of $\leq f$ edge deletions such that at most $d/2$ of $v$'s edges are deleted. Then, for any vertex $u$ with $\mathrm{deg}_{H - F}(u) \geq \frac{5f}{d}$, $u$ has a path of length $\leq 3$ to $v$ in $H - F$.
\end{claim}

\begin{proof}
    Notationally, we let $\mathrm{BAD}$ denote the set of all vertices in $L_1$ which do not have an edge to $v$ in $H - F$. From the above, we know that $|\mathrm{BAD}| \leq \frac{d}{2}$.
    
    Now, consider such a vertex $u$ in $H - F$.    
    To start, if $u$ has a direct edge to $v$ (i.e., $u$ is not in $\mathrm{BAD}$), then we are immediately done. Otherwise, we know that $u$ has $\geq \frac{5f}{d}$ many other neighbors in $L_1$, which we denote by $\Gamma'(u)$. Thus, in $H$ (before deletions) the total sum of degrees of vertices in $\Gamma'(u)$ is at least $5f$. Because $F$ only constitutes $f$ deletions, we know that 
    \[
    \sum_{w \in \Gamma'(u)} \mathrm{deg}_{H - F}(w) \geq  \left ( \sum_{w \in \Gamma'(u)} \mathrm{deg}_{H }(w) \right ) - 2f > \frac{\left ( \sum_{w \in \Gamma'(u)} \mathrm{deg}_{H }(w) \right )}{2}.
    \]
    Thus, there must be some vertex $w \in \Gamma'(u)$ such that $\mathrm{deg}_{H - F}(w) > \mathrm{deg}_{H}(w)/2$. If $w$ has an edge directly to $v$, we are of course immediately done. Otherwise, this implies that $w$ has $> d/2$ many neighbors in $H - F$; at the same time, we know that $|\mathrm{BAD}| \leq \frac{d}{2}$. This means there must be some vertex $\ell \in \Gamma_{H - F}(w)$ such that $\ell \notin \mathrm{BAD}$. So, we recover a path from $u \rightarrow w \rightarrow \ell \rightarrow v$ of length $3$.
\end{proof}

As a corollary of the above claim, we have the following:

\begin{corollary}
    Let $H$ be a $1$-hop star with root $v$, and vertex set $\{v, L_1\}$ such that $H$ has minimum degree $d$. Now, consider a set $F$ of $\leq f$ edge deletions such that at most $d/2$ of $v$'s edges are deleted. Then, for any edge $(u,w) \in H$, either:
    \begin{enumerate}
    \item At least one of $u, w$ has degree $< \frac{5f}{d}$ in $H - F$. 
    \item $u, w$ remain connected by a path of length $\leq 6$ in $H - F$.
    \end{enumerate}
\end{corollary}

\begin{proof}
    Suppose that $u, w$ both have degree $\geq \frac{5f}{d}$. Then, by \cref{clm:highDegreeConnectRoot}, both $u$ and $w$ have paths to $v$ in $H - F$ of length $\leq 3$. Composing these paths yields a path of length $\leq 6$ from $u$ to $w$ in $H - F$, as we desire.
\end{proof}

\subsubsection{Robustness of High Degree $2$-Hop Stars Under Edge Deletions}

In fact, the reasoning from the above subsection generalizes in large part to $2$-hop stars. In this section, we prove the following:

\begin{claim}\label{clm:highDegreeConnectRootDepth2}
    Let $H$ be a $2$-hop star with root $v$, and vertex set $\{v, L_1, L_2\}$ that has minimum degree $d$. Now, consider a set $F$ of $\leq f$ edge deletions such that at most $d/2$ of $v$'s edges are deleted, and assume $\frac{d}{2} \geq \frac{5f}{d}$. Then, for any vertex $u$ with $\mathrm{deg}_{H - F}(u) \geq \frac{5f}{d}$, if $u$ is in $L_1$, $u$ has a path of length $\leq 3$ to $v$ in $H - F$, and if $u$ is in $L_2$, $u$ has a path of length $\leq 4$ to $v$ in $H - F$.
\end{claim}

\begin{proof}[Proof of \cref{clm:highDegreeConnectRootDepth2}]
    We proceed in two parts. We first prove the above claim for vertices in $L_1$, and then prove the above claim for vertices in $L_2$. As in the proof of \cref{clm:highDegreeConnectRoot}, we let $\mathrm{BAD}$ denote the set of all vertices in $L_1$ which do not have an edge to $v$ in $H - F$.

    \begin{enumerate}
        \item Let $u \in L_1$ such that $u$'s degree in $H - F$ is $\geq \frac{5f}{d}$. If $u \notin \mathrm{BAD}$, then we are immediately done. Otherwise, we know that $u$ must have at least $\frac{5f}{d}$ many neighbors in $L_2$ in $H - F$. We denote these neighbors of $u$ by $\Gamma'(u)$. Thus, in $H$ (before deletions) the total sum of degrees of vertices in $\Gamma'(u)$ is at least $5f$. Because $F$ only constitutes $f$ deletions, we know that 
    \[
    \sum_{w \in \Gamma'(u)} \mathrm{deg}_{H - F}(w) \geq  \left ( \sum_{w \in \Gamma'(u)} \mathrm{deg}_{H }(w) \right ) - 2f > \frac{\left ( \sum_{w \in \Gamma'(u)} \mathrm{deg}_{H }(w) \right )}{2}.
    \]
    Thus, there must be some vertex $w \in \Gamma'(u)$ such that $\mathrm{deg}_{H - F}(w) > \mathrm{deg}_{H}(w)/2$. For this vertex $w \in L_2$, $w$ has $> d/2$ many neighbors in $H - F$; at the same time, we know that $|\mathrm{BAD}| \leq \frac{d}{2}$. This means there must be some vertex $\ell \in \Gamma_{H - F}(w)$ such that $\ell \notin \mathrm{BAD}$. So, we recover a path from $u \rightarrow w \rightarrow \ell \rightarrow v$ of length $3$, concluding this case.
    \item Let $u \in L_2$ such that $u$'s degree in $H - F$ is $\geq \frac{5f}{d}$. This means that $u$ has at least $\frac{5f}{d}$ many neighbors in $L_1$ in the graph $H - F$. We denote these neighbors of $u$ by $\Gamma'(u)$. We can now repeat the same reasoning as above: in $H$ (before deletions) the total sum of degrees of vertices in $\Gamma'(u)$ is at least $5f$. Because $F$ is only $n$ deletions, we know that 
    \[
    \sum_{w \in \Gamma'(u)} \mathrm{deg}_{H - F}(w) \geq  \left ( \sum_{w \in \Gamma'(u)} \mathrm{deg}_{H }(w) \right ) - 2f > \frac{\left ( \sum_{w \in \Gamma'(u)} \mathrm{deg}_{H }(w) \right )}{2},
    \]
    and so there is some vertex in $\Gamma'(u)$ which undergoes $\leq d/2$ many deletions. Thus, there is some vertex $w$ in $\Gamma'(u)$ which has degree $\geq d/2 \geq \frac{5f}{d}$. By case 1, we know that $w$ therefore has a path of length $\leq 3$ to $v$, and thus $u$ must have a path of length $\leq 4$ to $v$, as we desire.
    \end{enumerate}

    This concludes the claim, as we have shown it to be true for vertices in $L_1$ and $L_2$.
\end{proof}

As before, this claim also immediately gives us the following corollary:

\begin{corollary}\label{cor:nDeletionsStretch8}
    Let $H$ be a $2$-hop star with root $v$, and vertex set $\{v, L_1, L_2\}$ that has minimum degree $d$. Now, consider a set $F$ of $\leq f$ edge deletions such that at most $d/2$ of $v$'s edges are deleted, and assume $\frac{d}{2} \geq \frac{5f}{d}$. Then, for any edge $(u,w) \in H$, either:
    \begin{enumerate}
    \item At least one of $u, w$ has degree $< \frac{5f}{d}$ in $H - F$. 
    \item $u, w$ remain connected by a path of length $\leq 7$ in $H - F$.
    \end{enumerate}
\end{corollary}

\begin{proof}

    Suppose that $u, w$ both have degree $\geq \frac{5f}{d}$. Because $(u,w)$ is an edge in a $2$-hop star, it must be that one of $u,w$ is in $L_1$ and the other is in $L_2$.
    
    Then, by \cref{clm:highDegreeConnectRootDepth2}, one of $u,w$ has a path to $v$ in $H - F$ of length $\leq 4$ and the other a path of length $\leq 3$. Composing these paths yields a path of length $\leq 7$ from $u$ to $w$ in $H - F$, as we desire.
\end{proof}

\subsection{Building the Distance Oracle}

With the above claims established, we now detail the construction of our distance oracle. In particular, there is one key piece of intuition that need: the above claims show that for an edge $e = (u,w)$, the endpoints $u$ and $w$ are ``robustly connected'' under deletions, provided that $u, w$ retain a non-negligible degree, and, more importantly, that \emph{the root does not experience too many deletions}. As we shall see, this latter condition ends up being the more difficult condition to ensure when constructing our distance oracles. In fact, in order to ensure this condition holds, we will store \emph{extra stars}, each initialized with different roots, thus ensuring that for any set of deletions there is some star whose root does not exceed the deletion budget.

\subsubsection{Some Basic Subroutines}

With this intuition in mind, we now begin presenting some basic sub-routines that we will use in our algorithm. First, we use an algorithm called $\mathrm{Peel}$ to ensure that our starting graph has a large minimum degree:

\begin{algorithm}[H]
\caption{Peel$(G = (V, E), d)$}\label{alg:Peel}
$V' = V$. \\
\While{$\mathrm{mindeg}(G[V']) < d$}{
Let $u$ be a vertex of degree $< d$. \\
$V' \leftarrow V' \setminus \{ v\}$.
}
\Return{$G[V']$.}
\end{algorithm}

\begin{claim}
    When initialized on a graph $G = (V, E)$, \cref{alg:Peel} returns an induced subgraph $G[V']$ which is either empty or has minimum degree $\geq d$.
\end{claim}

\begin{proof}
    If the algorithm stops before peeling all vertices off, the resulting graph must have minimum degree $\geq d$.
\end{proof}

\begin{definition}
    For a graph $G$ with minimum degree $d$, we will also use $\mathrm{ConstructStar}(G, v, d)$ to refer to the algorithm of \cref{lem:highDegreeLowDiam} for building our stars. Importantly, we assume that the output of $\mathrm{ConstructStar}$ is a list of vertices \emph{not including the edges}. This is essential for getting our ultimate space savings, as we forego storing edges explicitly.
\end{definition}

Often, we will invoke \cref{clm:efficientUniqueSum} on the neighborhood of a single vertex:

\begin{corollary}\label{cor:recoverNeighborsLowDegree}
    When invoked on the neighborhood of a single vertex (in a graph with $n$ vertices) with sparsity threshold $k$, \cref{clm:efficientUniqueSum} requires $O(k \log(n))$ many bits to store, and ensures that if all but $\leq k$ edges are deleted, then the remaining edges can be recovered exactly. 
\end{corollary}

We will use $\mathrm{SparseRecovery}(E, k)$ to denote initializing the sketch of \cref{cor:recoverNeighborsLowDegree} on a set of edges $E$ with sparsity threshold $k$. 

With this, we can now proceed to the statement of our overall algorithm. 

\subsubsection{Formal Algorithm}

Our algorithm takes in a graph $G = (V, E)$ and returns a fault-tolerant distance oracle data structure with stretch $7$. The intuition behind the algorithm is as follows:

\begin{enumerate}
    \item For each edge $e = (u,w) \in E$, we want to find multiple high-degree stars that all contain this edge. In particular, every star we construct will have minimum degree $\approx f^{2/3}$, and thus after a sequence of $f$ deletions, there can be at most $\approx f^{1/3}$ many stars for which the roots all receive $\geq f^{2/3}$ many deletions. Thus, if we store just \emph{slightly more} than $f^{1/3}$ many stars that contain this edge $e$, then we are closer to making the presence of this edge $e$ robust to edge deletions.
    \item So, suppose that the edge $e$ participates in $\approx f^{1/3}$ many stars and that some of these stars even have roots that remain high degree after deletions. \emph{It is still possible} that the vertices $u, w$ become low degree, and thus the star does not necessarily witness a short path between $u$ and $w$. This shortcoming is essentially unavoidable; indeed, one can consider the case where all of the incident edges on $u, w$ except $(u,w)$ are deleted. To address this case, we add \emph{sparse recovery sketches} which ensure that if the degrees of $u$ or $w$ fall too low, then any remaining edges are recovered explicitly. 
\end{enumerate}

\begin{algorithm}[H]
    \caption{ConstructOracle$(G = (V, E))$}\label{alg:ConstructOracle}
    $\mathrm{TargetDegree} = n^{1/2}f^{1/3} \log(n)$. \\
    $\mathrm{CoveringCounts}: E \rightarrow \Z$, initialized to all be $0$. \\
    $\mathrm{CoveringThreshold} = 10 \cdot f^{1/3}$. \\
    $\mathrm{CoveredEdges} = \emptyset$. \\
    $\mathrm{UnusedVertices} = V$. \\
    $\mathrm{Stars} = \emptyset$. \\
    $\mathrm{StarDegrees} = \emptyset$. \\
    $\mathrm{SparseRecoverySketches} = \emptyset$. \\
    \For{$d \in \{n/2, n/4, n/8, \dots \mathrm{TargetDegree}\}$}{
    \label{line:iterateDegree}
    \While{$\mathrm{Peel}(G - \{e \in E: \mathrm{CoveringCounts}(e) \geq \mathrm{CoveringThreshold}\}, d) \neq \emptyset$\label{line:whileCondition}}{
    $G_d = \mathrm{Peel}(G - \{e \in E: \mathrm{CoveringCounts}(e) \geq \mathrm{CoveringThreshold}\}, d)$. \\
    Let $v \in V(G_d) \cap \mathrm{UnusedVertices}$. \label{line:newVertex}\\
    Let $\mathrm{Star} = \mathrm{ConstructStar}(G_d, v, d)$. \\
    $\mathrm{Stars} \leftarrow \mathrm{Stars} \cup \mathrm{Star}$. \label{line:storeStar} \\
    $\mathrm{UnusedVertices} \leftarrow \mathrm{UnusedVertices} \setminus \{v\}$. \\
    \For{$e \in G_d[\mathrm{Star}]$}{
    $\mathrm{CoveringCounts}(e) \leftarrow \mathrm{CoveringCounts}(e) + 1$. \\
    }
    $\mathrm{SparseRecoveryStar} = \emptyset$. \\
    $\mathrm{DegreesSingleStar} = \emptyset$. \\
    \For{$u \in \mathrm{Star}$}{
    $\mathrm{SparseRecoveryStar} \leftarrow \mathrm{SparseRecoveryStar} \cup \mathrm{SparseRecovery}(\Gamma_{G_d[\mathrm{Star}]}(u), 100 \cdot f^{1/3})$. \label{line:storeSparseRecovery} \\
    $\mathrm{deg}_{\mathrm{Star}}(u) = \mathrm{deg}_{G_d[\mathrm{Star}]}(u)$. \\
    $\mathrm{DegreesSingleStar} \leftarrow \mathrm{DegreesSingleStar} \cup \mathrm{deg}_{\mathrm{Star}}(u)$.\label{line:storeDegrees} \\
    }
    $\mathrm{SparseRecoverySketches} \leftarrow \mathrm{SparseRecoverySketches} \cup \mathrm{SparseRecoveryStar}$. \\
    $\mathrm{StarDegrees} \leftarrow \mathrm{StarDegrees} \cup \mathrm{DegreesSingleStar}$. \\
    }
    }
    Let $G_{\mathrm{remaining}} = G - \{e \in E: \mathrm{CoveringCounts}(e) \geq \mathrm{CoveringThreshold}\}$. \\
    \Return{$\bigg ( G_{\mathrm{remaining}}, \mathrm{Stars}, \mathrm{SparseRecoverySketches}, \mathrm{DegreesSingleStar} \bigg )$.\label{line:returnOutput}}
\end{algorithm}

With the algorithm stated, we now proceed to bound the space of the returned data structure and prove that it does indeed suffice as a distance oracle. 

\subsubsection{Bounding the Space}

To start, we bound the space required to store each star:

\begin{claim}\label{clm:boundStarSpace}
    For a star $H$ on $\ell$ many vertices, \cref{alg:ConstructOracle} stores $\widetilde{O}(\ell f^{1/3})$  bits.
\end{claim}

\begin{proof}
    Indeed, the space consumption comes from three places:
    \begin{enumerate}
        \item In \cref{line:storeStar}, \cref{alg:ConstructOracle} stores the exact description of the vertex set of the star (i.e., $v, L_1, L_2$), which requires $O(\ell \log(n))$ many bits of space. 
        \item In \cref{line:storeDegrees}, \cref{alg:ConstructOracle} stores a list of the degrees of the vertices in the star which requires $O(\ell \log(n))$ bits of space.
        \item In \cref{line:storeSparseRecovery}, \cref{alg:ConstructOracle} stores a sparse recovery sketch for each vertex in the star, with a sparsity threshold of $O(f^{1/3})$. By \cref{cor:recoverNeighborsLowDegree}, this requires $\widetilde{O}(f^{1/3})$ many bits per vertex in the star, yielding a total of $\widetilde{O}(\ell f^{1/3})$  bits.
    \end{enumerate}
    This yields the claim. 
\end{proof}

With this claim in hand, we can prove the following key lemma:

\begin{lemma}\label{lem:boundSpaceSingleValued}
    Consider all the stars which are stored for a single value of $d$ in \cref{line:iterateDegree} in \cref{alg:ConstructOracle}. The total space complexity of these stored stars (including their sparse recovery sketches) is $\widetilde{O} \left ( \frac{n^{2}f^{2/3}}{d} \right )$ bits.
\end{lemma}

Note that this lemma relies on a key observation:

\begin{fact}
    Once an edge has been used in $\geq \mathrm{CoveringThreshold} = 10 f^{1/3}$ many stars, it is effectively \emph{removed} from the graph, and not used in any future stars.
\end{fact}

\begin{proof}[Proof of \cref{lem:boundSpaceSingleValued}]
    Indeed, consider a fixed value of $d$. Note that, by \cref{line:whileCondition}, it must be the case that $\mathrm{Peel}(G - \{e \in E: \mathrm{CoveringCounts}(e) \geq \mathrm{CoveringThreshold}\}, 2d) = \emptyset$, as otherwise the algorithm would not yet have progressed to a current degree of $d$. Importantly, the fact that $\mathrm{Peel}(G - \{e \in E: \mathrm{CoveringCounts}(e) \geq \mathrm{CoveringThreshold}\}, 2d) = \emptyset$ implies that at the start of the iteration with value $d$, $G - \{e \in E: \mathrm{CoveringCounts}(e) \geq \mathrm{CoveringThreshold}\}$ has $\leq 2n d$ many edges, as otherwise iteratively peeling vertices of degree $\leq 2d$ would not yield an empty graph.
    
    Now, for each of the $\leq 2nd$ many edges in the graph, each one can be covered at most $10f^{1/3}$ many times before it is removed from the graph. In total then, there are at most $10f^{1/3} \cdot 2nd$ many edge slots (with multiplicity) to be covered by stars. Because the minimum degree in the graph is $d$, each star we recover on $\ell$ vertices covers $\geq \ell \cdot \frac{d^2}{2n}$ many edges (using that the minimum degree in each star is $\frac{d^2}{2n}$ by \cref{lem:highDegreeLowDiam}). Thus, we see that across all the stars we recover in this iteration, the number of vertices present in the stars we recover is bounded by 
    \[
    \frac{10 f^{1/3} \cdot 2nd}{\frac{d^2}{2n}} \leq \frac{40n^{2}f^{1/3}}{d}.
    \]
    In conjunction with \cref{clm:boundStarSpace}, we then see that the total amount of space dedicated to storing stars in the iteration with minimum degree $d$ is bounded by 
    \[
 \widetilde{O} \left ( \frac{n^{2}f^{1/3}}{d} \cdot f^{1/3}\right ) =  \widetilde{O} \left ( \frac{n^{2}f^{2/3}}{d} \right )
    \]
    bits, as we pay $\widetilde{O}(f^{1/3})$ bits for each vertex in a star we store.
\end{proof}

With this, we can now conclude the total space usage of \cref{alg:ConstructOracle}:

\begin{lemma}\label{lem:finalConstantStretchSpaceBound}
    When invoked on a graph $G$ on $n$ vertices, \cref{alg:ConstructOracle} requires $\widetilde{O}(n^{3/2}f^{1/3})$ bits of space. 
\end{lemma}

\begin{proof}
    Indeed, as per \cref{lem:boundSpaceSingleValued}, for each choice of $d$ in \cref{line:iterateDegree} in \cref{alg:ConstructOracle}, the total space complexity of these stored stars (including their sparse recovery sketches) is $ \widetilde{O} \left ( \frac{n^{2}f^{2/3}}{d} \right )$ bits. Thus across all values of $d \in \{n/2, n/4, n/8, \dots \mathrm{TargetDegree} \}$ (with $\mathrm{TargetDegree} = n^{1/2}f^{1/3}\log(n)$) the total space consumption is
    \[
    \leq \sum_{d \in \{n/2, n/4, n/8, \dots n^{1/2}f^{1/3}\log(n) \}} \widetilde{O} \left ( \frac{n^{2}f^{2/3}}{d} \right ) =  \widetilde{O} \left ( \frac{n^{2}f^{2/3}}{n^{1/2}f^{1/3}\log(n)} \right ) = \widetilde{O}(n^{3/2}f^{1/3})
    \]
    bits.

    Finally, in \cref{line:returnOutput}, \cref{alg:ConstructOracle} also stores all remaining uncovered edges after the termination of \cref{line:iterateDegree}. Note that, at this point the number of edges in the graph must be $\leq n^{1/2}f^{1/3} \log(n) \cdot n$, as $\mathrm{Peel}(G - \{e \in E: \mathrm{CoveringCounts}(e) \geq \mathrm{CoveringThreshold}\}, n^{1/2}f^{1/3} \log(n))$ yields an empty graph. Thus, storing these remaining edges require only $\widetilde{O}(n^{3/2}f^{1/3})$ many bits of space. 

    Together with the above, this yields the desired claim. 
\end{proof}

\subsubsection{Reporting Distances}

Before proceeding to the exact stretch analysis, we also show that \cref{alg:ConstructOracle} is able to run to completion; in particular, as stated currently, \cref{line:newVertex} in \cref{alg:ConstructOracle} implicitly assumes that the algorithm is always able to find a new vertex to use as the starting point for each star. We prove this is indeed possible below:

\begin{claim}
Whenever \cref{alg:ConstructOracle} reaches \cref{line:newVertex}, the set $V(G_d) \cap \mathrm{UnusedVertices} \neq \emptyset$, provided $f \geq n$. 
\end{claim}

\begin{proof}
    First, we remark that the minimum degree in the graph in \cref{alg:ConstructOracle} is $\geq n^{1/2}f^{1/3} \log(n)$. This immediately implies that there are $\geq n^{1/2}f^{1/3} \log(n)$ many distinct vertices in the graph at all times. 

    Thus, it suffices to show that less than $n^{1/2}f^{1/3} \log(n)$ many distinct stars are peeled off, as this would then directly imply that there is always some new choice of star center which is unused. 
    
    To bound the number of peeled stars, we revisit the proof of \cref{lem:boundSpaceSingleValued}; here, we showed that in the iterations where the minimum degree bound is $d$, each star we peel (on $\ell$ vertices) covers $\ell \cdot d^2 / 2n$ edges. Because the minimum degree in each star is $\geq \frac{d^2}{2n}$ by \cref{lem:highDegreeLowDiam}, we then also know that each peeled star covers $\geq \frac{d^4}{4n^2}$ many edges. As discussed in \cref{lem:boundSpaceSingleValued}, there are at most $2nd \cdot f^{1/3}$ many edge slots to cover in the iterations when $d$ is the minimum degree bound, thus, at most 
    \[
    \frac{2nd \cdot f^{1/3}}{\frac{d^4}{4n^2}}  = \frac{8 n^{3}f^{1/3}}{d^3}
    \]
    stars are peeled in these iterations. 

    Across all choices of $d$, we then see that 
    \[
    \leq \sum_{d \in \{n/2, n/4, n/8, \dots n^{1/2}f^{1/3}\log(n) \}} \frac{8 n^{3}f^{1/3}}{d^3} = O \left ( \frac{n^{3}f^{1/3}}{n^{3/2} f \log^3(n)}\right ) = O \left ( \frac{n^{3/2}}{f^{2/3}\log^3(n)}\right ) < n^{1/2}f^{1/3},
    \]
    provided that $f \geq n$. Thus, at most $n^{1/2}f^{1/3}$ many stars are peeled off in the entire execution of the algorithm, as we desire.
\end{proof}

Next, we proceed to report distances. To do this, we first use the following key claim:

\begin{claim}
    Let $\bigg ( G_{\mathrm{remaining}}, \mathrm{Stars}, \mathrm{SparseRecoverySketches}, \mathrm{DegreesSingleStar} \bigg )$ be as returned by \cref{alg:ConstructOracle} on a graph $G = (V, E)$, and let $F \subseteq E$ be a set of $f$ edge deletions. Then, given access only to $\mathrm{Stars}$, for any edge $e \in F$, we can determine \emph{exactly} which stars cover the edge $e$.
\end{claim}

Before proving the claim, we remark that such a claim is necessary in order to use our sparse recovery sketches: indeed, such sketches work only when the neighborhood of a vertex is updated accurately. If one mistakenly deletes edges that were not present to begin with, then the output from the sparse recovery sketch is not guaranteed to be meaningful (i.e., not guaranteed to be the true set of remaining edges incident on a given vertex). 

\begin{proof}
    First, for a star $\{v, L_1\}$, we say an edge $(u,w)$ is eligible to be in the star if $(u,w) \subseteq \{v, L_1\}$. For a star $\{v, L_1, L_2\}$, we say that $(u,w)$ is eligible to be in the star if $(u,w) \subseteq \{v \} \times L_1$ or $(u,w) \subseteq L_1 \times L_2$ (note that we do not consider $(u,w)$ eligible when both $u, w$ are in $L_1$ or both in $L_2$).

    With this notion, we claim that $(u,w)$ is covered exactly by the first $10f^{1/3}$ many stars for which $(u,w)$ is eligible. Thus, determining which stars cover the edge $e$ is deterministic; one must only iterate through the stored stars and compute whether an edge $(u,w)$ is eligible to be covered by such a star.

    The proof that eligibility mirrors the covering of edges is essentially by definition: an edge is present in a star if and only if (a) the edge is present in the graph at the time the star is created, and (b) the edge is eligible to be in the star (in this sense, stars are maximal). Note that condition (a) is always true when an edge has been covered $< 10 f^{1/3}$ many times. Together, these imply the above claim. 
\end{proof}

With this, we now present the algorithm for reporting distances when provided with a set of deletions $F$, and two vertices $s, t$ for whom the approximate distance is desired.

\begin{algorithm}[H]
    \caption{ReportDistances$(G_{\mathrm{remaining}}, \mathrm{Stars}, \mathrm{SparseRecoverySketches}, \mathrm{StarDegrees}, F, s, t)$}\label{alg:reportDistanceConstant}
    \For{$e = (u,w) \in F$}{
    Let $\mathrm{Stars}_e$ denote all stars which cover the edge $e$. \\
    \For{$\mathrm{Star} \in \mathrm{Stars}_e$}{
    Let $\mathrm{SparseRecoverySketches}_e$ denote the corresponding sparse recovery sketches for this star on vertices $u$ and $w$. Update these sketches to remove $e$ from the support of the sketch. \label{line:updateSupports}\\
    Let $\mathrm{DegreesSingleStar}_e$ denote the degree counts in $\mathrm{Star}$ for vertices $u$ and $w$. Decrement these degree counts by $1$. \\
    }
    }
$G_{\mathrm{approx}} = G_{\mathrm{remaining}} - F$. \\
\For{$\mathrm{Star} \in \mathrm{Stars}$}{
Let $\mathrm{Degrees}_{\mathrm{Star}}$ be the updated degree counts for vertices in $\mathrm{Star}$ (post deletions). \\
Let $\mathrm{SparseRecoverySketches}_{\mathrm{Star}}$ be the updated recovery sketches counts for vertices in $\mathrm{Star}$ (post deletions). \\
Let $v$ denote the center of $\mathrm{Star}$. \\
Let $\mathrm{RootDecrease}$ denote the number of covered edges by $\mathrm{Star}$ that were deleted in $F$ and incident to the root $v$. \\
\If{$\mathrm{RootDecrease} \leq \frac{n^{1/2}f^{1/3} \log(n)}{ 2}$}{
Let $v_{\mathrm{High}}$ denote all vertices in $\mathrm{Star}$ with post-deletion degree $\geq \frac{5 f^{1/3}}{\log^2(n)}$, and place all other vertices in $\mathrm{Star}$ in $v_{\mathrm{Low}}$. \\
$G_{\mathrm{approx}} \leftarrow G_{\mathrm{approx}} \cup K_{v_{\mathrm{High}}}$. \\
For vertices in $v_{\mathrm{Low}}$, use $\mathrm{SparseRecoverySketches}_{\mathrm{Star}}$ to exactly recover their neighborhoods, and add these edges to $G_{\mathrm{approx}}$. \\
}
}
\Return{$7 \cdot \mathrm{dist}_{G_{\mathrm{approx}}}(s,t)$.}
\end{algorithm}

Ultimately, we will show the following:

\begin{lemma}\label{lem:goodApprox}
    Let $G$ be a graph, let $G_{\mathrm{remaining}}, \mathrm{Stars}, \mathrm{SparseRecoverySketches}, \mathrm{StarDegrees}$ be as constructed by \cref{alg:ConstructOracle}, let $F \subseteq G$ be a set of $\leq f$ deletions, and let $\hat{d}_{G - F}(s, t)$ be the output from \cref{alg:reportDistanceConstant}. Then, 
    \[
    \mathrm{dist}_{G - F}(s, t) \leq \widehat{\mathrm{dist}}_{G - F}(s, t) \leq 7 \cdot \mathrm{dist}_{G - F}(s, t).
    \]
\end{lemma}

We first prove the upper bound of the above expression. To do this, we require the following claim:

\begin{claim}\label{clm:containsGminusF}
     Let $G$ be a graph, let $G_{\mathrm{remaining}}, \mathrm{Stars}, \mathrm{SparseRecoverySketches}, \mathrm{StarDegrees}$ be as constructed by \cref{alg:ConstructOracle}, let $F \subseteq G$ be a set of $\leq f$ deletions, and let $G_{\mathrm{approx}}$ be as constructed by \cref{alg:reportDistanceConstant}. Then, $G- F \subseteq G_{\mathrm{approx}}$.
\end{claim}

\begin{proof}
    First, observe that for any edge in $G_{\mathrm{remaining}} - F$, such an edge is also by default in $G_{\mathrm{approx}}$. So, we must only show that edges in $(G - G_{\mathrm{remaining}}) - F$ are in $G_{\mathrm{approx}}$. For these edges, we know that each one must be covered by $10 f^{1/3}$ many different stars as per the construction in \cref{alg:ConstructOracle}. 

    Now, fix such an edge $e = (u,w) \in (G - G_{\mathrm{remaining}}) - F$. Because $(u,w)$ is covered by $10 f^{1/3}$ many different stars, each with a different root, after applying the deletions from $F$, there must be some root which receives 
    \[
    \leq \frac{2f}{10 f^{1/3}} \leq \frac{f^{2/3}}{5}
    \]
    many deletions. Let us fix this star, and denote it by $\{v, L_1, L_2\}$ (the case where it is a $1$-hop star follows similarly). If $u, w$ have degree $\geq \frac{5f^{1/3}}{\log(n)}$, then $u,w$ are covered together by a clique in $G_{\mathrm{approx}}$, and so the edge $u,w$ is present. Otherwise, if one of $u, w$ has degree $< \frac{5f^{1/3}}{\log(n)}$, then the edge $(u,w)$ is recovered explicitly via our sparse recovery sketch (see \cref{clm:efficientUniqueSum}), and added to our graph $G_{\mathrm{approx}}$.
\end{proof}

\begin{corollary}\label{cor:upperBoundApprox}
     Let $G$ be a graph, let $G_{\mathrm{remaining}}, \mathrm{Stars}, \mathrm{SparseRecoverySketches}, \mathrm{StarDegrees}$ be as constructed by \cref{alg:ConstructOracle}, let $F \subseteq G$ be a set of $\leq f$ deletions, and let $G_{\mathrm{approx}}$ be as constructed by \cref{alg:reportDistanceConstant}. Then, 
     \[
     \widehat{\mathrm{dist}}_{G - F}(s, t) \leq 7 \cdot \mathrm{dist}_{G - F}(s, t).
     \]
\end{corollary}
\begin{proof}
    By \cref{clm:containsGminusF}, we know that $G- F \subseteq G_{\mathrm{approx}}$.

    So, any edge $(u,w) \in G - F$ is also present in $G_{\mathrm{approx}}$, so 
    \[
    \mathrm{dist}_{G_{\mathrm{approx}}}(s,t) \leq \mathrm{dist}_{G}(s,t),
    \]
    and thus 
    \[
    \widehat{\mathrm{dist}}_{G_{\mathrm{approx}}}(s,t) = 7 \cdot \mathrm{dist}_{G_{\mathrm{approx}}}(s,t) \leq 7 \cdot \mathrm{dist}_{G}(s,t),
    \]
    This concludes the proof.
\end{proof}

Now, it remains only to prove the lower bound; i.e., that $\mathrm{dist}_{G - F}(s, t) \leq \widehat{\mathrm{dist}}_{G - F}(s, t)$. For this, we rely on the following claim:

\begin{claim}\label{clm:shortPath}
    Let $G$ be a graph, let $G_{\mathrm{remaining}}, \mathrm{Stars}, \mathrm{SparseRecoverySketches}, \mathrm{StarDegrees}$ be as constructed by \cref{alg:ConstructOracle}, let $F \subseteq G$ be a set of $\leq f$ deletions, and let $G_{\mathrm{approx}}$ be as constructed by \cref{alg:reportDistanceConstant}. Then, for any edge $e = (u,w) \in G_{\mathrm{approx}}$, $u$ and $w$ are connected in $G - F$ by a path of length $\leq 7$.
\end{claim}

\begin{proof}
    For any edge $e \in G_{\mathrm{remaining}} - F$, these edges have a path of length $1$ in $G - F$, and so the property holds trivially. 

    Otherwise, we consider any edge $(u,w)$ in $G_{\mathrm{approx}} - (G_{\mathrm{remaining}} - F)$. There are only two ways for such an edge $(u,w)$ to be present:
    \begin{enumerate}
        \item There is some star such that either $u$ or $w$ was in $v_{\mathrm{Low}}$ in this star, and the edge $(u,w)$ was recovered by opening $\mathrm{SparseRecoverySketches}$. In this case, because we perform the appropriate deletions in \cref{line:updateSupports}, it must be the case that the edge $(u,w)$ is present in the graph $G - F$, and thus $u, w$ are connected by a path of length $1$ in $G - F$.
        \item There is some star whose root underwent $\leq n^{1/2}f^{1/3} \log(n) /2$ deletions and $u, w$ both have post-deletion degree $\geq \frac{5f^{1/3}}{\log^2(n)}$. In this case, by construction, we know that every star has minimum degree $d \geq (n^{1/2} f^{1/3} \log(n))^2 / n = f^{2/3}\log^2(n)$ (see \cref{lem:highDegreeLowDiam}). Thus, by \cref{cor:nDeletionsStretch8} if $u, w$ both have degree $\geq \frac{5f^{1/3}}{\log(n)}$, then $u$ and $w$ must be connected by a path of length $7$ within the post-deletion star, which ensures that $G - F$ contains a path between $u, v$ of length $\leq 7$.
    \end{enumerate}

    Thus, in all cases, for any edge $e  = (u,w) \in G_{\mathrm{approx}}$, there is a path of length $\leq 7$ between $(u,w)$ in $G - F$.
\end{proof}

This immediately gives the following corollary:

\begin{corollary}\label{cor:lowerBoundApprox}
    Let $G$ be a graph, let $G_{\mathrm{remaining}}, \mathrm{Stars}, \mathrm{SparseRecoverySketches}, \mathrm{StarDegrees}$ be as constructed by \cref{alg:ConstructOracle}, let $F \subseteq G$ be a set of $\leq f$ deletions, and let $G_{\mathrm{approx}}$ be as constructed by \cref{alg:reportDistanceConstant}. Then, 
     \[
     \widehat{\mathrm{dist}}_{G - F}(s, t) \geq \mathrm{dist}_{G - F}(s, t).
     \]
\end{corollary}

\begin{proof}
    Indeed, we know that $\widehat{\mathrm{dist}}_{G - F}(s, t) = 7 \cdot \mathrm{dist}_{G_{\mathrm{approx}}}(s, t)$. For any edge $(u,w)$ in this shortest path from $s$ to $t$ in $G_{\mathrm{approx}}$, we know that there is a path from $u$ to $w$ in $G - F$ of length $\leq 7$ by \cref{clm:shortPath}. This yields the corollary.
\end{proof}

\begin{proof}[Proof of \cref{lem:goodApprox}]
    This follows from \cref{cor:lowerBoundApprox} and \cref{cor:upperBoundApprox}.
\end{proof}

\begin{proof}[Proof of \cref{thm:constantStretch}.]
The space bound follows from \cref{lem:finalConstantStretchSpaceBound} and the correctness follows from \cref{lem:goodApprox}.
    
\end{proof}

\section{On $\widetilde{O}(n \sqrt{f})$ Size Fault-Tolerant Spanners Against an Oblivious Adversary}\label{sec:ObliviousSpanners}

In some settings, it is desirable to have a spanner instead of a distance oracle. Unfortunately, in the fault-tolerant setting, it is known that any spanner tolerant to $f$ edge deletions must retain $\Omega(nf)$ edges. However, this lower bound only holds for deterministic spanners, or more generally, spanners that are fault-tolerant against \emph{adaptive} edge deletions. In this section, we show that if one relaxes this condition and instead requires spanners that are resilient only against \emph{oblivious} edge deletions (i.e., edge deletions that do not depend on the internal randomness of the spanner), then it is in fact possible to beat this $nf$ lower bound. Formally, we show the following:

\begin{theorem}\label{thm:obliviousSpanner}
    Given a graph $G = (V, E)$ and a parameter $f$, there is a randomized sketch $\mathcal{S}(G)$ of size $\widetilde{O}(n \sqrt{f})$ bits, such that for any set of edge deletions $F \subseteq E, |F| \leq f$, with probability $1 - 1 / n^{100}$ over the randomness of the sketch, $\mathcal{S}(G)$ recovers an $O(\log(n)\log\log(n))$ spanner of $G - F$.
\end{theorem}

As in the previous section, we start by presenting the construction of the sketch, followed by a space analysis, and lastly the correctness guarantees.

\subsection{Building the Spanner Sketch}

The sketch follows the same general outline as in \cref{sec:buildDO}, with the minor caveat that now, it no longer suffices to simply know that a given component forms a good expander; instead, we require the ability to actually \emph{recover a subset of edges} which witnesses the low diameter property of this expander.

In order to accomplish this, within each expander component $V_i^{(j)}$, we instantiate $\frac{|G_j[V_i^{(j)}]|}{\log^3(n)}{D}$ many linear sketches of $\ell_0$-samplers. These sketches have the guarantee that, after any deletion of edges, they still return uniform, unbiased samples of the remaining edge set with high probability. Thus, for whichever expander components remain intact after deletion, we can then recover edges using these $\ell_0$-samplers and invoke \cref{clm:subsampleExpander} to argue that the recovered subgraph is indeed a witness to the expansion.

However, the above introduces an extra layer of subtlety: in particular, it is possible (for instance) that in the first round of the algorithm, the entire graph is an expander, and the entire graph can have $\Omega(n^2)$ edges. Thus, $\frac{|G_j[V_i^{(j)}]|}{\log^3(n)}{D} = \widetilde{\Omega} \left ( \frac{n^2}{\sqrt{f}}\right )$, which is much larger than we can afford to store (particularly in the regime when $f \leq n$). However, in this regime when the graph has $\Omega(n^2)$ edges, there is no reason for us to start with our minimum degree requirement being only $\approx \sqrt{f}$. Instead, we start with our minimum degree being as large as $n/2$, and iteratively decrease it by factors of $2$, thereby ensuring that when our minimum degree requirement is $D$, the number of edges in the graph is also bounded by $\widetilde{O}(n \cdot D)$.

With this, we formally present our algorithm below as \cref{alg:buildObliviousSpannerSketch}.

\begin{algorithm}[H]
    \caption{BuildObliviousSpannerSketch$(G = (V, E), f)$}\label{alg:buildObliviousSpannerSketch} 
    $X_{n} = E$. \\
    $D = n/2$. \\
    \While{$D \geq 8 \sqrt{f}\log(n)$}{
    $\ell = 1$. \\
    $G_{\ell, D} = (V, X_{2 \cdot D})$. \\
    Let $C$ be a sufficiently large constant. \\
    \While{$G_{\ell, D}$ has $\geq C n D \log^2(n)$ edges}{
    Let $V^{(\ell, D)}_1, \dots V^{(\ell, D)}_{k_{\ell, D}}$ denote the high-degree expander decomposition of $G_{\ell, D}$ as guaranteed by \cref{thm:expanderDecomp} with $D$, and let $X_{\ell+1, D}$ denote the crossing edges. \\
    \For{$j \in [k_{\ell, D}]$, and each vertex $v \in V_j^{(\ell, D)}$}{ 
    Let $\Gamma_{G_{\ell, D}[V_j^{(\ell, D)}]}(v)$ denote the edges leaving $v$ in the expander $G_{\ell, D}[V_j^{(\ell, D)}]$. \\
    Let $\mathcal{S}_{v, \ell_0}^{(\ell, D)}$ be $\frac{4f}{D} +1$ many $\ell_0$-samplers (with independent randomness) over $\Gamma_{G_{\ell, D}[V_j^{(\ell, D)}]}(v)$. \\
    }
    \For{$j \in [k_{\ell, D}]$}{ 
    Let $\mathcal{S}^{(j, \ell, d)}_{\ell_0}(G_{\ell, D}[V_j^{(\ell, D)}])$ denote a set of $\frac{|G_{\ell, D}[V_j^{(\ell, D)}]| \log^5(n)}{D}$ many linear sketches of $\ell_0$ samplers (each with independent randomness) over the edge set of $G_{\ell, D}[V_j^{(\ell, D)}]$.
    }
    Let $G_{\ell +1, D} = (V, X_{\ell+1, D})$. \\
    $\ell \leftarrow \ell + 1$.
    }
    $X_{D} = X_{\ell, D}$. \\
    $D \leftarrow D /2 $. 
    }
    $X = X_{\ell, D}$. \\
    $\mathcal{S}(G) = \bigg ( \bigg \{(V^{(j, D)}_1, \dots V^{(j, D)}_{k_{j, D}}): j \in [\ell], D \in \{n/2, \dots 8 \sqrt{f}\log(n) \} \bigg \},$ \\
    \quad \quad \quad \quad \quad  $\{\mathrm{deg}_{G_{j, D}[V_i^{(j, D)}]}(v): j \in [\ell], i \in k_{j, D}, v \in V^{(j, D)}_i, D \in \{n/2, \dots 8 \sqrt{f}\log(n) \} \},$ \\
    \quad \quad \quad \quad \quad $\{ \mathcal{S}_{v, \ell_0}^{(\ell, D)}: j \in [\ell], v \in \bigcup_{i \in k_{j, D}} V^{(j, D)}_i,  D \in \{n/2, \dots 8 \sqrt{f}\log(n) \}\}, $ \\
    \quad \quad \quad \quad \quad $\{\mathcal{S}^{(j, \ell, d)}_{\ell_0}(G_{\ell, D}[V_j^{(\ell, D)}]): j \in [\ell], i \in k_{j, D}, D \in \{n/2, \dots 8 \sqrt{f}\log(n) \} \}$ \\
    \quad \quad \quad \quad \quad $X \bigg )$ \\
\Return{$\mathcal{S}(G)$}.
\end{algorithm}

\subsection{Bounding the Space}

Now, we bound the space used by \cref{alg:buildObliviousSpannerSketch}. We prove the following lemma:

\begin{lemma}\label{lem:spaceBoundSpanner}
    For a graph $G = (V, E)$, and a parameter $f$, \cref{alg:buildObliviousSpannerSketch} returns a sketch that requires $\widetilde{O}(n \sqrt{f})$ bits of space.
\end{lemma}

\begin{proof}
    First, we bound the number iterations in \cref{alg:buildObliviousSpannerSketch}. For this, observe that for a given value of $D$, it must be the case that the starting graph $G_{1, D}$ has at most $2C n D \log^2(n)$ edges (as otherwise, the algorithm would not have terminated the previous iterations with $2 \cdot D$ instead of $D$. Now, by \cref{thm:expanderDecomp}, there exists a constant $C'$ such that after each invocation of the expander decomposition, the number of edges in $X_{i, D}$ is at most $|X_{i-1, D}|/2 + C' \cdot n \cdot D \cdot\log^2(n)$. So, if we let our large constant in \cref{alg:buildDataStructure} be $C = 4 \cdot C'$, then we know that in every iteration where $|X_{i, D}| \geq CnD\log^2(n)$, it is the case that $|X_{i+1, D}| \leq \frac{|X_{i, D|}}{2} + \frac{|X_{i, D}|}{4} \leq \frac{3|X_{i, D}|}{4}$. Thus after only $\ell = O(1)$ iterations, it must be the case that $|X_{\ell, D}| \leq C \cdot nD\log^2(n)$. Finally, because there are at most $O(\log(n))$ choices of $D$, we thus obtain that there are at most $O(\log(n))$ iterations in  \cref{alg:buildObliviousSpannerSketch}.

    The first piece of information that is stored in our data structure is the identity of all the components in each level of our expander decomposition. These components always form a partition of our vertex set, and thus storing $(V^{(\ell)}_1, \dots V^{(j)}_{k_{j}})$ for a single iteration requires $\widetilde{O}(n)$ bits of space. Across all $O(\log(n))$ of the decomposition, storing this information still constitutes only $\widetilde{O}(n)$ bits of space.

    Next, to store all the vertex degrees, for each vertex we store a number in $[n]$, which requires $O(\log(n))$ bits. Across all $n$ vertices and $O(\log(n))$ levels of the decomposition, this too constitutes only $\widetilde{O}(n)$ bits of space.

    Next, we consider the sparse recovery sketches that we store for each vertex. Recall that each sketch is $ \frac{4f}{D} + 1$ many $\ell_0$-samplers.  By \cref{thm:ell0samp}, we can instantiate each $\ell_0$-sampler with universe $u = \binom{V}{2}$, and failure parameter $\delta = 1 / n^{1000}$, and these will still take total space $\widetilde{O}(4f / D)$. Summing across all $n$ vertices and all $O(\log(n))$ iterations, the total space required by these samplers is bounded by $\widetilde{O}(nf / \sqrt{f}) = \widetilde{O}(n\sqrt{f})$, where we have used that $D \geq \sqrt{f}$ in each iteration.

    Next, for the $\ell_0$-samplers that we store for sampling expanders in each iteration, recall that for each expander component $V^{(j, D)}$, we store $\frac{|G_{j, D}[V^{(j, D)_i}]|\log^5(n)}{D}$ many linear sketches of $\ell_0$-samplers. Thus, in a single iteration, the number of $\ell_0$-samplers we store is bounded by
    \[
    \sum_{i \in [k_{j, D}]} \frac{|G_{j, D}[V^{(j, D)_i}]|\log^5(n)}{D} \leq \frac{|G_{j, D}| \log^5(n)}{D} \leq \frac{2CnD\log^2(n)\log^5(n)}{D} = \widetilde{O}(n).
    \]
    By \cref{thm:ell0samp}, we can instantiate each $\ell_0$-sampler with universe $u = \binom{V}{2}$, and failure parameter $\delta = 1 / n^{1000}$, and these will still take total space $\widetilde{O}(n)$. Summing across all $O(\log(n))$ iterations then yields a total space bound of $\widetilde{O}(n)$ bits for the $\ell_0$-samplers.

    Finally, we bound the space of the final crossing edges $X$. By the termination condition in \cref{alg:buildObliviousSpannerSketch}, it must be the case that the final crossing edges $|X| = O(n \sqrt{f}\log^3(n))$ (as in the last iteration, $D \leq 16 \sqrt{f}\log(n)$), and therefore the total space required for storing these edges is $\widetilde{O}(n \sqrt{f})$ bits of space. This concludes the lemma. 
\end{proof}

\subsection{Proving Spanner Recovery}

Finally, in this section we show how to use our aforementioned sketch to recover a spanner (with high probability) in the post-deletion graph. This recovery procedure is very similar to \cref{alg:reportDistance}, except that instead of replacing each expander component with a clique, we instead replace it with a sub-sample of the expander (but which retains the same expansion properties).

Note that just as \cref{alg:reportDistance} required \cref{clm:findingEdge}, we use the following generalization of \cref{clm:findingEdge} to our setting (where now the degree lower bound is also evolving over time):

\begin{claim}\label{clm:findingEdgeSpanner}
    Let $G = (V, E)$ be a graph, let $f$ be a parameter, and let $\mathcal{S}(G)$ be constructed as in \cref{alg:buildObliviousSpannerSketch}. Then, an edge $e \in G_{j, D}[V_i^{(j, D)}]$ if and only if $V_i^{(j, D)}$ is the component with (1) the largest value of $D$, and (2) conditioned on having the largest value of $D$, has the \emph{smallest} value of $j$ such that $e \subseteq V_i^{(j, D)}$.
\end{claim}

This claim follows from the same reasoning as \cref{clm:findingEdge}: because we construct our expanders as induced subgraphs, whichever induced subgraph is the \emph{first} to contain an edge $e$ will be the expander that contains $e$. In \cref{alg:buildObliviousSpannerSketch}, expanders are built in decreasing order of $D$, and then increasing order of $j$, hence yielding \cref{clm:findingEdgeSpanner}.

Now, the formal pseudocode of our recovery algorithm is provided below as \cref{alg:recoverSpanner}.

\begin{algorithm}[H]
\caption{RecoverSpanner$(\mathcal{S}(G), F)$}\label{alg:recoverSpanner}
\For{$e = (u,v) \in F$}{
\If{exists component $V^{(j, D)}_i$ such that $e \subseteq V^{(j, D)}_i$}{
Let $V^{(j, D)}_i$ denote the first such component (i.e., for the largest value of $D$, smallest value of $j$, as per \cref{clm:findingEdgeSpanner}). \\
$\mathrm{deg}_{G_{j, D}[V_i^{(j, D)}]}(v) \leftarrow \mathrm{deg}_{G_{j, D}[V_i^{(j, D)}]}(v) -1 $. \\
$\mathrm{deg}_{G_{j, D}[V_i^{(j, D)}]}(u) \leftarrow \mathrm{deg}_{G_{j}[V_i^{(j, D)}]}(u) -1 $. \\
$\mathcal{S}_{v, \ell_0}^{(j, D)} \leftarrow \mathcal{S}_{v, \ell_0}^{(j, D)} - \mathcal{S}_{v, \ell_0}^{(j, D)}(e)$. \label{line:deleteNeighborhoodEdges1}\\
$\mathcal{S}_{u, \ell_0}^{(j, D)} \leftarrow \mathcal{S}_{u, \ell_0}^{(j, D)} - \mathcal{S}_{u, \ell_0}^{(j, D)}(e)$.\label{line:deleteNeighborhoodEdges2} \\
$\mathcal{S}^{(i, j, D)}_{\ell_0}(G_{\ell, D}[V_i^{(j, D)}]) \leftarrow \mathcal{S}^{(i, j, d)}_{\ell_0}(G_{j, D}[V_i^{(j, D)}]) - \mathcal{S}^{(i, j, d)}_{\ell_0}(e)$. \\
}
\Else{
$X \leftarrow X- e$. \label{line:deleteCrossingEdgesinF}\\
}
}
Let $H = (V, \emptyset)$. \\
\For{$V^{(j, D)}_i: D \in \{n/2, \dots 8\sqrt{f}\log(n)\}, j \in [\ell], i \in [k_{j, D}]$}{
Let $V'^{(j, D)}_i \subseteq V^{(j, D)}_i$ denote all vertices $u \in V^{(j, D)}_i$ such that $\mathrm{deg}_{G_{j}[V_i^{(j)}]}(u) \geq D/2$. \label{line:defineV'}\\
\For{$u \in V^{(j, D)}_i - V'^{(j, D)}_i$}{
Open the $\ell_0$-samplers $\mathcal{S}_{v, \ell_0}^{(j, D)}$ to reveal edges $\Gamma_{V^{(j, D)}_i}(u)$. \\
$H \leftarrow H + \Gamma_{V^{(j)}_i}(u)$. \label{line:addFullNeighborhoodSpanner} \\
$\mathcal{S}^{(i, j, D)}_{\ell_0}(G_{\ell, D}[V_i^{(j, D)}]) \leftarrow \mathcal{S}^{(i, j, d)}_{\ell_0}(G_{j, D}[V_i^{(j, D)}]) - \mathcal{S}^{(i, j, D)}_{\ell_0}(\Gamma_{V^{(j)}_i}(u))$. \label{line:deleteEdgesell0samp}\\
}
Open the $\ell_0$-samplers $\mathcal{S}^{(i,j, D)}_{\ell_0}(G_{j, D}[V_i^{(j, D)}])$ to reveal a subgraph $\widetilde{G}_{j, D}[V_i^{(j, D)}]$. \\
$H \leftarrow H + \widetilde{G}_{j, D}[V_i^{(j, D)}]$
\label{line:addSubsampleExpander} \tcp{Add the subsampled expander.} 

}
$H \leftarrow H + X$. \label{line:finishHSpanner}\\
\Return{$H$.}
\end{algorithm}

Ultimately, we will show that $H$ is an $O(\log(n)\log\log(n))$ spanner of $G - F$ with high probability. We start by showing that $H$ is indeed a subgraph of $G - F$:

\begin{claim}\label{clm:containsH}
    Let $G = (V, E)$ be a graph, let $f$ be a parameter, and let $\mathcal{S}(G)$ be the result of invoking \cref{alg:buildObliviousSpannerSketch} on $G, f$. Now, suppose that we call \cref{alg:recoverSpanner} with $\mathcal{S}(G)$ and a set of deletions $F$. Then, the graph $H$ as returned by \cref{alg:buildObliviousSpannerSketch} is a subgraph of $G - F$. 
\end{claim}

\begin{proof}
    Let us consider any edge $e \in H$. We claim that any such edge $e$ must also be in $G - F$. Indeed, we can observe that there are only three ways for an edge to be added to $H$:
    \begin{enumerate}
        \item In \cref{line:finishHSpanner}, we add the final crossing edges $X$ which were not in $F$ (in \cref{line:deleteCrossingEdgesinF} we deleted all edges in $X \cap F$). Note that $X \subseteq G$ by construction, and so $X - F \subseteq G - F$.
        \item In \cref{line:addFullNeighborhoodSpanner}, whenever a vertex has degree $\leq D/2$ in an expander component $G_{j, D}[V_i^{(j, D)}] - F$, we open its $\ell_0$ samplers to recover edges in its neighborhood and add those edges to $H$. Importantly, these are still edges that are a subset of $G - F$, as in \cref{line:deleteNeighborhoodEdges1} and \cref{line:deleteNeighborhoodEdges2}, we subtracted $g(e)$ from $ \mathcal{S}_u^{(j, D)}$ whenever $e = (u,v)$ appeared in $F$. 
        \item Otherwise, in \cref{line:addSubsampleExpander}, for vertices whose degree remains $\geq \frac{4f}{D} + 1$ in some expander component $G_{j, D}[V_i^{(j, D)}]$ after performing deletions $F$, we open our $\ell_0$-samplers for the remaining edges. Importantly, because we have deleted all edges in $F$ from these samplers (see \cref{line:deleteEdgesell0samp}), by \cref{def:ell0samp}, the edges we recover must be in $G_{j, D}[V_i^{(j, D)}] - F$, and so these edges must be contained in $G- F$ as well.
    \end{enumerate}
\end{proof}

Next, we show that with high probability, $H$ is indeed an $O(\log(n)\log\log(n))$ spanner of $G$. To do this, we have the following subclaim:

\begin{claim}\label{clm:recoverExpander}
    Let $G, F, H$ be as in \cref{clm:containsH}. Then, for a lopsided-expander component $G_{j, D}[V_i^{(j, D)}]$, when opening the $\ell_0$ samplers $\mathcal{S}^{(i, j, D)}_{\ell_0}(G_{j, D}[V_i^{(j, D)}])$ to reveal a subgraph $\widetilde{G}_{j, D}[V_i^{(j, D)}]$, it is the case that $\widetilde{G}_{j, D}[V_i^{(j, D)}]$ contains an $\Omega(1 / \log(n))$-lopsided expander over the vertices $V'^{(j, D)}_i$ (defined in \cref{line:defineV'}) with probability $1 - 1 / n^{997}$.
\end{claim}

\begin{proof}
    Indeed, let us consider the subgraph $G_{j, D}[V_i'^{(j, D)}]$ (before sampling). By \cref{lem:robustnessfDeletion}, we know that $G_{j, D}[V_i'^{(j, D)}]$ is an $\Omega(1 / \log(n))$ lopsided-expander.

    Now, let us denote the number of edges in $G_{j, D}[V_i'^{(j, D)}]$ by $W_{\mathrm{sub}}$, and denote the number of edges in $G_{j, D}[V_i^{(j, D)}]$ by $W$. By the definition of $\ell_0$-samplers, whenever we open such a sampler (and it does not return $\perp$), we expect the sample to lie in $G_{j, D}[V_i'^{(j, D)}]$ with probability $\frac{W_{\mathrm{sub}}}{W}$ (see \cref{def:ell0samp}). Thus, when opening $\frac{|G_{\ell, D}[V_i^{(j, D)}]| \log^4(n)}{D}$ many $\ell_0$-samplers (and conditioning on no failures), we expect
    \[
    \geq \frac{|G_{j, D}[V_i^{(j, D)}]| \log^5(n)}{D} \cdot \frac{W_{\mathrm{sub}}}{W} \geq \frac{W_{\mathrm{sub}}\log^5(n)}{D}
    \]
    many samples to land in $G_{j, D}[V_i'^{(j, D)}]$. In particular, using the fact that the minimum degree in $G_{j, D}[V_i'^{(j, D)}]$ is still $\Omega(D)$, and that $G_{j, D}[V_i'^{(j, D)}]$ still has $\Omega(D^2)$ edges (see \cref{lem:robustnessfDeletion} for the proof of this), this implies that opening the $\ell_0$-samplers yields a random sample of 
    \[
    \Omega \left (  \frac{|G_{j, D}[V_i'^{(j, D)}]| \cdot \log^5(n)}{\mathrm{mindeg}(G_{j, D}[V_i'^{(j, D)}])} \right )
    \]
    in expectation, and thus, using a Chernoff bound along with \cref{clm:subsampleExpander}, implies that probability $1 - 2^{- \Omega(\log^2(n))}$ the resulting subsample contains an $\Omega(1 / \log(n))$ lopsided-expander over $G_{j, D}[V_i'^{(j, D)}]$. Taking a union over the aforementioned failure probabilities, along with the failure probabilities of the individual $\ell_0$-samplers ($1 / n^{1000}$) of which there are at most $n^2 \mathrm{polylog}(n)$, then yields the above claim. 
\end{proof}

Finally then, we can prove the bound on the stretch of our spanner:

\begin{lemma}\label{lem:lowStretch}
    Let $G, F, H$ be as in \cref{clm:containsH}. Then, $H$ is an $O(\log(n)\log\log(n))$ spanner of $G - F$ with probability $1 - 1 / n^{100}$.
\end{lemma}

\begin{proof}
    Consider any edge $(u,v) \in G - F$. We claim that the distance between $u$ and $v$ in $H$ will be bounded by $O(\log(n)\log\log(n))$ (with high probability), and thus $H$ is an $O(\log(n)\log\log(n))$ spanner.

    There are a few cases for us to consider:

    \begin{enumerate}
        \item Suppose there is no choice of $j, D$ such that $e = (u,v) \subseteq V_i^{(j, D)}$. This means that $e$ must be in $X - F$. But, the set $X$ is explicitly stored ($F$ is removed) in $\mathcal{S}(G)$, and thus $e \in H$. 
        \item Instead, suppose there is a choice of $i, j, D$ such that $e = (u,v) \subseteq V_i^{(j, D)}$, and let $V_i^{(j, D)}$ be the first such set (as per \cref{clm:findingEdgeSpanner}). Again, we have several cases:
        \begin{enumerate}
            \item If one of $u, v$ (WLOG $u$) has degree $\leq \frac{4f}{D}$, then when opening the $\ell_0$-samplers $\mathcal{S}_{u, \ell_0}^{(j, D)}$, with probability $1 - 1 / n^{999}$ (taking a union bound over the failure probabilities of the $\ell_0$-samplers) it is the case that all edges in the neighborhood of $u$ (in this lopsided expander component) are recovered. Thus, $e = (u,v)$ is recovered and is therefore in $H$. 
            \item Otherwise, suppose that both of $u, v$ have degree $\geq \frac{4f}{D} +1$. Then, as per \cref{lem:robustnessfDeletion}, we know that there are at most $\frac{4f}{D}$ vertices in $V_{i}^{(j, D)}$ which are \emph{not in} $V_{i}'^{(j, D)}$. So, $u, v$ both recover with probability $1 - 1 / n^{999}$ (taking a union bound over the failure probabilities of the $\ell_0$-samplers), an edge to some vertices $u', v'$ in $V_{i}'^{(j, D)}$. Now, by \cref{clm:recoverExpander}, we know that we have also recovered an $\Omega(1 / \log(n))$ lopsided-expander over $V_{i}'^{(j, D)}$, and so $u', v'$ are at distance $O(\log(n)\log\log(n))$ in $H$. Thus, $u, v$ are at distance $O(\log(n)\log\log(n)) + 2$, thereby yielding the claim. 
        \end{enumerate}
    \end{enumerate}
\end{proof}

Finally, we prove \cref{thm:obliviousSpanner}.

\begin{proof}[Proof of \cref{thm:obliviousSpanner}.]
The size of the sketch follows from \cref{lem:spaceBoundSpanner}. The spanner property follows from \cref{clm:containsH} and \cref{lem:lowStretch}.
\end{proof}

\section{Bounded Deletion Streams}\label{sec:streaming}

The previous sections showed how to construct distance oracles and spanners in the \emph{static setting}, where one is provided with the graph $G$, and then tasked with constructing a data structure or sketch which can handle future deletions.

Perhaps more naturally though, one may ask if the same space bounds are possible when working in the \emph{streaming setting}. Here, in every iteration, an edge is either inserted or deleted. The goal is to maintain a sketch such that at the termination of the stream, the sketch can be used to recover a distance oracle or spanner of the resulting graph. As above, we let $f$ denote the number of potential \emph{deletions} that appear in the stream, and we are concerned with bounding the maximum space taken up by our data structure / sketch at any point of the stream.

\subsection{Deterministic Distance Oracles}\label{sec:streamDO}

Our first streaming result shows that in a strong sense, nontrivial space guarantees are possible even in this more restrictive setting:

\begin{theorem}\label{thm:faultTolerantDOStream}
    There is a deterministic streaming algorithm which, when given a parameter $f \in [0, \binom{n}{2}]$, and then invoked on a stream of edge insertions building a graph $G$, and at most $f$ deletions (denoted $\mathrm{Deletions} \subseteq G$):
    \begin{enumerate}
        \item Uses space at most $\widetilde{O}(n^{4/3} f^{1/3})$ bits. 
        \item Supports a query operation which, for any vertices $u, v \in V$, returns a value $\widehat{\mathrm{dist}}_{G - \mathrm{Deletions}}(u,v)$ such that 
        \[
        \mathrm{dist}_{G - \mathrm{Deletions}}(u,v)\leq \widehat{\mathrm{dist}}_{G - \mathrm{Deletions}}(u,v) \leq O(\log(n)\log\log(n)) \cdot \mathrm{dist}_{G - \mathrm{Deletions}}(u,v).
        \]
    \end{enumerate}
\end{theorem}

\subsubsection{Streaming Algorithm Implementation}

In this section, we discuss the implementation of the aforementioned streaming algorithm. The general structure resembles that of \cref{sec:buildDO}, with the difference being that instead of performing the decomposition over the entire graph at once, we instead store a \emph{buffer} of edges, and iteratively perform an expander decomposition whenever the number of edges in the buffer becomes sufficiently large.

In order to achieve our desired space bound tradeoff, we let our buffer capacity be set at $C \cdot n^{4/3}f^{1/3} \log^4(n)$, where $C$ is a sufficiently large constant. Now, whenever a new edge is inserted into the stream, we first check if there is some \emph{existing} expander component in which this edge can be inserted. If not, then we insert the edge into the buffer. Once the buffer reaches capacity, we perform an expander decomposition using minimum degree $D = n^{1/3} f^{1/3} \log^2(n)$, and store our sketches for these expander components. The crossing edges (non-expander edges) are added back into the buffer, and we repeat this process. 

This procedure is formally outlined in \cref{alg:streamingBuildDO}.

\begin{algorithm}[p]
    \caption{StreamingDistanceOracle$(f, e_i, \Delta_i)$}\label{alg:streamingBuildDO}
    $\mathrm{Deletions} = \emptyset$. \\
    $\mathrm{Buffer} = \emptyset$. \\
    $\mathrm{ExpanderSketches} = []$. \\
    $r = 1$. \\
    Let $D = n^{1/3}f^{1/3}\log^2(n)$ and let $g$ be a function in accordance with \cref{fact:uniqueSum} with parameters $u = \binom{n}{2}$ and $k = \frac{4f}{D}$. \\
    \For{$(e_i = (u_i, v_i), \Delta_i)$ in the stream}{
    \If{$\Delta_i = -1$}{
    $\mathrm{Deletions} \leftarrow \mathrm{Deletions} \cup \{ e_i \}$. \\
    }
    \ElseIf{$\exists V_{j}^{(\ell)} \in \mathrm{ExpanderSketches}: e_i \subseteq V_{j}^{(\ell)}$}{
    Let $V_{j}^{(\ell)}$ be the first such vertex set in $\mathrm{ExpanderSketches}$ (i.e., smallest value of $\ell$).\label{line:checkFirst} \\
    $\mathcal{S}_{v_i}^{(\ell)} \leftarrow \mathcal{S}_{v_i}^{(\ell)} + g(e_i)$. \\
    $\mathcal{S}_{u_i}^{(\ell)} \leftarrow \mathcal{S}_{u_i}^{(\ell)} + g(e_i)$. \\
    $\mathrm{deg}_{G_{\ell}[V_{j}^{(\ell)}]}(u_i) \leftarrow \mathrm{deg}_{G_{\ell}[V_{j}^{(\ell)}]}(u_i) +1$. \\
    $\mathrm{deg}_{G_{\ell}[V_{j}^{(\ell)}]}(v_i) \leftarrow \mathrm{deg}_{G_{\ell}[V_{j}^{(\ell)}]}(v_i) +1$. \\
    }
    \Else{
    $\mathrm{Buffer} \leftarrow \mathrm{Buffer} \cup \{e_i \}$. \\
    \If{$|\mathrm{Buffer}| \geq C \cdot n^{4/3}f^{1/3} \log^4(n)$ \label{line:bufferFull}}{
    Let $V^{(r)}_1, \dots V^{(r)}_{k_r}$ denote the high-degree expander decomposition of $\mathrm{Buffer}$ as guaranteed by \cref{thm:expanderDecomp} with $D = n^{1/3}f^{1/3}\log^2(n)$, and let $X$ denote the crossing edges. \\
    Let $G_r$ denote the expander edges. \\
    \For{$j \in [k_{r}]$, and each vertex $v \in V_j^{(r)}$}{ 
    Let $\Gamma_{G_{r}[V_j^{(r)}]}(v)$ denote the edges leaving $v$ in the expander $G_{r}[V_j^{(r)}]$. \\
    Let $\mathcal{S}_v^{(r)} = g(\Gamma_{G_{r}[V_j^{(r)}]}(v))$. \\
    $\mathrm{deg}_{G_{r}[V_{j}^{(r)}]}(v) = |\Gamma_{G_{r}[V_j^{(r)}]}(v)|$. \\
    }
    $\mathrm{ExpanderSketches} \leftarrow \mathrm{ExpanderSketches} \cup \left \{ \left(V_j^{(r)}, \mathcal{S}_v^{(r)}: v \in V_j^{(r)}, \mathrm{deg}_{G_{r}[V_{j}^{(r)}]}(v):v \in V_j^{(r)} \right): j \in [k_r] \right\}$. \\
    $\mathrm{Buffer} \leftarrow X$.\\
    $r \leftarrow r + 1$. 
    }
    }
    }
\Return${\mathrm{Deletions}, \mathrm{Buffer}, \mathrm{ExpanderSketches}}$.
\end{algorithm}

\subsubsection{Bounding the Space}

We now bound the maximum space required by \cref{alg:streamingBuildDO}.

\begin{claim}\label{clm:spaceBoundStreamDO}
    When \cref{alg:streamingBuildDO} is invoked on a stream of edges with at most $f \geq \sqrt{n}$ deletions on an $n$ vertex graph, the maximum space stored by the algorithm is $\widetilde{O}(n^{4/3}f^{1/3})$ many bits. 
\end{claim}

\begin{proof}
    First, we bound the final value of $r$ (i.e., the number of times that the buffer reaches capacity). For this, observe that every time the buffer reaches capacity (in \cref{line:bufferFull}), we perform the expander decomposition of \cref{thm:expanderDecomp} with $D = n^{1/3}f^{1/3}\log^2(n)$. In particular, this guarantees that the number of crossing edges is bounded by \[
    \frac{C \cdot n^{4/3}f^{1/3} \log^4(n)}{2} + C' \cdot n \cdot D \cdot \log^2(n) \leq \frac{C \cdot n^{4/3}f^{1/3} \log^4(n)}{2} + C' \cdot n^{4/3}f^{1/3} \log^4(n),
    \]
    where $C'$ is a sufficiently large constant to make the crossing edge bound hold in \cref{thm:expanderDecomp}. Now, by choosing $C$ such that $C = 4 \cdot C'$, we then can bound the number of crossing edges by 
    \[
    \frac{3C \cdot n^{4/3}f^{1/3} \log^4(n)}{4},
    \]
    which ensures that another $\frac{C \cdot n^{4/3}f^{1/3} \log^4(n)}{4}$ edge insertions are required before the buffer fills up once again. Thus, the number of times the buffer fills is bounded by $\widetilde{O}(\frac{n^2}{n^{4/3}f^{1/3}}) = \widetilde{O}(\frac{n^{2/3}}{f^{1/3}})$.

    Now, each time the buffer fills, we store three things: the sparse recovery sketches, the degrees of the vertices, and the identities of the expander components. To store the degrees of the vertices requires $O(\log(n))$ bits per vertex, and so $\widetilde{O}(n)$ bits across all vertices. Likewise, the expander components form a partition of the vertex set, and so require only $\widetilde{O}(n)$ bits to store. Lastly, for the sparse recovery sketches, we store $\mathcal{S}_v^{(r)} = g(\Gamma_{G_{r}[V_j^{(r)}]}(v))$ which is the sum of at most $n$ numbers, each of size at most $\mathrm{poly}(\binom{n}{2}^{4f/D})$. The total bit complexity of such a number is then bounded by 
    \[
    O\left(\log\left (\mathrm{poly}\left(\binom{n}{2}^{4f/D}\right) \right)\right) = \widetilde{O}\left(\frac{4f}{D}\right) = \widetilde{O}\left(\frac{f^{2/3}}{n^{1/3}}\right).
    \]
    Now, across all $n$ vertices, this requires space $\widetilde{O}(n^{2/3}f^{2/3})$. Thus, across all $\widetilde{O}(\frac{n^{2/3}}{f^{1/3}})$ refills of the buffer, the total space required by our sketches is bounded by $\widetilde{O}(n^{4/3}f^{1/3})$ many bits (provided $f \geq n^{1/2}$).

    Finally, we store the remaining crossing edges at the end of the stream. By construction, these are bounded by the size of our buffer, and so require only $\widetilde{O}(n^{4/3}f^{1/3})$ bits to store. This concludes the claim. 
\end{proof}

\subsubsection{Reporting Distances}

Now, we explain how to use the data stored by the streaming algorithm to report distances after the deletions. This follows the intuition of \cref{sec:buildDO} very closely. As we did there, we first use the following claim which allows us to know \emph{which expander} we delete an edge from:

\begin{claim}\label{clm:findingEdgeStream}
    Let $G = (V, E)$ be the graph resulting from a stream of edges, let $f$ be a parameter, and let $\mathrm{ExpanderSketches}$ be constructed as in \cref{alg:streamingBuildDO}. Then, an edge $e \in G_{j}[V_i^{(j)}]$ if and only if $V_i^{(j)}$ is the component with the \emph{smallest} value of $j$ such that $e \subseteq V_i^{(j)}$.
\end{claim}

Note that this claim is only true because of a small detail in our streaming algorithm. Namely, in \cref{line:checkFirst}, we check whether an edge $e_i$ which is being inserted can be placed into some \emph{already created} expander. If this is the case, then we place it in its first possible expander, thereby keeping the above invariant in place. With this stated, the formal recovery algorithm is as stated ion \cref{alg:reportDistanceStream}.

\begin{algorithm}[h!]
\caption{StreamReportDistance$(\mathrm{ExpanderSketches}, \mathrm{Deletions}, \mathrm{Buffer}, a, b)$}\label{alg:reportDistanceStream}
\For{$e = (u,v) \in \mathrm{Deletions}$}{
\If{exists component $V^{(j)}_i$ such that $e \subseteq V^{(j)}_i$}{
Let $V^{(j)}_i$ denote the first such component (i.e., for the smallest value of $j$). \\
$\mathrm{deg}_{G_{j}[V_i^{(j)}]}(v) \leftarrow \mathrm{deg}_{G_{j}[V_i^{(j)}]}(v) -1 $. \\
$\mathrm{deg}_{G_{j}[V_i^{(j)}]}(u) \leftarrow \mathrm{deg}_{G_{j}[V_i^{(j)}]}(u) -1 $. \\
$\mathcal{S}_v^{(j)} \leftarrow \mathcal{S}_v^{(j)} - g(e)$. \\
$\mathcal{S}_u^{(j)} \leftarrow \mathcal{S}_u^{(j)} - g(e)$. \\
}
\Else{
$\mathrm{Buffer} \leftarrow \mathrm{Buffer} - e$. \\
}
}
Let $H = (V, \emptyset)$. \\
\For{$V^{(j)}_i: j \in [\ell], i \in [k_j]$}{
Let $V'^{(j)}_i \subseteq V^{(j)}_i$ denote all vertices $u \in V^{(j)}_i$ such that $\mathrm{deg}_{G_{j}[V_i^{(j)}]}(u) \geq \frac{4f}{D} + 1$. \\
$H \leftarrow H + K_{V'^{(j)}_i}$. \label{line:addCliqueStream} \tcp{Add a clique on the high degree vertices.} 
\For{$u \in V^{(j)}_i - V'^{(j)}_i$}{
Let $\Gamma_{V^{(j)}_i}(u)$ denote the unique set of $\leq \frac{4f}{D}$ edges in $\binom{n}{2}$ such that $g(\Gamma_{V^{(j)}_i}(u)) = \mathcal{S}_u^{(j)}$. \\
$H \leftarrow H + \Gamma_{V^{(j)}_i}(u)$. \label{line:addFullNeighborhoodStream}
}
}
$H \leftarrow H + X_{\ell}$. \label{line:finishHStream}\\
Let $\delta$ denote the distance between $a, b$ in $H$, and let $C''$ denote a large constant.  \\
\Return{$\delta \cdot C'' \cdot \log(n)\log\log(n)$.}
\end{algorithm}

We include the correctness proofs for this recovery algorithm though they now follow exactly the outline established in \cref{sec:buildDO}.

\begin{claim}\label{clm:HcontainsStream}
    Let $G = (V, E)$ be a graph, let $f$ be a parameter, and let $\mathrm{ExpanderSketches}, \mathrm{Deletions}, \mathrm{Buffer}$ be the result of invoking \cref{alg:streamingBuildDO} on a stream that yields $G$ (in particular, $G$ here only includes the insertions). Now, suppose that we call \cref{alg:reportDistanceStream} with $\mathrm{ExpanderSketches}, \mathrm{Deletions}, \mathrm{Buffer}$ and vertices $a, b$. Then, it must be the case that $G - \mathrm{Deletions} \subseteq H$, where $H$ is the resulting auxiliary graph whose construction terminates in \cref{line:finishHStream}.
\end{claim}

\begin{proof}
    Let us consider any edge $e \in G - \mathrm{Deletions}$. By definition, any such edge must have also been in the initial graph $G$ on which we constructed our data structure. Now, there are several cases for us to consider:
    \begin{enumerate}
        \item Suppose $e \notin G_{j}[V_i^{(j)}]$ for any $i, j$. Then, it must be the case that $e \in \mathrm{Buffer}$ (i.e., the final set of crossing edges). But in \cref{line:finishHStream}, we add all edges in $\mathrm{Buffer} - \mathrm{Deletions}$ to $H$, and so if $e \notin \mathrm{Deletions}$ and $e \in \mathrm{Buffer}$, then $e \in H$. 
        \item Otherwise, this means that $e \in G_{j}[V_i^{(j)}]$ for some $i, j$. Again, we have two cases:
        \begin{enumerate}
            \item Suppose that after performing the deletions from $\mathrm{Deletions}$ it is the case that \[
            \mathrm{deg}_{G_{j}[V_i^{(j)}] - \mathrm{Deletions}}(u) \geq \frac{4f}{D} + 1\] and $\mathrm{deg}_{G_{j}[V_i^{(j)}] - \mathrm{Deletions}}(v) \geq \frac{4f}{D} + 1$. Then, this means that $u, v \in V'^{(j)}_i $, and so $e \in K_{V'^{(j)}_i}$. Because we add $K_{V'^{(j)}_i}$ to $H$, then $e \in H$. 
            \item Otherwise, at least one of $u,v $ (WLOG $u$) has $\mathrm{deg}_{G_{j}[V_i^{(j)}] - F}(u) \leq \frac{4f}{D}$. This means that $u$ has at most $\frac{4f}{D}$ neighbors in $G_{j}[V_i^{(j)}]$ after performing the deletions in $\mathrm{Deletions}$, and that in particular, $\mathcal{S}_u^{(j)}$ is the sum of $\leq \frac{4f}{D}$ many terms. By \cref{fact:uniqueSum}, there must then be a unique set of $\leq \frac{4f}{D}$ edges in $\binom{V}{2}$, denoted $\Gamma_{V^{(j)}_i}(u)$, such that $g(\Gamma_{V^{(j)}_i}(u)) = \mathcal{S}_u^{(j)}$. Because at this point \[
            \mathcal{S}_u^{(j)} = \sum_{e \in \Gamma_{G_{j}[V_i^{(j)}]}(u)} g(e) - \sum_{e \in \Gamma_{G_{j}[V_i^{(j)}]}(u) \cap \mathrm{Deletions}} g(e) = \sum_{e \in \Gamma_{G_{j}[V_i^{(j)}] - \mathrm{Deletions}}(u)} g(e),
            \]
            this means that this set of edges $\Gamma_{V^{(j)}_i}(u)$ that we recover are indeed all remaining incident edges on vertex $u$ in $G_{j}[V_i^{(j)}] - \mathrm{Deletions}$. Because $e \notin \mathrm{Deletions}$, it must also be the case then that $e \in \Gamma_{V^{(j)}_i}(u)$, and so is therefore also in $H$. 
        \end{enumerate}
    \end{enumerate}
    
    Thus, for any edge $e \in G - \mathrm{Deletions}$, we do indeed have that $e \in H$.
\end{proof}

As an immediate corollary to the above claim, we can also lower bound the true distance between any vertices in $G - F$:

\begin{corollary}\label{cor:distanceLowerBoundStream}
    Let $G, H, \mathrm{Deletions}$ be defined as in \cref{clm:HcontainsStream}. Then, for any vertices $a, b \in V$,  $\mathrm{dist}_{H}(a,b) \leq \mathrm{dist}_{G - \mathrm{Deletions}}(a,b)$.
\end{corollary}

\begin{proof}
    This trivially follows because any path in $G - \mathrm{Deletions}$ between $a, b$ will also be a path between $a, b$ in $H$. 
\end{proof}

More subtly, we also have the following upper bound on the distance between any vertices in $G - \mathrm{Deletions}$:

\begin{claim}\label{clm:distanceUpperBoundStream}
    Let $G, H, \mathrm{Deletions}$ be defined as in \cref{clm:HcontainsStream}. Then, for any vertices $a, b \in V$,  $\mathrm{dist}_{G - \mathrm{Deletions}}(a,b) \leq O(\log(n)\log\log(n)) \cdot \mathrm{dist}_{H}(a,b)$.
\end{claim}

\begin{proof}
    Let $a = v_0 \rightarrow v_1 \rightarrow v_2 \dots \rightarrow v_w \rightarrow b = v_{w+1}$ denote the shortest path between $a, b$ in $H$. Now consider any edge $v_{p}, v_{p+1}$ which is taken in this path. We claim that either: 
    \begin{enumerate}
        \item $(v_{p}, v_{p+1}) \in G - \mathrm{Deletions}$. 
        \item $v_{p}, v_{p+1}$ were vertices whose degree remained $\geq \frac{4f}{D} + 1$ in some expander component $G_{j}[V_i^{(j)}]$.
    \end{enumerate}

    This follows essentially by construction. There are three ways edges are added to $H$:
    \begin{enumerate}
        \item In \cref{line:finishHStream}, we add the final crossing edges $\mathrm{Buffer}$ which were not in $\mathrm{Deletions}$. Note that $\mathrm{Buffer} \subseteq G$ by construction, and so $\mathrm{Buffer} - \mathrm{Deletions} \subseteq G - \mathrm{Deletions}$.
        \item In \cref{line:addFullNeighborhoodStream}, whenever a vertex has degree $\leq \frac{4f}{D}$ in an expander component $G_{j}[V_i^{(j)}] - \mathrm{Deletions}$, we recover its entire neighborhood and add those edges to $H$. Importantly, these are still edges that are a subset of $G - \mathrm{Deletions}$.
        \item Otherwise, in \cref{line:addCliqueStream}, for vertices whose degree remains $\geq \frac{4f}{D} + 1$ in some expander component $G_{j}[V_i^{(j)}]$ after performing deletions $\mathrm{Deletions}$, we add a clique on these vertices.
    \end{enumerate}

    So, either an edge in $H$ is in $G - \mathrm{Deletions}$, or the edge was added in \cref{line:addCliqueStream}, and corresponds to the remaining high degree vertices in some expander component $G_{j}[V_i^{(j)}]$, as we claimed.

    Finally then, to conclude, we claim that for any edge $(v_{p}, v_{p+1})$ which are vertices whose degree remained $\geq \frac{4f}{D} + 1$ in some expander component $G_{j}[V_i^{(j)}]$, it must be the case that 
    \[
    \mathrm{dist}_{G-\mathrm{Deletions}}(v_p, v_{p+1}) \leq O(\log(n)\log\log(n)).
    \]
    This follows exactly via \cref{lem:robustnessfDeletion}: indeed, starting with the graph $G_{j}[V_i^{(j)}]$ which has minimum degree $D$ (with $D \geq 8 \sqrt{f}\log(n)$), we performed $\leq f$ deletions. \cref{lem:robustnessfDeletion} shows exactly that, under these conditions, any vertices whose degree remains $\geq \frac{4f}{D} + 1$ must be connected by a path of length $O(\log(n)\log\log(n))$.

    Thus, for the path $a = v_0 \rightarrow v_1 \rightarrow v_2 \dots \rightarrow v_w \rightarrow b = v_{w+1}$ in $H$, we know that every edge used in the path is either present in $G - \mathrm{Deletions}$, or can be replaced by a path of length $O(\log(n)\log\log(n))$ in $G - \mathrm{Deletions}$. Thus, $\mathrm{dist}_{G - \mathrm{Deletions}}(a,b) \leq O(\log(n)\log\log(n)) \cdot \mathrm{dist}_{H}(a,b)$.
\end{proof}

Now, we prove \cref{thm:faultTolerantDOStream}.

\begin{proof}[Proof of \cref{thm:faultTolerantDOStream}.]
First, if $f$ is $\leq \sqrt{n}$, then we simply use an $f$-fault-tolerant $O(\log(n)\log\log(n))$ spanner, which requires $\widetilde{O}(n f) \leq \widetilde{O}(n^{4/3} f^{1/3})$ bits of space, and can be constructed using a greedy algorithm in an insertion-only stream (see \cite{chechik2009fault,bodwin2022partially} for discussion of this greedy algorithm).

Otherwise, for other values of $f$, given a stream of edges whose insertions yield a graph $G$ and with deletions $\mathrm{Deletions}$, we instantiate the data structure using \cref{alg:streamingBuildDO}. By \cref{clm:spaceBoundStreamDO}, we know that this data structure requires $\widetilde{O}(n^{4/3} f^{1/3})$ bits of space to store. 
        
        After performing the deletions  $\mathrm{Deletions}$, and when queried with vertices $u, v$, using \cref{alg:reportDistanceStream}, we know that the algorithm constructs a graph $H$ (depending on $F$) such that (by \cref{cor:distanceLowerBoundStream} and \cref{clm:distanceUpperBoundStream})
        \[
        \mathrm{dist}_{H}(u,v) \leq \mathrm{dist}_{G - \mathrm{Deletions}}(u,v) \leq O(\log(n)\log\log(n)) \cdot \mathrm{dist}_{H}(u,v).
        \]
        In particular, if we let $C''$ denote a sufficiently large constant such that 
        \[
        \mathrm{dist}_{H}(u,v) \leq \mathrm{dist}_{G - \mathrm{Deletions}}(u,v) \leq C'' \log(n)\log\log(n) \cdot \mathrm{dist}_{H}(u,v),
        \]
        then outputting $\widehat{\mathrm{dist}}_{G - \mathrm{Deletions}}(u,v) = C'' \log(n)\log\log(n) \cdot \mathrm{dist}_{H}(u,v)$ satisfies 
        \[
        \mathrm{dist}_{G - \mathrm{Deletions}}(u,v) \leq \widehat{\mathrm{dist}}_{G - \mathrm{Deletions}}(u,v) \leq C'' \log(n)\log\log(n) \cdot \mathrm{dist}_{G - \mathrm{Deletions}}(u,v),
        \]
        thereby satisfying the conditions of our theorem. 
\end{proof}

\subsection{Fault-Tolerant Oblivious Spanners}

Finally, we show that in bounded-deletion streams, we can also construct \emph{oblivious spanners}, matching the same space bounds as our distance oracles from the previous section. Formally, we show the following:

\begin{theorem}\label{thm:obliviousSpannerStream}
    There is a streaming algorithm which, when given a parameter $f \in [0, \binom{n}{2}]$, and then invoked on a stream of edge insertions building a graph $G$, and at most $f$ oblivious deletions (denoted $\mathrm{Deletions} \subseteq G$):
    \begin{enumerate}
        \item Uses space at most $\widetilde{O}(n^{4/3} f^{1/3})$ bits. 
    \item Returns (with probability $1 - 1 / n^{100}$ over the randomness of the sketch) an $O(\log(n)\log\log(n))$ spanner of $G - \mathrm{Deletions}$. 
    \end{enumerate}
\end{theorem}

\subsubsection{Building the Spanner Sketch}

Naturally, the streaming algorithm is a combination of \cref{sec:streamDO} and \cref{sec:ObliviousSpanners}, see \cref{alg:streamingBuildSpanner}.

\begin{algorithm}[p]
    \caption{StreamingObliviousSpanner$(f, e_i, \Delta_i)$}\label{alg:streamingBuildSpanner}
    $\mathrm{Deletions} = \emptyset$. \\
    $\mathrm{Buffer} = \emptyset$. \\
    $\mathrm{ExpanderSketches} = []$. \\
    $r = 1$. \\
    Let $D = n^{1/3}f^{1/3}\log^2(n)$. \\
    \For{$(e_i = (u_i, v_i), \Delta_i)$ in the stream}{
    \If{$\Delta_i = -1$}{
    $\mathrm{Deletions} \leftarrow \mathrm{Deletions} \cup \{ e_i \}$. \\
    }
    \ElseIf{$\exists V_{j}^{(\ell)} \in \mathrm{ExpanderSketches}: e_i \subseteq V_{j}^{(\ell)}$}{
    Let $V_{j}^{(\ell)}$ be the first such vertex set in $\mathrm{ExpanderSketches}$ (i.e., smallest value of $\ell$). \\
    $\mathrm{deg}_{G_{\ell}[V_j^{(\ell)}]}(v) \leftarrow \mathrm{deg}_{G_{\ell}[V_j^{(\ell)}}(v) +1 $. \\
$\mathrm{deg}_{G_{\ell}[V_j^{(\ell)}}(u) \leftarrow \mathrm{deg}_{G_{\ell}[V_j^{(\ell)}}(u) +1 $. \\
$\mathcal{S}_{v, \ell_0}^{(\ell)} \leftarrow \mathcal{S}_{v, \ell_0}^{(\ell)} + \mathcal{S}_{v, \ell_0}^{(\ell)}(e)$. \\
$\mathcal{S}_{u, \ell_0}^{(\ell)} \leftarrow \mathcal{S}_{u, \ell_0}^{(\ell)} + \mathcal{S}_{u, \ell_0}^{(\ell)}(e)$. \\
    }
    \Else{
    $\mathrm{Buffer} \leftarrow \mathrm{Buffer} \cup \{e_i \}$. \\
    \If{$|\mathrm{Buffer}| \geq C \cdot n^{4/3}f^{1/3} \log^4(n)$}{
    Let $V^{(r)}_1, \dots V^{(r)}_{k_r}$ denote the high-degree expander decomposition of $\mathrm{Buffer}$ as guaranteed by \cref{thm:expanderDecomp} with $D = n^{1/3}f^{1/3}\log^2(n)$, and let $X$ denote the crossing edges. \\
    Let $G_r$ denote the expander edges. \\
    \For{$j \in [k_{r}]$, and each vertex $v \in V_j^{(r)}$}{ 
    Let $\Gamma_{G_{r}[V_j^{(r)}]}(v)$ denote the edges leaving $v$ in the expander $G_{r}[V_j^{(r)}]$. \\
    Let $\mathcal{S}_{v, \ell_0}^{(r)}$ be $\frac{4f}{D} +1$ many $\ell_0$-samplers (with independent randomness) over $\Gamma_{G_{r}[V_j^{(r)}]}(v)$. \\
    $\mathrm{deg}_{G_{r}[V_{j}^{(r)}]}(v) = |\Gamma_{G_{r}[V_j^{(r)}]}(v)|$.\\
    Let $\mathcal{S}^{(j, r)}_{\ell_0}(G_{r}[V_j^{(r)}])$ denote a set of $\frac{|G_{r}[V_j^{(r)}]| \log^5(n)}{D}$ many linear sketches of $\ell_0$ samplers (each with independent randomness) over the edge set of $G_{r}[V_j^{(r)}]$.
    }
    $\mathrm{ExpanderSketches} \leftarrow \mathrm{ExpanderSketches} \cup \left \{ \left(V_j^{(r)}, \mathcal{S}_v^{(r)}: v \in V_j^{(r)}, \mathrm{deg}_{G_{r}[V_{j}^{(r)}]}(v):v \in V_j^{(r)}, \mathcal{S}^{(j, r)}_{\ell_0} \right): j \in [k_r] \right\}$. \\
    $\mathrm{Buffer} \leftarrow X$.\\
    $r \leftarrow r + 1$. 
    }
    }
    }
\Return${\mathrm{Deletions}, \mathrm{Buffer}, \mathrm{ExpanderSketches}}$.
\end{algorithm}

\subsubsection{Bounding the Space}

We now bound the maximum space required by \cref{alg:streamingBuildSpanner}.

\begin{claim}\label{clm:spaceBoundStreamSpanner}
    When \cref{alg:streamingBuildSpanner} is invoked on a stream of edges with at most $f \geq \sqrt{n}$ deletions on an $n$ vertex graph, the maximum space stored by the algorithm is $\widetilde{O}(n^{4/3}f^{1/3})$ many bits. 
\end{claim}

\begin{proof}
    As in \cref{clm:spaceBoundStreamDO}, we know that the buffer fills $\widetilde{O}(n^{2/3} / f^{1/3})$ times before the end of the stream. So, it remains only to bound the additional space used by the sketch each time the buffer fills. 

    Indeed, whenever the buffer fills, the algorithm stores $O(f / D)$ $\ell_0$-samplers for each vertex (with failure probability $\delta = 1 / n^{1000}$) for a total space consumption of $O(nf/D) = \widetilde{O}(n^{2/3}f^{2/3})$ bits of space (using \cref{thm:ell0samp}). Likewise, the algorithm stores the degrees of all the vertices, which requires $\widetilde{O}(n)$ bits of space across all $n$ vertices. The algorithm also stores the identity of the expander components, which, because they form a partition, requires $\widetilde{O}(n)$ bits of space. Lastly, the algorithm stores a set of $\frac{|G_{r}[V_j^{(r)}]| \log^5(n)}{D}$ many linear sketches of $\ell_0$ samplers (each with independent randomness) over the edge set of $G_{r}[V_j^{(r)}]$. Because there are a total of $\widetilde{O}(n^{4/3}f^{1/3})$ edges, the total number of $\ell_0$-samplers we store here is bounded by 
    \[
    \frac{\widetilde{O}(n^{4/3}f^{1/3})}{D} = \widetilde{O}(n),
    \]
    and using the same $\ell_0$-sampler parameters, leads to a space bound of $\widetilde{O}(n)$ bits. Thus, when $f \geq \sqrt{n}$, the space used is dominated by the $\widetilde{O}(n^{2/3}f^{2/3}$ bits for the vertex $\ell_0$-samplers.

    Now, across all $\widetilde{O}(n^{2/3} / f^{1/3})$ buffer refills, the total space consumption is 
    \[
    \widetilde{O}(n^{2/3} / f^{1/3}) \cdot\widetilde{O}(n^{2/3}f^{2/3}) = \widetilde{O}(n^{4/3}f^{1/3})
    \]
    bits of space.

    Likewise, the space required to store the buffer is always $\widetilde{O}(n^{4/3}f^{1/3})$ bits of space, and so our total space bound is still $\widetilde{O}(n^{4/3}f^{1/3})$ bits. 
\end{proof}

\subsubsection{Proving Spanner Recovery}

First, we present the spanner recovery algorithm, see \cref{alg:recoverSpannerStream}. As in \cref{sec:streamDO}, we use \cref{clm:findingEdgeStream} to denote the rule for finding the \emph{first} expander that contains an edge $e$.

\begin{algorithm}[h!]
\caption{RecoverSpannerStream$(\mathrm{ExpanderSketches}, \mathrm{Buffer}, \mathrm{Deletions})$}\label{alg:recoverSpannerStream}
\For{$e = (u,v) \in \mathrm{Deletions}$}{
\If{exists component $V^{(j)}_i$ such that $e \subseteq V^{(j)}_i$}{
Let $V^{(j)}_i$ denote the first such component (i.e., smallest value of $j$, as per \cref{clm:findingEdgeStream}). \\
$\mathrm{deg}_{G_{j}[V_i^{(j)}]}(v) \leftarrow \mathrm{deg}_{G_{j}[V_i^{(j)}]}(v) -1 $. \\
$\mathrm{deg}_{G_{j}[V_i^{(j)}]}(u) \leftarrow \mathrm{deg}_{G_{j}[V_i^{(j)}]}(u) -1 $. \\
$\mathcal{S}_{v, \ell_0}^{(j)} \leftarrow \mathcal{S}_{v, \ell_0}^{(j)} - \mathcal{S}_{v, \ell_0}^{(j)}(e)$. \label{line:deleteNeighborhoodEdges1Stream}\\
$\mathcal{S}_{u, \ell_0}^{(j)} \leftarrow \mathcal{S}_{u, \ell_0}^{(j)} - \mathcal{S}_{u, \ell_0}^{(j)}(e)$.\label{line:deleteNeighborhoodEdges2Stream} \\
$\mathcal{S}^{(i, j)}_{\ell_0}(G_{j}[V_i^{(j)}]) \leftarrow \mathcal{S}^{(i, j)}_{\ell_0}(G_{j}[V_i^{(j)}]) - \mathcal{S}^{(i,j)}_{\ell_0}(e)$. \\
}
\Else{
$\mathrm{Buffer} \leftarrow \mathrm{Buffer}- e$. \label{line:deleteCrossingEdgesinFStream}\\
}
}
Let $H = (V, \emptyset)$. \\
\For{$V^{(j)}_i: j \in [r], i \in [k_{j}]$}{
Let $V'^{(j)}_i \subseteq V^{(j)}_i$ denote all vertices $u \in V^{(j)}_i$ such that $\mathrm{deg}_{G_{j}[V_i^{(j)}]}(u) \geq D/2$. \label{line:defineV'Stream}\\
\For{$u \in V^{(j)}_i - V'^{(j)}_i$}{
Open the $\ell_0$-samplers $\mathcal{S}_{v, \ell_0}^{(j)}$ to reveal edges $\Gamma_{V^{(j)}_i}(u)$. \\
$H \leftarrow H + \Gamma_{V^{(j)}_i}(u)$. \label{line:addFullNeighborhoodSpannerStream} \\
$\mathcal{S}^{(i,j)}_{\ell_0}(G_{j}[V_i^{(j)}]) \leftarrow \mathcal{S}^{(i,j)}_{\ell_0}(G_{j}[V_i^{(j)}]) - \mathcal{S}^{(i,j)}_{\ell_0}(\Gamma_{V^{(j)}_i}(u))$. \label{line:deleteEdgesell0sampStream}\\
}
Open the $\ell_0$-samplers $\mathcal{S}^{(i, j, d)}_{\ell_0}(G_{j}[V_i^{(j)}])$ to reveal a subgraph $\widetilde{G}_{j}[V_i^{(j)}]$. \\
$H \leftarrow H + \widetilde{G}_{j}[V_i^{(j)}] - \mathrm{Deletions}$.\label{line:removeDeletions} \\
\label{line:addSubsampleExpanderStream} \tcp{Add the subsampled expander.} 

}
$H \leftarrow H + \mathrm{Buffer}$. \label{line:finishHSpannerStream}\\
\Return{$H$.}
\end{algorithm}

Now we prove that the above algorithm does indeed return an $O(\log(n)\log\log(n))$ spanner of the post-deletion graph. 

We start by showing that $H$ is indeed a subgraph of $G - \mathrm{Deletions}$:

\begin{claim}\label{clm:containsHstream}
    Let $G = (V, E)$ be the graph resulting from all the insertions of the stream, let $f$ be a parameter, and let $\mathrm{ExpanderSketches}, \mathrm{Buffer}$ be the result of invoking \cref{alg:streamingBuildSpanner} on $G$, and let $\mathrm{Deletions}$ denote the set of $\leq f$ deletions. Now, suppose that we call \cref{alg:recoverSpannerStream} with $\mathrm{ExpanderSketches}, \mathrm{Buffer}$ and $\mathrm{Deletions}$. Then, the graph $H$ as returned by \cref{alg:recoverSpannerStream} is a subgraph of $G - \mathrm{Deletions}$. 
\end{claim}

\begin{proof}
    Let us consider any edge $e \in H$. We claim that any such edge $e$ must also be in $G - \mathrm{Deletions}$. Indeed, we can observe that there are only three ways for an edge to be added to $H$:
    \begin{enumerate}
        \item In \cref{line:finishHSpannerStream}, we add the final crossing edges $\mathrm{Buffer}$ which were not in $\mathrm{Deletions}$ (in \cref{line:deleteCrossingEdgesinFStream} we deleted all edges in $\mathrm{Buffer} \cap \mathrm{Deletions}$). Note that $\mathrm{Buffer} \subseteq G$ by construction, and so $\mathrm{Buffer} - \mathrm{Deletions} \subseteq G - \mathrm{Deletions}$.
        \item In \cref{line:addFullNeighborhoodSpannerStream}, whenever a vertex has degree $\leq D/2$ in an expander component $G_{j, D}[V_i^{(j, D)}] - \mathrm{Deletions}$, we open its $\ell_0$ samplers to recover edges in its neighborhood and add those edges to $H$. Importantly, these are still edges that are a subset of $G - \mathrm{Deletions}$, as in \cref{line:deleteNeighborhoodEdges1Stream} and \cref{line:deleteNeighborhoodEdges2Stream}, we subtracted $g(e)$ from $ \mathcal{S}_u^{(j, D)}$ whenever $e = (u,v)$ appeared in $\mathrm{Deletions}$. 
        \item Otherwise, in \cref{line:addSubsampleExpanderStream}, for vertices whose degree remains $\geq \frac{4f}{D} + 1$ in some expander component $G_{j, D}[V_i^{(j, D)}]$ after performing deletions $\mathrm{Deletions}$, we open our $\ell_0$-samplers for the remaining edges. Importantly, because we have deleted all edges in $\mathrm{Deletions}$ from these recovered edges (see \cref{line:removeDeletions}), by \cref{def:ell0samp}, the edges we recover must be in $G_{j, D}[V_i^{(j, D)}] - \mathrm{Deletions}$, and so these edges must be contained in $G- \mathrm{Deletions}$ as well.
    \end{enumerate}
\end{proof}

Next, we show that with high probability, $H$ is indeed an $O(\log(n)\log\log(n))$ spanner of $G$. To do this, we have the following subclaim:

\begin{claim}\label{clm:recoverExpanderStream}
    Let $G, \mathrm{Deletions}, H$ be as in \cref{clm:containsHstream}. Then, for an expander component $G_{j}[V_i^{(j)}]$, when opening the $\ell_0$ samplers $\mathcal{S}^{(i,j)}_{\ell_0}(G_{j}[V_i^{(j)}])$ to reveal a subgraph $\widetilde{G}_{j}[V_i^{(j)}]$, it is the case that $\widetilde{G}_{j}[V_i^{(j)}]$ contains an $\Omega(1 / \log(n))$-lopsided-expander over the vertices $V'^{(j)}_i$ (defined in \cref{line:defineV'Stream}) with probability $1 - 1 / n^{997}$.
\end{claim}

\begin{proof}
    Indeed, let us consider the subgraph $G_{j}[V'^{(j)}_i]$ (before sampling). By \cref{lem:robustnessfDeletion}, we know that $G_{j}[V'^{(j)}_i]$ is an $\Omega(1 / \log(n))$ lopsided-expander.

    Now, let us denote the number of edges in $G_{j}[V'^{(j)}_i]$ by $W_{\mathrm{sub}}$, and denote the number of edges in $G_{j}[V_i^{(j)}]$ by $W$. By the definition of $\ell_0$-samplers, whenever we open such a sampler (and it does not return $\perp$), we expect the sample to lie in $G_{j}[V'^{(j)}_i]$ with probability $\frac{W_{\mathrm{sub}}}{W}$ (see \cref{def:ell0samp}). Thus, when opening $\frac{|G_{j}[V_i^{(j)}]| \log^5(n)}{D}$ many $\ell_0$-samplers (and conditioning on no failures), we expect
    \[
    \geq \frac{|G_{j}[V_i^{(j)}]| \log^5(n)}{D}\cdot \frac{W_{\mathrm{sub}}}{W} \geq \frac{W_{\mathrm{sub}}\log^5(n)}{D}
    \]
    many samples to land in $G_{j}[V'^{(j)}_i]$. Next, we also recall that  the minimum degree in $G_{j}[V'^{(j)}_i]$ is still $\Omega(D)$, and that $G_{j}[V'^{(j)}_i]$ still has $\Omega(D^2)$ edges (see \cref{lem:robustnessfDeletion} for the proof of this). This implies that opening the $\ell_0$-samplers yields a random sample of 
    \[
    \Omega \left (  \frac{|G_{j}[V'^{(j)}_i]| \cdot \log^5(n)}{\mathrm{mindeg}(G_{j}[V'^{(j)}_i])} \right )
    \]
    many edges from $G_{j}[V'^{(j)}_i]$ in expectation. Now, observe that when we perform deletions in \cref{alg:streamingBuildSpanner}, \emph{we do not} propagate these deletions to the $\ell_0$-samplers of $G_{j}[V_i^{(j)}]$. Thus, it is possible that as many as $f$ edges that are recovered are in fact edges in $\mathrm{Deletions}$. When this is the case, we ignore those edges. Importantly, $f  \leq \frac{D^2}{\log(n)}$, and so in expectation, we still expect 
    \[
    \Omega \left (  \frac{|G_{j}[V'^{(j)}_i]| \cdot \log^5(n)}{\mathrm{mindeg}(G_{j}[V'^{(j)}_i])} \right )
    \]
    genuine edges from $G_{j}[V'^{(j)}_i]$. 
    
    Finally, using a Chernoff bound along with \cref{clm:subsampleExpander}, implies that probability $1 - 2^{- \Omega(\log^2(n))}$ the resulting subsample contains an $\Omega(1 / \log(n))$ expander over $G_{j}[V'^{(j)}_i]$. Taking a union over the aforementioned failure probabilities, along with the failure probabilities of the individual $\ell_0$-samplers ($1 / n^{1000}$) of which there are at most $n^2 \mathrm{polylog}(n)$, then yields the above claim. 
\end{proof}

Finally then, we can prove the bound on the stretch of our spanner:

\begin{lemma}\label{lem:lowStretchStream}
    Let $G, \mathrm{Deletions}, H$ be as in \cref{clm:containsHstream}. Then, $H$ is an $O(\log(n)\log\log(n))$ spanner of $G - \mathrm{Deletions}$ with probability $1 - 1 / n^{100}$.
\end{lemma}

\begin{proof}
    Consider any edge $(u,v) \in G - \mathrm{Deletions}$. We claim that the distance between $u$ and $v$ in $H$ will be bounded by $O(\log(n)\log\log(n))$ (with high probability), and thus $H$ is an $O(\log(n)\log\log(n))$ spanner.

    There are a few cases for us to consider:

    \begin{enumerate}
        \item Suppose there is no choice of $i, j$ such that $e = (u,v) \subseteq V_i^{(j)}$. This means that $e$ must be in $\mathrm{Buffer} - \mathrm{Deletions}$. But, the set $\mathrm{Buffer}$ is explicitly stored ($\mathrm{Deletions}$ is removed) and thus $e \in H$. 
        \item Instead, suppose there is a choice of $i, j$ such that $e = (u,v) \subseteq V_i^{(j)}$, and let $V_i^{(j)}$ be the first such set (as per \cref{clm:findingEdgeStream}). Again, we have several cases:
        \begin{enumerate}
            \item If one of $u, v$ (WLOG $u$) has degree $\leq \frac{4f}{D}$, then when opening the $\ell_0$-samplers $\mathcal{S}_{u, \ell_0}^{(j)}$, with probability $1 - 1 / n^{999}$ (taking a union bound over the failure probabilities of the $\ell_0$-samplers) it is the case that all edges in the neighborhood of $u$ (in this expander component) are recovered. Thus, $e = (u,v)$ is recovered and is therefore in $H$. 
            \item Otherwise, suppose that both of $u, v$ have degree $\geq \frac{4f}{D} +1$. Then, as per \cref{lem:robustnessfDeletion}, we know that there are at most $\frac{4f}{D}$ vertices in $V_{i}^{(j)}$ which are \emph{not in} $V_{i}'^{(j)}$. So, $u, v$ both recover with probability $1 - 1 / n^{999}$ (taking a union bound over the failure probabilities of the $\ell_0$-samplers), an edge to some vertices $u', v'$ in $V_{i}'^{(j)}$. Now, by \cref{clm:recoverExpanderStream}, we know that we have also recovered an $\Omega(1 / \log(n))$ lopsided-expander over $V_{i}'^{(j)}$, and so $u', v'$ are at distance $O(\log(n)\log\log(n))$ in $H$. Thus, $u, v$ are at distance $O(\log(n)\log\log(n)) + 2$, thereby yielding the claim. 
        \end{enumerate}
    \end{enumerate}
\end{proof}

Finally, we prove \cref{thm:obliviousSpannerStream}.

\begin{proof}[Proof of \cref{thm:obliviousSpannerStream}.]
First, if $f$ is $\leq \sqrt{n}$, then we simply use an $f$-fault-tolerant $O(\log(n)\log\log(n))$ spanner, which requires $\widetilde{O}(n f) \leq \widetilde{O}(n^{4/3} f^{1/3})$ bits of space, and can be constructed using a greedy algorithm in an insertion-only stream (see \cite{chechik2009fault,bodwin2022partially} for discussion of this greedy algorithm).

Otherwise, for other values of $f$, the size of the sketch follows from \cref{lem:spaceBoundSpanner}. The spanner property follows from \cref{clm:containsH} and \cref{lem:lowStretch}.
\end{proof}
    
\bibliographystyle{alpha}
\bibliography{ref}

@inproceedings{GGKKS22,
  author       = {Anupam Gupta and
                  Vijaykrishna Gurunathan and
                  Ravishankar Krishnaswamy and
                  Amit Kumar and
                  Sahil Singla},
  editor       = {Joseph (Seffi) Naor and
                  Niv Buchbinder},
  title        = {Online Discrepancy with Recourse for Vectors and Graphs},
  booktitle    = {Proceedings of the 2022 {ACM-SIAM} Symposium on Discrete Algorithms,
                  {SODA} 2022, Virtual Conference / Alexandria, VA, USA, January 9 -
                  12, 2022},
  pages        = {1356--1383},
  publisher    = {{SIAM}},
  year         = {2022},
  url          = {https://doi.org/10.1137/1.9781611977073.57},
  doi          = {10.1137/1.9781611977073.57},
  bibsource    = {dblp computer science bibliography, https://dblp.org}
}

@inproceedings{GKKS20,
  author       = {Anupam Gupta and
                  Ravishankar Krishnaswamy and
                  Amit Kumar and
                  Sahil Singla},
  editor       = {Nitin Saxena and
                  Sunil Simon},
  title        = {Online Carpooling Using Expander Decompositions},
  booktitle    = {40th {IARCS} Annual Conference on Foundations of Software Technology
                  and Theoretical Computer Science, {FSTTCS} 2020, December 14-18, 2020,
                  {BITS} Pilani, {K} {K} Birla Goa Campus, Goa, India (Virtual Conference)},
  series       = {LIPIcs},
  volume       = {182},
  pages        = {23:1--23:14},
  publisher    = {Schloss Dagstuhl - Leibniz-Zentrum f{\"{u}}r Informatik},
  year         = {2020},
  url          = {https://doi.org/10.4230/LIPIcs.FSTTCS.2020.23},
  doi          = {10.4230/LIPICS.FSTTCS.2020.23},
  bibsource    = {dblp computer science bibliography, https://dblp.org}
}

@inproceedings{peleg1987optimal,
  title={An optimal synchronizer for the hypercube},
  author={Peleg, David and Ullman, Jeffrey D},
  booktitle={Proceedings of the sixth annual ACM Symposium on Principles of distributed computing},
  pages={77--85},
  year={1987}
}

@inproceedings{dark2020optimal,
  author       = {Jacques Dark and
                  Christian Konrad},
  editor       = {Shubhangi Saraf},
  title        = {Optimal Lower Bounds for Matching and Vertex Cover in Dynamic Graph
                  Streams},
  booktitle    = {35th Computational Complexity Conference, {CCC} 2020, July 28-31,
                  2020, Saarbr{\"{u}}cken, Germany (Virtual Conference)},
  series       = {LIPIcs},
  volume       = {169},
  pages        = {30:1--30:14},
  publisher    = {Schloss Dagstuhl - Leibniz-Zentrum f{\"{u}}r Informatik},
  year         = {2020},
  url          = {https://doi.org/10.4230/LIPIcs.CCC.2020.30},
  doi          = {10.4230/LIPICS.CCC.2020.30},
  bibsource    = {dblp computer science bibliography, https://dblp.org}
}

@article{peleg1989trade,
  title={A trade-off between space and efficiency for routing tables},
  author={Peleg, David and Upfal, Eli},
  journal={Journal of the ACM (JACM)},
  volume={36},
  number={3},
  pages={510--530},
  year={1989},
  publisher={ACM New York, NY, USA}
}

@article{awerbuch1992efficient,
  title={Efficient broadcast and light-weighted spanners},
  author={Awerbuch, Baruch},
  journal={manuscript},
  year={1992}
}

@book{peleg2000distributed,
  title={Distributed computing: a locality-sensitive approach},
  author={Peleg, David},
  year={2000},
  publisher={SIAM}
}

@article{althofer1993sparse,
  title={On sparse spanners of weighted graphs},
  author={Alth{\"o}fer, Ingo and Das, Gautam and Dobkin, David and Joseph, Deborah and Soares, Jos{\'e}},
  journal={Discrete \& Computational Geometry},
  volume={9},
  number={1},
  pages={81--100},
  year={1993},
  publisher={Springer}
}

@inproceedings{levcopoulos1998efficient,
  title={Efficient algorithms for constructing fault-tolerant geometric spanners},
  author={Levcopoulos, Christos and Narasimhan, Giri and Smid, Michiel},
  booktitle={Proceedings of the thirtieth annual ACM symposium on Theory of computing},
  pages={186--195},
  year={1998}
}

@inproceedings{bodwin2021optimal,
  title={Optimal vertex fault-tolerant spanners in polynomial time},
  author={Bodwin, Greg and Dinitz, Michael and Robelle, Caleb},
  booktitle={Proceedings of the 2021 ACM-SIAM Symposium on Discrete Algorithms (SODA)},
  pages={2924--2938},
  year={2021},
  organization={SIAM}
}

@inproceedings{KLS93,
  author       = {Sanjeev Khanna and
                  Nathan Linial and
                  Shmuel Safra},
  title        = {On the Hardness of Approximating the Chromatic Number},
  booktitle    = {Second Israel Symposium on Theory of Computing Systems, {ISTCS} 1993,
                  Natanya, Israel, June 7-9, 1993, Proceedings},
  pages        = {250--260},
  publisher    = {{IEEE} Computer Society},
  year         = {1993},
  url          = {https://doi.org/10.1109/ISTCS.1993.253464},
  doi          = {10.1109/ISTCS.1993.253464},
  bibsource    = {dblp computer science bibliography, https://dblp.org}
}

@article{bringmann2025near,
  title={Near-optimal directed low-diameter decompositions},
  author={Bringmann, Karl and Fischer, Nick and Haeupler, Bernhard and Latypov, Rustam},
  journal={arXiv preprint arXiv:2502.05687},
  year={2025}
}

@inproceedings{popova2026new,
  title={New Greedy Spanners and Applications},
  author={Popova, Elizaveta and Tzalik, Elad},
  booktitle={17th Innovations in Theoretical Computer Science Conference (ITCS 2026)},
  pages={107--1},
  year={2026},
  organization={Schloss Dagstuhl--Leibniz-Zentrum f{\"u}r Informatik}
}

@inproceedings{parter2022nearly,
  title={Nearly optimal vertex fault-tolerant spanners in optimal time: sequential, distributed, and parallel},
  author={Parter, Merav},
  booktitle={Proceedings of the 54th Annual ACM SIGACT Symposium on Theory of Computing},
  pages={1080--1092},
  year={2022}
}

@inproceedings{BHP24,
  author       = {Greg Bodwin and
                  Bernhard Haeupler and
                  Merav Parter},
  editor       = {David P. Woodruff},
  title        = {Fault-Tolerant Spanners against Bounded-Degree Edge Failures: Linearly
                  More Faults, Almost For Free},
  booktitle    = {Proceedings of the 2024 {ACM-SIAM} Symposium on Discrete Algorithms,
                  {SODA} 2024, Alexandria, VA, USA, January 7-10, 2024},
  pages        = {2609--2642},
  publisher    = {{SIAM}},
  year         = {2024},
  url          = {https://doi.org/10.1137/1.9781611977912.93},
  doi          = {10.1137/1.9781611977912.93},
  bibsource    = {dblp computer science bibliography, https://dblp.org}
}

@inproceedings{bodwin2018optimal,
  title={Optimal vertex fault tolerant spanners (for fixed stretch)},
  author={Bodwin, Greg and Dinitz, Michael and Parter, Merav and Williams, Virginia Vassilevska},
  booktitle={Proceedings of the Twenty-Ninth Annual ACM-SIAM Symposium on Discrete Algorithms},
  pages={1884--1900},
  year={2018},
  organization={SIAM}
}

@article{bodwin2018optimalCorrection,
  title={Private Communication},
  author={Bodwin, Greg and Dinitz, Michael and Parter, Merav and Williams, Virginia Vassilevska},
  year={2025},
}

@article{guruswami2012essential,
  title={Essential coding theory},
  author={Guruswami, Venkatesan and Rudra, Atri and Sudan, Madhu},
  journal={Draft available at http://www. cse. buffalo. edu/atri/courses/coding-theory/book},
  volume={2},
  number={1},
  year={2012}
}

@inproceedings{khanna2025streaming,
  author       = {Sanjeev Khanna and
                  Christian Konrad and
                  Jacques Dark},
  editor       = {Keren Censor{-}Hillel and
                  Fabrizio Grandoni and
                  Jo{\"{e}}l Ouaknine and
                  Gabriele Puppis},
  title        = {Streaming Maximal Matching with Bounded Deletions},
  booktitle    = {52nd International Colloquium on Automata, Languages, and Programming,
                  {ICALP} 2025, July 8-11, 2025, Aarhus, Denmark},
  series       = {LIPIcs},
  volume       = {334},
  pages        = {106:1--106:20},
  publisher    = {Schloss Dagstuhl - Leibniz-Zentrum f{\"{u}}r Informatik},
  year         = {2025},
  url          = {https://doi.org/10.4230/LIPIcs.ICALP.2025.106},
  doi          = {10.4230/LIPICS.ICALP.2025.106},
  bibsource    = {dblp computer science bibliography, https://dblp.org}
}

@article{khanna2025near,
  title={Near-optimal Hypergraph Sparsification in Insertion-only and Bounded-deletion Streams},
  author={Khanna, Sanjeev and Putterman, Aaron and Sudan, Madhu},
  journal={arXiv preprint arXiv:2504.16321},
  year={2025}
}

@inproceedings{jayaram2018data,
  title={Data streams with bounded deletions},
  author={Jayaram, Rajesh and Woodruff, David P},
  booktitle={Proceedings of the 37th ACM SIGMOD-SIGACT-SIGAI Symposium on Principles of Database Systems},
  pages={341--354},
  year={2018}
}

@article{leighton1999multicommodity,
  title={Multicommodity max-flow min-cut theorems and their use in designing approximation algorithms},
  author={Leighton, Tom and Rao, Satish},
  journal={Journal of the ACM (JACM)},
  volume={46},
  number={6},
  pages={787--832},
  year={1999},
  publisher={ACM New York, NY, USA}
}

@inproceedings{jowhari2011tight,
  title={Tight bounds for lp samplers, finding duplicates in streams, and related problems},
  author={Jowhari, Hossein and Sa{\u{g}}lam, Mert and Tardos, G{\'a}bor},
  booktitle={Proceedings of the thirtieth ACM SIGMOD-SIGACT-SIGART symposium on Principles of database systems},
  pages={49--58},
  year={2011}
}

@inproceedings{filtser2021graph,
  title={Graph spanners by sketching in dynamic streams and the simultaneous communication model},
  author={Filtser, Arnold and Kapralov, Michael and Nouri, Navid},
  booktitle={Proceedings of the 2021 ACM-SIAM Symposium on Discrete Algorithms (SODA)},
  pages={1894--1913},
  year={2021},
  organization={SIAM}
}

@inproceedings{bodwin2019trivial,
  title={A trivial yet optimal solution to vertex fault tolerant spanners},
  author={Bodwin, Greg and Patel, Shyamal},
  booktitle={Proceedings of the 2019 ACM Symposium on Principles of Distributed Computing},
  pages={541--543},
  year={2019}
}

@article{chebyshev1852memoire,
  title={M{\'e}moire sur les nombres premiers},
  author={Chebyshev, Pafnuty Lvovich},
  journal={J. Math. Pures Appl.},
  volume={17},
  pages={366--390},
  year={1852}
}

@inproceedings{bodwin2022partially,
  title={Partially optimal edge fault-tolerant spanners},
  author={Bodwin, Greg and Dinitz, Michael and Robelle, Caleb},
  booktitle={Proceedings of the 2022 Annual ACM-SIAM Symposium on Discrete Algorithms (SODA)},
  pages={3272--3286},
  year={2022},
  organization={SIAM}
}

@inproceedings{dinitz2020efficient,
  title={Efficient and simple algorithms for fault-tolerant spanners},
  author={Dinitz, Michael and Robelle, Caleb},
  booktitle={Proceedings of the 39th Symposium on Principles of Distributed Computing},
  pages={493--500},
  year={2020}
}

@inproceedings{chechik2009fault,
  title={Fault-tolerant spanners for general graphs},
  author={Chechik, Shiri and Langberg, Michael and Peleg, David and Roditty, Liam},
  booktitle={Proceedings of the forty-first annual ACM symposium on Theory of computing},
  pages={435--444},
  year={2009}
}

@inproceedings{awerbuch1990network,
  title={Network synchronization with polylogarithmic overhead},
  author={Awerbuch, Baruch and Peleg, David},
  booktitle={Proceedings [1990] 31st Annual Symposium on Foundations of Computer Science},
  pages={514--522},
  year={1990},
  organization={IEEE}
}

@article{peleg1989graph,
  title={Graph spanners},
  author={Peleg, David and Sch{\"a}ffer, Alejandro A},
  journal={Journal of graph theory},
  volume={13},
  number={1},
  pages={99--116},
  year={1989},
  publisher={Wiley Online Library}
}

@inproceedings{BBGNSS22,
  author       = {Aaron Bernstein and
                  Jan van den Brand and
                  Maximilian Probst Gutenberg and
                  Danupon Nanongkai and
                  Thatchaphol Saranurak and
                  Aaron Sidford and
                  He Sun},
  editor       = {Mikolaj Bojanczyk and
                  Emanuela Merelli and
                  David P. Woodruff},
  title        = {Fully-Dynamic Graph Sparsifiers Against an Adaptive Adversary},
  booktitle    = {49th International Colloquium on Automata, Languages, and Programming,
                  {ICALP} 2022, July 4-8, 2022, Paris, France},
  series       = {LIPIcs},
  volume       = {229},
  pages        = {20:1--20:20},
  publisher    = {Schloss Dagstuhl - Leibniz-Zentrum f{\"{u}}r Informatik},
  year         = {2022},
  url          = {https://doi.org/10.4230/LIPIcs.ICALP.2022.20},
  doi          = {10.4230/LIPICS.ICALP.2022.20},
  bibsource    = {dblp computer science bibliography, https://dblp.org}
}

@article{SW18,
  author       = {Thatchaphol Saranurak and
                  Di Wang},
  title        = {Expander Decomposition and Pruning: Faster, Stronger, and Simpler},
  journal      = {CoRR},
  volume       = {abs/1812.08958},
  year         = {2018},
  url          = {http://arxiv.org/abs/1812.08958},
  eprinttype    = {arXiv},
  eprint       = {1812.08958},
  bibsource    = {dblp computer science bibliography, https://dblp.org}
}

@inproceedings{CKL22,
  author       = {Yu Chen and
                  Sanjeev Khanna and
                  Huan Li},
  title        = {On Weighted Graph Sparsification by Linear Sketching},
  booktitle    = {63rd {IEEE} Annual Symposium on Foundations of Computer Science, {FOCS}
                  2022, Denver, CO, USA, October 31 - November 3, 2022},
  pages        = {474--485},
  publisher    = {{IEEE}},
  year         = {2022},
  url          = {https://doi.org/10.1109/FOCS54457.2022.00052},
  doi          = {10.1109/FOCS54457.2022.00052},
  bibsource    = {dblp computer science bibliography, https://dblp.org}
}

@article{KPS24d,
  title={Near-optimal Size Linear Sketches for Hypergraph Cut Sparsifiers},
  booktitle={To appear, FOCS 2024},
  author={Khanna, Sanjeev and Putterman, Aaron L and Sudan, Madhu},
  journal={arXiv preprint arXiv:2407.03934},
  year={2024}
}

@inproceedings{AGM12,
  author       = {Kook Jin Ahn and
                  Sudipto Guha and
                  Andrew McGregor},
  editor       = {Yuval Rabani},
  title        = {Analyzing graph structure via linear measurements},
  booktitle    = {Proceedings of the Twenty-Third Annual {ACM-SIAM} Symposium on Discrete
                  Algorithms, {SODA} 2012, Kyoto, Japan, January 17-19, 2012},
  pages        = {459--467},
  publisher    = {{SIAM}},
  year         = {2012},
  url          = {https://doi.org/10.1137/1.9781611973099.40},
  doi          = {10.1137/1.9781611973099.40},
  bibsource    = {dblp computer science bibliography, https://dblp.org}
}

@inproceedings{AGM12b,
  author       = {Kook Jin Ahn and
                  Sudipto Guha and
                  Andrew McGregor},
  editor       = {Michael Benedikt and
                  Markus Kr{\"{o}}tzsch and
                  Maurizio Lenzerini},
  title        = {Graph sketches: sparsification, spanners, and subgraphs},
  booktitle    = {Proceedings of the 31st {ACM} {SIGMOD-SIGACT-SIGART} Symposium on
                  Principles of Database Systems, {PODS} 2012, Scottsdale, AZ, USA,
                  May 20-24, 2012},
  pages        = {5--14},
  publisher    = {{ACM}},
  year         = {2012},
  url          = {https://doi.org/10.1145/2213556.2213560},
  doi          = {10.1145/2213556.2213560},
  bibsource    = {dblp computer science bibliography, https://dblp.org}
}

@inproceedings{GMT15,
  author       = {Sudipto Guha and
                  Andrew McGregor and
                  David Tench},
  editor       = {Tova Milo and
                  Diego Calvanese},
  title        = {Vertex and Hyperedge Connectivity in Dynamic Graph Streams},
  booktitle    = {Proceedings of the 34th {ACM} Symposium on Principles of Database
                  Systems, {PODS} 2015, Melbourne, Victoria, Australia, May 31 - June
                  4, 2015},
  pages        = {241--247},
  publisher    = {{ACM}},
  year         = {2015},
  url          = {https://doi.org/10.1145/2745754.2745763},
  doi          = {10.1145/2745754.2745763},
  bibsource    = {dblp computer science bibliography, https://dblp.org}
}

@article{CF14,
  author       = {Graham Cormode and
                  Donatella Firmani},
  title        = {A unifying framework for  l0-sampling algorithms},
  journal      = {Distributed Parallel Databases},
  volume       = {32},
  number       = {3},
  pages        = {315--335},
  year         = {2014},
  url          = {https://doi.org/10.1007/s10619-013-7131-9},
  doi          = {10.1007/S10619-013-7131-9},
  bibsource    = {dblp computer science bibliography, https://dblp.org}
}

@inproceedings{ADKKP16,
  author       = {Ittai Abraham and
                  David Durfee and
                  Ioannis Koutis and
                  Sebastian Krinninger and
                  Richard Peng},
  editor       = {Irit Dinur},
  title        = {On Fully Dynamic Graph Sparsifiers},
  booktitle    = {{IEEE} 57th Annual Symposium on Foundations of Computer Science, {FOCS}
                  2016, 9-11 October 2016, Hyatt Regency, New Brunswick, New Jersey,
                  {USA}},
  pages        = {335--344},
  publisher    = {{IEEE} Computer Society},
  year         = {2016},
  url          = {https://doi.org/10.1109/FOCS.2016.44},
  doi          = {10.1109/FOCS.2016.44},
  bibsource    = {dblp computer science bibliography, https://dblp.org}
}

\appendix

\section{Proof of \cref{clm:efficientUniqueSum}}\label{sec:efficientUniqueSum}

First, we restate \cref{clm:efficientUniqueSum} for the convenience of the reader (similar claims have been used in i.e., \cite{KPS24d}, Appendix A):

\begin{claim}
    For a universe $[u]$ and a constant $k$, let $q$ be a prime of size $[u, 2u]$. There is a polynomial time computable linear function $g: \zo^{[u]} \rightarrow \F_q^{2k+1}$ such that:
    \begin{enumerate}
        \item For any $S \subseteq [u]$, $|S| \leq k$, $g(\zo^S)$ is unique. 
        \item For any vector $z \in \F_q^{2k+1}$ such that there exists $S \subseteq [u]$, $|S| \leq k$ and $g(\zo^S) = z$, there is a polynomial time algorithm for recovering $S$ given $z$.
    \end{enumerate}
\end{claim}

\begin{proof}
    First, to see that there is a prime in $[u, 2u]$ we use Bertrand's postulate \cite{chebyshev1852memoire}. Finding such a prime can then be done by brute force in time $\mathrm{poly}(u)$.

    Next, we construct the function $g$. To do this, we fix a Reed-Solomon code $C$ of distance $d = 2k+1$ and block-length $u$ over field $q$. Because Reed-Solomon codes are maximum distance separable (MDS) (see \cite{guruswami2012essential}), we know that $\mathrm{dim}(C) = u - (2k+1)+1 = u - 2k$. Now, we let $H \in \F_q^{u \times (u - (2k))}$ denote a generating matrix of $C$, i.e., such that $C = \mathrm{Im}(H)$.

    Finally, we let $G$ denote the parity check matrix of $H$, i.e., $\mathrm{rank}(G) = 2k$, $G \in \F_q^{ (2k) \times u}$ and $G H = 0$. Equivalently stated, a vector $v \in \mathrm{Im}(G)$ if and only if $v H = 0$. Now, for a set $S \subseteq [u]$, the function $g$ computes $g(\zo^S) = G \cdot \zo^S$, where $\zo^S$ is understood to be the $\zo$-valued indicator vector of $S$ in $[u]$. Thus, clearly $g$ is a linear function, and $g$ is polynomial time computable. 

    Now, it remains to show that the above conditions are satisfied. To see condition 1, observe that $C$ has distance $\geq 2k+1$, which means that any codeword $c \in C$ has $\mathrm{wt}(c) \geq 2k+1$. Now, a codeword $c \in C$ if and only if $G \cdot c = 0$, so the above implies that  $G \cdot c = 0$ if and only if $\wt(c) = 0$ or $\wt(c) \geq 2k+1$. Let us suppose for the sake of contradiction that there exist two sets $S, S', |S| \leq k, |S'| \leq k$ such that $g(\zo^S ) = g(\zo^{S'})$. This implies that $g(\zo^{S \oplus S'}) = 0$, but $1 \leq \wt(\zo^{S \oplus S'}) \leq 2k$, which thus contradicts the aforementioned property. 

    Finally, to see the efficient decoding property, we recall (see for instance \cite{guruswami2012essential}, Section 12) that syndrome decoding of Reed-Solomon codes is known to be polynomial time computable up to distance $d/2 > k$. In our setting, the syndrome is given by $z = G \cdot \zo^S$, and syndrome decoding refers to finding the \emph{smallest Hamming weight vector} $v$ such that $z = Gv$, which, when $|S| \leq k$, must be $v = \zo^S$. Thus, there is a polynomial time algorithm for recovering $S$.
\end{proof}

\section{On Relaxing the ``Edge-Subset'' Condition for Fault-Tolerant Distance Oracles}\label{sec:edgeSubsetCondition}

Recall that, in defining a distance-oracle (see \cref{def:distanceOracle}), when we instantiate such a data structure on a graph $G = (V, E)$, we only require that the data structure returns the correct answer for queries $F, a, b$ when $F \subseteq E$ (and for $|F| \leq f$). Naturally, one can ask whether the complexity of such a data structure changes when we instead allow for \emph{arbitrary} deletion sets $F \subseteq \binom{V}{2}$, $|F| \leq f$, with the understanding that distances should be reported for the graph $G - F\cap G$ (i.e., one projects the deletion set back onto the original graph). 

It turns out that in this regime, there is a clear $\Omega(nf)$ lower bound:

\begin{claim}
    For integers $f, t$, let $\mathcal{S}$ be a data structure, such that, when instantiated on a graph $G$, $\mathcal{S}(G)$ supports queries on triples $F, a, b$, $|F| \leq f, F \subseteq \binom{V}{2}$, and returns a value $\widehat{\mathrm{dist}}_{G-F}(a, b)$ such that
    \[
    \mathrm{dist}_{G - F \cap G}(a, b) \leq \widehat{\mathrm{dist}}_{G-F}(a, b) \leq t \cdot \mathrm{dist}_{G - F \cap G}(a, b).
    \]
    Then, there exist graphs $G$ for which $\mathcal{S}(G)$ requires $\Omega(nf)$ bits to store.
\end{claim}

\begin{proof}[Proof sketch.]
    For the vertex set $V$ on $n$ vertices, let us break $V$ into $\frac{n}{f}$ groups of $f$ vertices, which we denote by $V_1, \dots V_{n/f}$. On each one of these vertex sets, we define an arbitrary graph $G_i = (V_i, E_i)$, and let $G = G_1 \cup G_2 \cup \dots \cup G_{\frac{n}{f}}$.

    We claim that $\mathcal{S}(G)$ can be used to recover the exact identity of any $G$ of this form. Because there are $\left (2^{\binom{f}{2}} \right )^{n/f} = 2^{n(f-1)/2}$ many graphs of this form, it follows that $\mathcal{S}(G)$ requires $\Omega(nf)$ bits to represent.

    To see why such reconstruction is possible, let us focus our attention on $V_1$, and let $u \in V_1$. Now, to decide if there is an edge from $u \in V_1$ to another vertex $v \in V_1$, we simply perform the query with $u, v, F = V_1 - \{u,v\}$. If there is an edge from $u$ to $v$, then the distance from $u$ to $v$ is finite, and this will be reflected by the query. Otherwise, if there is no edge from $u$ to $v$, there will now also not be any edges from $u$ to $V_1 - \{u,v\}$, and so $u$ is an isolated vertex. Thus, the reported distances will be $\infty$. This enables us to determine the presence of any edge $(u,v)$ in the graph $G$, and thus recover the entire graph $G$. 
\end{proof}

In fact, notice that in the above argument, the queries we perform are not even adaptive. Indeed, there is a simple set of $f \cdot (f-1) \cdot \frac{n}{f}$ many deterministic queries that can be done independently of one another to determine the entire graph $G$. This means that, even in the \emph{oblivious} sketch model, there is no non-trivial compression that can be achieved.

\section{One-Way Communication Protocols and Deriving Lower Bounds}\label{sec:oneWayCommunication}

Both our data structures immediately yield efficient communication protocols in the {\em one-way two-party communication setting}. In this setting, Alice holds the edges $E$ of a graph $G=(V, E)$ and Bob holds at most $f$ edge deletions $F \subseteq E$. Alice sends a message to Bob who then outputs the result of the computation. Protocols have access to both private and public randomness, and the objective is to solve a problem with Alice communicating as few bits as possible. This setting has previously been considered, for example, in \cite{dark2020optimal, khanna2025streaming}.

Since, in this setting, Alice can compute both our data structures from \cref{thm:introDO} and \cref{thm:IntroObliviousSpanner} on her edge set $E$, send the resulting data structure to Bob who then applies the deletions $F$, we obtain a deterministic protocol for $f$-fault-tolerant $O(\log(n)\log\log(n))$-distance oracles, and a randomized communication protocol for the construction of a $O(\log(n)\log\log(n))$-spanner when Bob's edge deletions are oblivious. Both protocols have communication cost $\tilde{O}(n \sqrt{f})$ bits.

We also believe that improving the communication cost to $o(n \sqrt{f})$ bits in either protocol may be hard to achieve. Indeed, let $R(n,f)$ denote the communication complexity, i.e., the minimum amount of communication required, on $n$-vertex graphs when at most $f$ deletions are allowed. Then, a direct sum argument using a graph construction similar to the one used in \cite{khanna2025streaming} allows us to establish the following inequality:
$$R(n, f) =  \Omega(n / \sqrt{f} \cdot R(\sqrt{f}, f)) \ . $$
This inequality is justified as follows. Let $\Pi$ be a protocol designed for $n$-vertex graphs that allows for at most $f$ deletions. We consider an input graph $G=(V, E)$ consisting of $n/\sqrt{f}$ components $V_1, V_2, \dots, V_{n/\sqrt{f}}$, each of size $\sqrt{f}$, so that, for every $i$, the edges of each component $V_i$ are chosen independently of the other components, and $G[V_i]$ constitutes a hard input graph on $\sqrt{f}$ vertices when $f$ deletions are allowed. Then, since Bob can add the $f$ edge deletions to any of the components $V_1, \dots, V_{n/\sqrt{f}}$, the message sent in protocol $\Pi$ must contain sufficient information so that Bob can solve each of the components in the presence of up to $f$ deletions. We note that since all $n / \sqrt{f}$ components are independent, a clever compression of information is not possible. Hence, the protocol $\Pi$ must be able to solve $n / \sqrt{f}$ independent instances, each on $\sqrt{f}$ vertices with up to $f$ edge deletions per instance.

This idea can be formalized by, first, arguing that the information complexity of a problem in the one-way communication setting is within constant factors of its communication complexity, which follows from standard message compression arguments. Then, the idea outlined above can implemented via a traditional direct sum argument for information complexity.

To justify that $R(n, f)$ might be as large as $\Omega(n \cdot \sqrt{f})$, we note that, currently, no non-trivial protocol is known if Bob's deletions are {\em arbitrary}, i.e., $F$ is an arbitrary subset of the edge set $E$. Indeed, for a graph on $n$ vertices, with an arbitrary number of deletions, the best known upper bound for constructing oblivious $n^{\frac{2}{3}(1 - \alpha)}$-spanners in this setting is due to the work of Filtser, Kapralov, and Nouri \cite{filtser2021graph}, who created such sketches in space $\widetilde{O}(n^{1+ \alpha})$. Unfortunately, when we desire stretch $\mathrm{polylog}(n)$, this forces the size to be $n^{2 - o(1)}$. While there is no exact matching lower bound, the work of Chen, Khanna, and Li \cite{CKL22} showed that for sketches which only store \emph{vertex incidence information} and for $\alpha \in [0,1/10]$, it is indeed the best possible (up to $n^{o(1)}$ factors in the space). Because vertex incidence sketches make up the \emph{vast} majority of linear sketches used for graph problems (\cite{AGM12, AGM12b, GMT15} to name a few) we view $n^{2 - o(1)}$ as a potential lower bound for general linear sketches of stretch $\mathrm{polylog}(n)$.

Hence, if $R(n, n^2)$ was indeed as large as $n^{2 - o(1)}$, then we would obtain $R(n, f) = \Omega \left( (n / \sqrt{f}) \cdot f^{1 - o(1)} \right) = \Omega(n f^{1/2 - o(1)})$, thereby nearly matching the upper bounds achieved by our protocols.

\end{document}